\documentclass[11pt]{article}
\pdfoutput=1 
\usepackage{amssymb}
\usepackage{amsmath}
\usepackage{amstext}
\usepackage{mathrsfs}
\usepackage{multirow}
\usepackage{graphicx}
\usepackage{epsfig}
\usepackage{verbatim} 
\usepackage{fancybox}
\usepackage{color}
\usepackage{mathdots}
\usepackage{anyfontsize}
\usepackage{ulem}
\usepackage{enumitem}
\usepackage{caption}
\usepackage{amsthm}
\usepackage{subfig}
\usepackage{mathtools}
\usepackage{youngtab}
\usepackage{braket}
\usepackage{bbm}
\usepackage{parskip}
\usepackage[numbers,sort&compress]{natbib}
\usepackage{ytableau}
\usepackage[utf8]{inputenc}
\usepackage{tikz}
\usetikzlibrary{arrows.meta}
\usepackage{simplewick}
\usepackage{overpic}
\usetikzlibrary{positioning, trees,decorations, decorations.markings, decorations.pathmorphing, arrows, shapes.geometric}

\newcommand{\Comment}[1]{{}}
\definecolor{darkblue}{rgb}{0.15,0.35,0.55}
\definecolor{reddish}{rgb}{0.65, 0.2, 0.2}
\usepackage[linktocpage=true]{hyperref}
\hypersetup{
colorlinks=true,
citecolor=darkblue,
linkcolor=reddish,
urlcolor=darkblue,
pdfauthor={},
pdftitle={},
pdfsubject={}
}

\usepackage{cleveref}

\newcommand{\pvec}[1]{\vec{#1}\mkern2mu\vphantom{#1}}

\newcommand{\be}{\begin{equation}}
\newcommand{\ee}{\end{equation}}

\renewcommand{\vec}[1]{\boldsymbol{#1}}
\renewcommand{\pvec}[1]{\vec{#1}}
\newcommand{\cutoff}{\mathsf{\Lambda}}
\newcommand{\D}{\mathcal{D}}

\numberwithin{equation}{section}

\makeatletter
\def\thickhline{%
  \noalign{\ifnum0=`}\fi\hrule \@height \thickarrayrulewidth \futurelet
   \reserved@a\@xthickhline}
\def\@xthickhline{\ifx\reserved@a\thickhline
               \vskip\doublerulesep
               \vskip-\thickarrayrulewidth
             \fi
      \ifnum0=`{\fi}}
\makeatother

\newlength{\thickarrayrulewidth}
\newtheorem{theorem}{Theorem}[section] 

\newtheorem{proposition}[theorem]{Proposition} 
\newtheorem{lemma}[theorem]{Lemma}
\theoremstyle{definition}

\theoremstyle{remark}

\tikzset{
  mgraviton/.style={decorate, draw=black,
    decoration={amplitude=3.5pt, segment length=5pt}},
    electron/.style={draw=blue, postaction={decorate},
        decoration={markings,mark=at position .55 with {\arrow[draw=blue]{>}}}}
}

\begin{document}

\renewcommand{\thefootnote}{\fnsymbol{footnote}}
~
\vspace{1.75truecm}
\thispagestyle{empty}
\begin{center}
{\LARGE \bf{
Bootstrapping Euclidean Lattices
}}
\end{center} 

\vspace{1cm}
\centerline{\begingroup
\renewcommand{\thefootnote}{$\Lambda$}
\Large
	Francesco Bertucci\footnote{\href{mailto:fbertucc@olemiss.edu}{\texttt{fbertucc@olemiss.edu}}}
\endgroup
\begingroup
\renewcommand{\thefootnote}{$\Gamma$}
\Large
	and James Bonifacio\footnote{\href{mailto:bonifacio@phy.olemiss.edu}{\texttt{bonifacio@phy.olemiss.edu}}}
\endgroup}
\vspace{.5cm}

\centerline{{\it Department of Physics and Astronomy,}}
 \centerline{{\it   University of Mississippi, University, MS 38677}} 
 \vspace{.25cm}
 
\vspace{1cm}
\begin{abstract}
\noindent
We derive spectral identities involving the Laplace spectrum and the integrals of products of eigenfunctions on flat tori and orbifolds. Using semidefinite programming, we derive upper bounds on sums of squares of triple products from these spectral identities, following the approach of the conformal bootstrap. Physically, these upper bounds give constraints on sums of squares of cubic coupling constants in toroidal compactifications of higher-dimensional field theories. 
For Euclidean lattices, one of the bootstrap bounds leads to an upper bound on the mean square number of minimal vectors that are minimal distance from each minimal vector, divided by the kissing number of the lattice. 
By finding exact functionals for the semidefinite programming problems, we prove that this bound is saturated in 2, 4, 8, and 24 dimensions by the hexagonal lattice, the $D_4$ lattice, the $E_8$ lattice, and the Leech lattice, respectively. 
\end{abstract}

\newpage
\setcounter{tocdepth}{2}
\tableofcontents
\renewcommand*{\thefootnote}{\arabic{footnote}}
\setcounter{footnote}{0}

\newpage 
\section{Introduction}

Conformal field theories (CFTs), hyperbolic manifolds, and Euclidean lattices have several deep connections. One of the key properties associated with these objects is their spectra: the spectrum of scaling dimensions of conformal primary operators of a CFT, the spectrum of eigenvalues of the Laplace operator on a hyperbolic manifold, and the spectrum of vector norms in a Euclidean lattice. In each case, determining the largest allowed value of the spectral gap is an important and well-studied problem that is amenable to linear and semidefinite programming techniques.

For CFTs, a powerful approach for constraining the spectrum is the conformal bootstrap. The idea of the conformal bootstrap is to use the operator product expansion (OPE) to expand four-point correlators in terms of conformal blocks in multiple ways. Equating these expansions gives consistency conditions that constrain the allowed values of the scaling dimensions and OPE coefficients. After being successfully applied to solve certain 2D CFTs in the 80s \cite{Belavin:1984vu}, this idea was more recently revitalized to find constraints on CFTs in dimensions $d>2$ by using linear and semidefinite programming to find numerical bounds on scaling dimensions and OPE coefficients  \cite{Rattazzi:2008pe, Rychkov:2009ij, Caracciolo:2009bx, Poland:2011ey}---see Ref.~\cite{Poland:2018epd} for a review. Perhaps the most successful application of the numerical conformal bootstrap has been to the study of the strongly coupled 3D Ising CFT, which led to the most accurate determination of its critical exponents \cite{ElShowk:2012ht,El-Showk:2014dwa,Kos:2014bka,Kos:2016ysd,Chang:2024whx}.

It has recently been shown that the approach of the conformal bootstrap can be adapted to study hyperbolic manifolds and orbifolds \cite{Bonifacio:2020xoc, Bonifacio:2021msa,Bonifacio:2021aqf,Kravchuk:2021akc, Radcliffe:2024jcg, Bonifacio:2023ban,Gesteau:2023brw, Adve:2025rvf, Adve:2025sld}. Applying spectral expansions to write integrals of products of functions in multiple ways leads to spectral identities or consistency conditions that can be used to derive bounds on the spectrum of Laplace operators and the integrals of products of eigenfunctions. This bootstrap approach has been applied in general dimensions to Einstein manifolds \cite{Bonifacio:2020xoc} and hyperbolic manifolds \cite{Bonifacio:2021msa}, to bound the spectral gap of the Laplace operator \cite{Bonifacio:2021aqf,Kravchuk:2021akc, Radcliffe:2024jcg} and the Dirac operator \cite{Gesteau:2023brw} on closed hyperbolic surfaces, and to bound eigenvalues of the curl operator on closed hyperbolic 3-manifolds \cite{Bonifacio:2023ban}. Other recent developments include using the spectral identities to derive a Weyl bound for triple product L-functions \cite{Adve:2025rvf} and a proof by Adve that any solution to the spectral identities in two dimensions with a discrete spectrum must come from a compact hyperbolic surface \cite{Adve:2025sld}.

Another powerful approach for constraining the Laplace spectra of hyperbolic manifolds and orbifolds is through the Selberg trace formula, which relates weighted sums of Laplace eigenvalues to weighted sums of lengths of closed geodesics.  For example, the trace formula has been a crucial tool in recent breakthroughs related to the spectral gap of typical hyperbolic surfaces of large volume, as reviewed in \cite{monk2026spectral}. The trace formula can also be combined with linear programming methods to bound the spectral gap and other quantities on hyperbolic manifolds \cite{bourque2019kissing, Bourque:2023woe, Bonifacio:2023ban}. 

The application of linear programming methods to the Selberg trace formula is inspired by the linear programming bounds for sphere packing densities by Cohn and Elkies \cite{cohn2003new}, which itself builds on Delsarte's linear programming bounds for codes \cite{MR314545}. The sphere packing bounds utilize the Poisson summation formula, which is the Euclidean analog of the Selberg trace formula. We show in Fig.~\ref{fig:bestpackings} a comparison between these linear programming bounds and the densest known sphere packings, for both lattice sphere packings and general sphere packings. As suggested by Fig.~\ref{fig:bestpackings}, the linear programming bounds are numerically extremely close to being saturated in 8 and 24 dimensions by lattice packings corresponding to the $E_8$ lattice and the Leech lattice, respectively.  It is an old result, due to Blichfeldt, that $E_8$ gives the densest lattice packing in 8 dimensions \cite{Blichfeldt1935}. The proof that the Leech lattice gives the densest lattice packing in 24 dimensions was given more recently by Cohn and Kumar and employed linear programming methods \cite{MR2600869}. The proof that these lattice packings are optimal amongst \textit{general} sphere packings was given by Viazovska for the $E_8$ lattice \cite{viazovska2017sphere} and by Cohn, Kumar, Miller, Radchenko, and Viazovska for the Leech lattice \cite{cohn2017sphere}. The proofs of \cite{viazovska2017sphere, cohn2017sphere} involve the construction of particular ``magic functions" that were conjectured in \cite{cohn2003new}: these functions, together with their Fourier transforms, have prescribed zeros at the vector lengths of the optimal lattices.
\begin{figure}[ht!]
	\begin{center}
		\epsfig{file=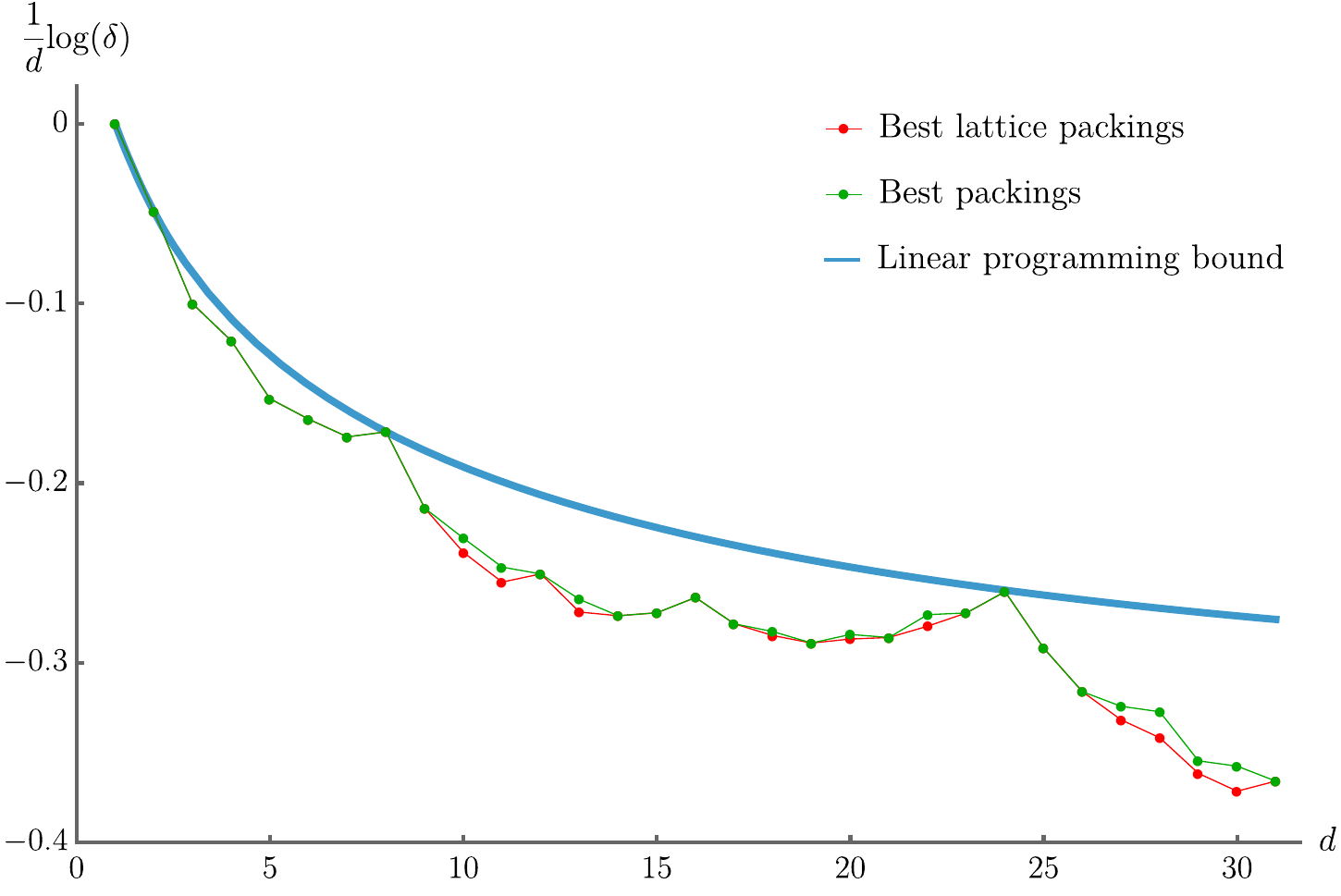,width=12cm}
	\end{center}
	\caption{Density of lattice and general sphere packings up to 31 dimensions, compared to the linear programming bounds by Cohn and Elkies \cite{cohn2003new}. The packing data is taken from Table 1.1(a) and Table 1.1(b) in Ref.~\cite{Conway:1988oqe}. The blue curve is obtained by interpolating the points in Table 3 of Ref.~\cite{cohn2003new}. The center density $\delta$ is related to the sphere packing density $\Delta$ through $\delta=\Delta/V_d$, where $V_d$ is the volume of the $d$-dimensional unit ball, $V_d=\pi^{d/2}/\Gamma (d/2+1)$.}
	\label{fig:bestpackings}
\end{figure}

There is a concrete relation between the conformal bootstrap and the sphere packing problem. In particular, the spinless modular bootstrap of 2D CFTs with central charge $c$ is related to the sphere packing problem for $2c$-dimensional Euclidean lattices \cite{Hartman:2019pcd} (see also \cite{Afkhami-Jeddi:2020hde}). In the spinless modular bootstrap, the invariance of the partition function under inversions naturally defines a linear optimization problem to put an upper bound on the scaling dimension of the first non-trivial primary operator. As shown in \cite{Hartman:2019pcd}, there is a precise relation between this linear optimization problem and the one by Cohn and Elkies \cite{cohn2003new}. In particular, for $c=4$ and $c=12$ the modular bootstrap extremal functionals found by Maz\'{a}\v{c} \cite{Mazac:2016qev} reproduce the magic functions of \cite{viazovska2017sphere, cohn2017sphere}.

A more direct connection between CFTs and Euclidean lattices is that lattices can be used to define certain classes of 2D CFTs known as lattice (or code) CFTs. These constructions are due to the work of Dolan, Goddard, and Montague \cite{Dolan:1989vr,Dolan:1989kf,Dolan:1994st}. Given a lattice in $d$ dimensions, this procedure gives a chiral CFT with central charge $d$. Alternatively, we can think of this as $d$ free bosons compactified on a $d$-dimensional lattice, where we only retain the holomorphic sector of the theory. Orbifold (or twisted) versions of these constructions also exist, as in the famous Monster CFT \cite{Frenkel:1988xz}. 
These constructions produce chiral free theories, with the scaling dimensions and OPE coefficients completely determined by the lattice. Some recent studies considered non-chiral CFTs that arise from Lorentzian lattices (or quantum codes) \cite{Dymarsky:2020qom}, as well as the construction of code CFTs on surfaces of higher genus \cite{Henriksson:2021qkt}.

In this paper, we explore another connection between CFTs and Euclidean lattices by applying conformal bootstrap methods to flat manifolds and orbifolds. Our approach can be viewed as the limit of the hyperbolic bootstrap in which the curvature goes to zero. In this way, it is analogous to studying flat space $S$-matrices as limits of AdS correlators.  Flat orbifolds are much simpler than hyperbolic orbifolds: many of their properties can be calculated analytically and, using physics terminology, their four-point correlators involve the exchange of a finite number of states. They are like integrable theories, whereas the high-energy eigenmodes of hyperbolic manifolds have imprints of chaos. Flat orbifolds thus provide a simple class of examples that are amenable to the bootstrap while still being nontrivial.  Any flat $d$-orbifold can be written as $\mathbb{R}^d / \Gamma$, where $\Gamma$ is a crystallographic group. This includes flat tori $\mathbb{R}^d / \Lambda$ as a special case, where $\Lambda$ is a Euclidean lattice. Our bootstrap bounds can thus be interpreted in terms of crystallographic groups and Euclidean lattices, which play important roles in various areas such as chemistry, number theory, and coding theory. 

The density of a sphere packing defined by a $d$-dimensional lattice $\Lambda$ is determined by the length of the shortest nonzero vectors in $\Lambda$, which are called minimal vectors. This shortest length also determines the spectral gap of the flat torus $\mathbb{R}^d/\Lambda^*$. It would be natural to search for bootstrap bounds on the spectral gap of tori with unit volume, which would bound the Hermite constant. However, we are only able to probe ratios of eigenvalues with our spectral identities, due to the absence of a curvature scale. We instead focus on triple products of eigenfunctions. Let $\phi_i$ for $i = 0, 1, 2, \dots$ be orthonormal eigenfunctions of the non-negative Laplace operator on a closed connected orientable Riemannian manifold with eigenvalues $\lambda_i$ in non-decreasing order, and let $\sigma_i$ denote the distinct eigenvalues,
$
0=\sigma_0 < \sigma_1 < \sigma_2 < \dots \rightarrow \infty.
$
 We define the following integrals of products of three eigenfunctions, which we refer to as triple products or triple overlap integrals:
\be
c_{ijk} \coloneqq \int \phi_i \phi_j \phi_k \, dV.
\ee
These appear as coefficients in the spectral decomposition of the product of two eigenfunctions,
$
\phi_i \phi_j = \sum_{k=0}^{\infty} c_{ijk} \phi_k,
$
where the sum has finitely many non-vanishing terms on a flat manifold.
For an incomplete selection of additional papers studying the triple products $c_{ijk}$, see \cite{Sarnak94, bernstein2010subconvexity, LU20193271, WYMAN2022109404}. 

Let us define
\be \label{eq:c_ij_def}
c_{ij}(\sigma_k) \coloneqq \sqrt{V  \sum_l c^2_{ijl} },
\ee
where the sum runs over the eigenspace with eigenvalue $\sigma_k$ and $V$ is the volume of the manifold. The factor of $\sqrt{V}$ makes this quantity invariant under constant rescalings of the metric. For a flat orbifold, we can define $c_{ij}(\sigma_k)$ by restricting to the invariant eigenmodes on its torus cover. 
 The main quantities that we bound for flat manifolds and orbifolds are $c_{11}(\sigma_1)$ and $c_{11}(\sigma_2)$. 
On a flat torus $\mathbb{R}^d/\Lambda$, the maximum value of $c_{11}(\sigma_1)$ is bounded from below by $\nu_1(\Lambda^*)/\sqrt{ N_1(\Lambda^*)}$, where $\nu_1(\Lambda^*)$ (see \eqref{eq:RMS_def})  is the root-mean-square number of minimal vectors that are minimal distance from each minimal vector in the dual lattice $\Lambda^*$ and $N_1(\Lambda^*)$ is the number of minimal vectors, i.e., the kissing number, of $\Lambda^*$. In terms of the sphere packing defined by $\Lambda^*$, $\nu_1(\Lambda^*)$ is the root mean square of the number of common neighbors of the central sphere and each sphere touching it.

As an example, we show in Fig.~\ref{fig:c11_plot} numerical upper bounds on $c_{11}(\sigma_1)$ for flat $d$-tori with $d=2, \dots, 31$, together with example tori defined by laminated lattices. 
We show rigorously that the bounds are saturated in dimensions $d=2,4,8,$ and $24$ by the tori $\mathbb{R}^d/\Lambda_d^*$, with $\Lambda_d$ equal to the $A_2$, $D_4$, $E_8$, and Leech lattices, by finding exact functionals for the semidefinite optimization problems. This gives the following
\begin{theorem} \label{thm:main_theorem}
For all full-rank positive definite lattices $\Lambda$ in dimensions 2, 4, 8, and 24, the quantity $\nu^2_1(\Lambda)/ N_1(\Lambda)$ is less than or equal to the value of this quantity attained by the $A_2$, $D_4$, $E_8$, and Leech lattices, respectively. Explicitly, these values are $\frac{2^2}{6}$, $\frac{8^2}{24}$, $\frac{56^2}{240}$, and $\frac{4600^2}{196560}$, respectively.
\end{theorem}
\begin{figure}[h!t]
\begin{center}
	\epsfig{file=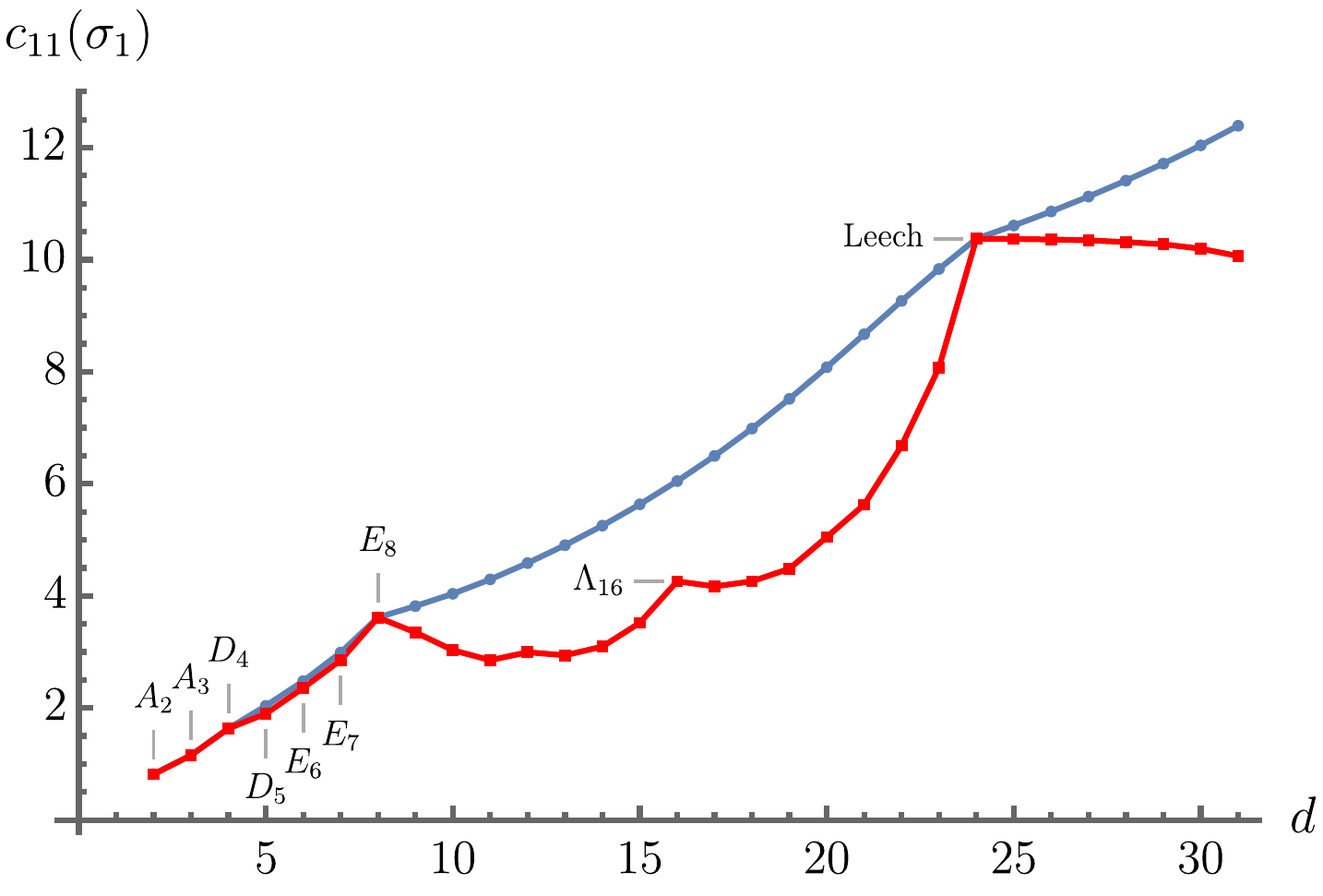, width=11.65cm}
\end{center}
\caption{The blue dots show the bootstrap upper bounds on $c_{11}(\sigma_1)$ for flat tori. The red dots show $\nu_1(\Lambda_d)/\sqrt{ N_1(\Lambda_d)}$, where $\Lambda_d$ are $d$-dimensional laminated lattices.
}
\label{fig:c11_plot}
\end{figure}

The Laplace eigenvalues and triple products $c_{ijk}$ admit a physical interpretation in the context of Kaluza--Klein theories. Suppose we start with a theory on a higher-dimensional product spacetime $\mathbb{R}^{1,3} \times \mathcal{M}$, where $\mathcal{M}$ is a closed manifold. Dimensional reduction provides an equivalent formulation of the theory in the lower-dimensional Minkowski spacetime where various towers of massive degrees of freedom appear. The masses of the new states are determined by the Laplace spectrum and their interactions are given by overlaps of eigenmodes. For example, in the dimensional reduction of general relativity, the Kaluza--Klein excitations of the graviton are spin-2 particles with squared masses $m^2_i=\lambda_i$ and cubic interactions proportional to $\sqrt{V} c_{ijk}/M_p$, where $M_p$ is the 4D Planck mass \cite{Hinterbichler:2013kwa, Bonifacio:2019ioc}. We can thus schematically represent $c^2_{ij}(\sigma_k)$ by the sum of exchange diagrams shown below, although the actual contribution to the four-point amplitude is more complicated \cite{Bonifacio:2019ioc, SekharChivukula:2019qih}.

\begin{figure}[!ht]
\begin{center}
\begin{tikzpicture}[thick, node distance=1.1cm and 1.1cm]
\coordinate (svertex1);

\coordinate[above left=of svertex1, label=left:$i$] (s1);
\coordinate[below left=of svertex1, label=left:$j$] (s2);
\coordinate[right=1.5cm of svertex1] (svertex2);
\coordinate[above right=of svertex2, label=right:$i$] (s3);
\coordinate[below right=of svertex2, label=right:$j$] (s4);

\draw[mgraviton] (s1) -- (svertex1);
\draw[mgraviton] (svertex1) -- (s2);
\draw[mgraviton] (s3) -- (svertex2);
\draw[mgraviton] (svertex2) -- (s4);
\draw[mgraviton] (svertex1) -- node[above=2pt] {$l$} (svertex2);

\node[left=2cm of svertex1] {$\displaystyle \sum_{m_l^2=\sigma_k}$};

\end{tikzpicture}
\end{center}
\end{figure}

The outline of the rest of this paper is as follows: we review aspects of Euclidean lattices and flat orbifolds in Section \ref{sec:lattices}. In Section \ref{sec:sumrules}, we derive expressions for the integrals of products of eigenfunctions that are needed to derive the spectral identities. The spectral identities and bootstrap bounds are then presented in Section \ref{sec:bounds} and we conclude in Section \ref{sec:conclusion}.

\section{Euclidean lattices} \label{sec:lattices}

In this section, we give a brief review of Euclidean lattices. For a comprehensive treatment, see the classic book by Conway and Sloane \cite{Conway:1988oqe}. An extensive online catalog of lattices is provided by Nebe and Sloane \cite{LatticeCatalog}.  

Two important motivations for studying lattices are the sphere packing problem and the kissing number problem. The sphere packing problem asks how densely we can pack identical non-overlapping spheres in $d$ dimensions. If the centers of the spheres form a lattice, then the packing is called a lattice packing. 
The kissing number problem asks for the greatest number of non-overlapping unit spheres that can touch a common unit sphere.

\subsection{Review of lattices}

Let $\{\vec{v}_i\}_{i=1}^d$ be a basis for $\mathbb{R}^d$. This basis generates a $d$-dimensional Euclidean lattice,
\be\label{eq:fundparallelotope}
	\Lambda = \left\{ \sum_{i=1}^d \xi_i \vec{v}_i :  \xi_i \in \mathbb{Z} \right\}.
\ee
The fundamental parallelotope of $\Lambda$  consists of the points
\be
	\theta_1 \vec{v}_1 + \hdots + \theta_d \vec{v}_d, \qquad \text{with} \qquad 0 \leqslant \theta_i < 1.
\ee
The generator matrix of $\Lambda$  is the matrix $A$ whose rows are the vectors $\vec{v}_i$.
The lattice thus consists of all vectors of the form $\vec{\xi} A$ with $\vec{\xi} \in \mathbb{Z}^d$ and the volume of the fundamental parallelotope is $|\! \det A|$.
The Gram matrix for $\Lambda$ is defined as
\be
	G =A  A^T,
\ee
where $A^T$ denotes the transpose of $A$. 
The Gram matrix is a symmetric matrix with components $G_{ij} =\vec{v}_i \cdot \vec{v}_j$, where $\cdot$ denotes the Euclidean inner product. 
The dual lattice of $\Lambda$  is 
\be
	\Lambda^* =\left\{ \vec{k} \in\mathbb{R}^d : \vec{k}\cdot \vec{x} \in\mathbb{Z} \text{ for all } \vec{x} \in\Lambda \right\}.
\ee
The generator matrix of $\Lambda^*$ is $(A^{-1})^T$ and the Gram matrix of $\Lambda^*$ is $G^{-1}$. 
Lattices $\Lambda$ and $\Lambda'$ with generator matrices $A$ and $A'$ are said to be equivalent, written $\Lambda \cong \Lambda'$,  if
\be
	A' =c\, U A B,
\ee
where $c$ is a nonzero constant, $U$ is a matrix with integer entries and $\mathrm{det}\,U =\pm 1$, and $B$ is a real orthogonal matrix. 

The norm of a lattice vector $\vec{x} \in\Lambda$ is defined as its length squared, 
$
 \vec{x} \cdot \vec{x}= |\vec{x}|^2 .
$
If we write $\vec{x}=\vec{\xi}  A$, then
\be
|\vec{x}|^2 =\sum_{i,j=1}^d \xi_i \xi_j \vec{v}_i \cdot \vec{v}_j =\vec{\xi} G \vec{\xi}^T.
\ee
The function $q_G(\vec{\xi}) \coloneqq \vec{\xi} G \vec{\xi}^T$ is the quadratic form associated with $\Lambda$.
The distinct norms of the lattice $\Lambda$ are denoted by $n_i(\Lambda)$, $i\in \mathbb{Z}_{\geq 0},$ with 
\be
n_0(\Lambda) =0 < n_1(\Lambda) < n_2(\Lambda) < \dots.
\ee 
We denote the $k$\textsuperscript{th} shell centered at the lattice point $\vec{x} \in \Lambda$ by
\be \label{eq:kth_shell_at_x}
S_k(\vec{x}, \Lambda) \coloneqq \{ \vec{y} \in \Lambda : |\vec{x} - \vec{y}|^2 = n_k(\Lambda) \}.
\ee
For the $k$\textsuperscript{th} shell centered at the origin, we write 
\be \label{eq:kth_shell}
S_k( \Lambda) \coloneqq S_k(\vec{0}, \Lambda).
\ee
The size of the $k$\textsuperscript{th} shell is denoted by $N_k(\Lambda)$,
\be
N_k(\Lambda) \coloneqq \left| \left\{ \vec{x} \in \Lambda \, : \,  \vec{x} \cdot \vec{x} = n_k(\Lambda) \right\}  \right| = |S_k( \Lambda)|.
\ee
The number of vectors with the smallest nonzero norm, $N_1(\Lambda)$, is called the kissing number of the lattice. 

A lattice is called integral if $\vec{x}\cdot \vec{y} \in\mathbb{Z}$ for all $\vec{x},\vec{y}\in\Lambda$. If $|\vec{x}|^2 $ is even for all $\vec{x}\in\Lambda$, then the lattice is called even. An integral lattice such that $\Lambda = \Lambda^*$ or, equivalently, with $\mathrm{det}\,\Lambda =1$, is called unimodular or self-dual.
The theta series associated to a lattice $\Lambda$  is
\be
	\Theta_\Lambda (z) =\sum_{\vec{x} \in\Lambda} q^{|\vec{x}|^2/2 }=\sum_{m=0}^{\infty}  N_{m}(\Lambda) q^{n_m(\Lambda)/2},
\ee
where $q=e^{2 \pi i z}$. The theta series of an even self-dual lattice is a modular form of weight $d/2$ for $SL(2,\mathbb{Z})$. Additionally, if $P$ is a homogeneous harmonic polynomial of degree $s>0$, then for an even self-dual lattice $\Lambda$ the weighted theta series,
\be \label{eq:weighted_theta}
\Theta_{\Lambda,P} (z) =\sum_{\vec{x} \in\Lambda} P(\vec{x} ) q^{|\vec{x}|^2/2 },
\ee
is a cusp form of weight $d/2+s$ for $SL(2,\mathbb{Z})$.

We can define complex lattices as linear combinations of complex basis vectors $\vec{w}_i \in \mathbb{C}^d$, $i=1, \dots d$,
\be
	\Lambda_{\mathbb{C}} = \left\{ \sum_{i=1}^d \zeta_i \vec{w}_i \, \bigg| \, \zeta_i \in J \right\},
\ee
where $J$ is the ring of integers of a number field.
For example, we can take $J$ as the Eisenstein integers,
\be \label{eq:Eisenstein_Integers}
		\mathscr{E} =\left\{ a+\omega b : a,b\in\mathbb{Z} \right\},
\ee
with $\omega \coloneqq e^{2 \pi i /3}$. A complex lattice $\Lambda_\mathbb{C}$ in $\mathbb{C}^d$ gives a corresponding real lattice in $\mathbb{R}^{2d}$.

\subsection{Examples} 
Here we list some important examples of lattices in various dimensions:
\begin{itemize}
	\item Cubic lattice $\mathbb{Z}^d$:
	the cubic lattice in $d$ dimensions is
	\be
		\mathbb{Z}^d =\{(x_1,\dots,x_d) :  x_i\in\mathbb{Z}\}.
	\ee
	The generator matrix is the $d\times d$ identity matrix. $\mathbb{Z}^d$ is self-dual. $\mathbb{Z}$ solves the sphere packing problem and the kissing number problem in one dimension.
	\item ${A_d}$  lattice: the lattice $A_d$ with $d\geqslant1$ is
	\be
		A_d =\{(x_0,x_1,\dots,x_d)\in\mathbb{Z}^{d+1} : x_0+\dots+x_d=0 \}.
	\ee
	We have $A_1^*\cong A_1\cong \mathbb{Z}$. $A_2 \cong A_2^*$ is the hexagonal lattice, which solves the sphere packing problem and the kissing number problem in two dimensions. $A_3$ is the face-centered cubic (fcc)  lattice and  $A_3^*$ is the body-centered cubic (bcc) lattice. The lattice $A_3$ solves the sphere packing and the kissing number problems in three dimensions.
	\item $D_d$ lattice: the checkerboard lattice $D_d$ with $d\geqslant3$ is
	\be
		D_d =\{(x_1,\dots,x_d)\in\mathbb{Z}^d :  x_1+\dots+x_d \, \, \,  \text{even} \}.
	\ee
	We have $D_3\cong A_3$,  $D_3^*\cong A_3^*$, and $D_4^*\cong D_4$. In four dimensions, $D_4$ gives the densest lattice packing and solves the kissing number problem \cite{Musin_D4_Kissing}, while $D_5$ gives the densest lattice packing in five dimensions.
	\item $E_6$ lattice: the $E_6$ lattice can be defined in terms of the $E_8$ lattice (see below) as
	\be
		E_6 =\left\{ (x_1,\dots,x_8) \in E_8 : x_6=x_7=x_8 \right\}.
	\ee
	It can also be defined as the real form of the three-dimensional complex lattice over the Eisenstein integers \eqref{eq:Eisenstein_Integers} generated by the vectors
	\be
		(\theta,0,0), \qquad (0,\theta,0), \qquad (1,1,1),
	\ee
	where $\theta\coloneqq\omega-\bar\omega$. In six dimensions, $E_6$ has the highest known density and kissing number and it is the densest lattice packing.
	\item $E_7$ lattice: the $E_7$ lattice can be defined in terms of the $E_8$ lattice as
	\be
		E_7 =\left\{ (x_1,\dots,x_8) \in E_8 : x_7=x_8 \right\}.
	\ee
	In seven dimensions, $E_7$ has the highest known density and kissing number and it is the densest lattice packing.
	\item $E_8$ lattice: the $E_8$ lattice is
	\be
		E_8 =\left\{ (x_1,\dots,x_8) : \text{all } x_i\in\mathbb{Z} \text{ or all } x_i\in\mathbb{Z}+\frac12, \, \sum_{i=1}^8 x_i \equiv 0 \,\, (\text{mod }2) \right\}.
	\ee
	It can also be constructed by applying Construction A  to the Hamming code $\mathscr{H}_8$. $E_8$ is even and self-dual. The first few terms of the theta series are
	\be
	\Theta_{E_8} (z) =1+240 q+2160 q^2+6720 q^3+ 17520 q^4 + \dots.
	\ee
  In eight dimensions, the $E_8$ lattice solves the sphere packing problem \cite{viazovska2017sphere} and the kissing number problem \cite{Levenstein1979, ODLYZKO1979210}. 
	\item Laminated lattice $\Lambda_{12}$: $\Lambda_{12}$ (called $\Lambda_{12}^{\textrm{max}}$ in \cite{LatticeCatalog}) is the laminated lattice in 12 dimensions with kissing number 648, which is the largest kissing number of the three inequivalent laminated lattices in 12 dimensions \cite{laminatedLattices}. The Gram matrix can be found in \cite{LatticeCatalog}.
	\item Coxeter--Todd lattice $K_{12}$:	the Coxeter--Todd lattice $K_{12} \cong K^*_{12}$ can be defined as the real form of the six-dimensional complex lattice over the Eisenstein integers generated by the vectors of the form
	\be
		\frac{1}{\sqrt{2}} (\pm\theta,\pm1^5),
	\ee
	where $\theta\coloneqq\omega-\bar\omega$ can be in any of the six positions and there are an even number of minus signs.	 $K_{12}$ gives the densest known sphere packing in 12 dimensions.
	 \item Barnes--Wall lattice $\Lambda_{16}$: the Barnes--Wall lattice $\Lambda_{16} \cong \Lambda^*_{16}$ can be constructed by applying Construction B (see \cite{Conway:1988oqe})  to the first-order Reed--Muller code of length 16. $\Lambda_{16}$ has the highest known density and kissing number in 16 dimensions.
	 \item Leech lattice $\Lambda_{24}$:	 the Leech lattice  is generated by vectors of the form
	 \be
	 	\frac1{\sqrt{8}}(\mp3,\pm1^{23}),
	 \ee
	 where $\mp3$ may be in any of the 24 positions and the upper signs are taken in correspondence with the $1$'s in the codewords of the binary Golay code $\mathscr{C}_{24}$ \cite{Conway:1988oqe}. $\Lambda_{24}$ is the unique even, self-dual lattice with minimal norm 4 in 24 dimensions. The Leech lattice solves the sphere packing problem \cite{cohn2017sphere} and the kissing number problem \cite{Levenstein1979, ODLYZKO1979210}  in 24 dimensions. The first few terms of the theta series for the Leech lattice are
	\be
	\Theta_{\Lambda_{24}} (z) =1+196 560 q^2+16 773 120 q^3+398 034 000 q^4+  \dots.
	\ee
\end{itemize}

\subsection{Spherical codes and designs}
Denote by  $\mathbb{S}^{d-1}$ the $(d-1)$-dimensional unit sphere, 
\be
 \mathbb{S}^{d-1}  = \{ \vec{x} \in \mathbb{R}^d : \, \vec{x} \cdot \vec{x} = 1\}.
\ee
Let $X$ be a non-empty finite subset of  $\mathbb{S}^{d-1}$. If $\phi$ is the largest angle such that for any two distinct points $\vec{x}, \vec{y} \in X$ we have
\be
\vec{x} \cdot \vec{y} \leq \cos \phi,
\ee
then $X$ is called a spherical code with minimal angle $\phi$. The maximum number of points in such a spherical code is denoted by $A(d, \phi)$.  The kissing number in $d$ dimensions is $A(d, \pi/3)$.

A spherical $t$-design \cite{Delsarte_Goethals_Seidel} is a finite set $X\subset  \mathbb{S}^{d-1} $ such that for all polynomials $f$ of degree at most $t$ we have
\be
\frac{1}{|X|} \sum_{\vec{x} \in X} f(\vec{x}) =  \frac{\int_{ \mathbb{S}^{d-1}} f(\vec{x}) \, dV }{\int_{ \mathbb{S}^{d-1}} dV}.
\ee
This is equivalent to requiring that
\be \label{eq:harmonic_zero}
\sum_{\vec{x} \in X} p(\vec{x}) = 0
\ee
for all non-constant homogeneous harmonic polynomials $p$ of degree at most $ t$. By the addition theorem for spherical harmonics, another equivalent condition for $d>2$ is
\be \label{eq:gegenbauer_vanish}
\sum_{\vec{x}, \, \vec{y} \in X} C_s^{(\frac{d}{2}-1)}(\vec{x}\cdot \vec{y}) = 0, \quad s \in \{1, 2, \dots ,t \},
\ee
where $C_s^{(\frac{d}{2}-1)}$ is a Gegenbauer polynomial.
When suitably rescaled, each shell of the $E_8$ lattice forms a spherical $7$-design and also satisfies \eqref{eq:harmonic_zero} for degree-10 harmonic polynomials, so each shell is said to form a spherical $7\frac{1}{2}$-design. Similarly, each shell of the Leech lattice forms a spherical $11$-design and also satisfies \eqref{eq:harmonic_zero} for degree-14 harmonic polynomials, so each shell is said to form a spherical $11\frac{1}{2}$-design \cite{MR752931, Pache_designs}. 

\subsection{Flat orbifolds}

We conclude this section with a brief review of flat orbifolds---for more detailed accounts, see \cite{thurstonBook, Bettiol2018}. 
A crystallographic group (or space group) $\Gamma$ is a discrete co-compact subgroup of the Euclidean isometry group, $E(d) \coloneqq O(d) \ltimes \mathbb{R}^d$. The quotient $\mathcal{O}=\mathbb{R}^d/\Gamma$ is a closed flat $d$-orbifold. If $\Gamma$ is torsion-free, then $\Gamma$ is a Bieberbach group and  $\mathbb{R}^d/\Gamma$ is a closed flat $d$-manifold \cite{Charlap:1986bgfm,wolf2011spaces}. For the special case of a lattice $\Lambda$, the quotient $\mathbb{R}^d /\Lambda$ is a flat torus.

Let $\mathrm{Aff}(\mathbb{R}^d) \coloneqq \mathrm{ GL}(d, \mathbb{R})  \ltimes \mathbb{R}^d$ be the group of affine transformations of $\mathbb{R}^d$ and let $r: \mathrm{Aff}(\mathbb{R}^d) \to \mathrm{ GL}(d, \mathbb{R})$ be the canonical projection. The subgroup $r(\Gamma) \subset O(d)$ is the holonomy group. The orbifold $\mathbb{R}^d/\Gamma$ is orientable if and only if $r(\Gamma) \subset SO(d)$. The translation subgroup defined by $\Lambda \coloneqq {\rm ker}(r) \cap \Gamma$ is a lattice and the torus $\mathbb{R}^d/\Lambda$ is a finite cover of the orbifold $\mathbb{R}^d/\Gamma$.

The number of affine equivalence classes of closed flat orbifolds of a given dimension is finite and known explicitly for $d \leq 6$.
In $d=2$, there are 17 classes, corresponding to the 17 wallpaper groups. Two of these are also manifolds, namely the 2-torus and the Klein bottle. Out of these 17, only five are orientable; they are commonly called (in crystallographic notation) $p1$, $p2$, $p3$, $p4$, and $p6$, with the group $p1 \simeq\mathbb{Z}^2$ corresponding to the 2-torus. Their actions on $\mathbb{R}^2$ as well as the underlying lattice are presented in \cite{10.1063/1.2735211,kaye_orbifold}.
In $d=3$, there are 219 classes of flat orbifolds. The sublist of flat manifolds has 10 elements. Of these, six are orientable: they are usually denoted by $E_i$ with $i=1,\dots,6$.
These manifolds have been studied in cosmology as possible spatial geometries of the universe---see, e.g., \cite{Peng:2019xoj, COMPACT:2023rkp}. 

\section{Overlap integrals} \label{sec:sumrules}

In this section, we derive expressions for the integrals of products of eigenfunctions that we use to derive spectral identities in the next section.

\subsection{Tensor fields}
Let $(\mathcal{M},g)$ be a closed $d$-dimensional Riemannian manifold with Levi--Civita connection $\nabla$.
Denote by $ {\rm Sym}^s(T^* \mathcal{M})$ the $s$\textsuperscript{th} symmetric power of the cotangent bundle of $\mathcal{M}$ and $S^s(\mathcal{M}) =\Gamma({\rm Sym}^s(T^* \mathcal{M}))$ the space of smooth sections of this bundle, so $\varphi^{(s)} \in S^s(\mathcal{M})$ is a smooth symmetric rank-$s$ tensor field. We often use abstract index notation so that, e.g., $\varphi^{(s)}$ is denoted by $\varphi^{(s)}_{a_1 \dots a_s}$. 

Define the $L^2$-inner product for symmetric tensor fields
\be \label{eq:inner_product}
\langle \varphi_1^{(s)}, \varphi_2^{(s)} \rangle = \int_{\mathcal{M}} g^{a_1 b_1} \dots g^{a_s b_s} \varphi^{(s)}_{1, a_1 \dots a_s} \bar{\varphi}^{(s)}_{2, b_1 \dots b_s} \, dV,
\ee
where the bar denotes complex conjugation.

For $s\geq 2$, we denote by $L$ the map which multiplies a tensor by the metric and symmetrizes the result,
\be
L: S^{s-2}(\mathcal{M}) \rightarrow S^{s}(\mathcal{M}), \qquad \varphi^{(s-2)}_{a_1 \dots a_{s-2}} \mapsto g_{(a_1 a_2} \varphi^{(s-2)}_{a_3 \dots a_s)}.
\ee
The formal adjoint of $L$ with respect to the inner product \eqref{eq:inner_product} is the contraction map on symmetric rank-$s$ tensor fields,
\be
L^\dagger: S^{s}(\mathcal{M}) \rightarrow S^{s-2}(\mathcal{M}), \qquad \varphi^{(s)}_{a_1 \dots a_s} \mapsto g^{a_1 a_2} \varphi^{(s)}_{a_1 \dots a_s}.
\ee
We denote by $S^s_0(\mathcal{M}) \subset S^s(\mathcal{M})$ the subspace of traceless symmetric tensor fields, which is the kernel of $L^\dagger$,
\be
S^s_0(\mathcal{M})  = \ker\left\{ L^\dagger: S^{s}(\mathcal{M}) \rightarrow S^{s-2}(\mathcal{M}) \right\}.
\ee
 For notational convenience, we also define $S^s_0(\mathcal{M})=S^s(\mathcal{M})$ for $s=0$ and $s=1$.

We denote by $\D$ the symmetrized covariant derivative on symmetric tensor fields,
\be \label{eq:d_def}
\D: S^{s}(\mathcal{M}) \rightarrow S^{s+1}(\mathcal{M}), \quad \varphi^{(s)}_{a_1 \dots a_{s}} \mapsto \nabla_{(a_1} \varphi^{(s)}_{a_2 \dots a_{s+1})}.
\ee
The formal adjoint of $\D$ is minus the divergence,
\be
\D^\dagger: S^{s+1}(\mathcal{M}) \rightarrow S^{s}(\mathcal{M}), \qquad \varphi^{(s+1)}_{a_1 \dots a_{s+1}} \mapsto - \nabla^{a_{1}}\varphi^{(s+1)}_{a_1 \dots a_{s+1}}.
\ee
The kernel of $\D^\dagger$ consists of divergence-free symmetric tensor fields. 

We also need to introduce the symmetrized traceless covariant derivative acting on traceless symmetric tensor fields. For $s >0$, we define
\be
P: S_0^{s}(\mathcal{M}) \rightarrow S_0^{s+1}(\mathcal{M}), \quad W^{(s)}_{a_1 \dots a_{s}} \mapsto \nabla_{(a_1} W^{(s)}_{a_2 \dots a_{s+1})}- \frac{s}{2s+d-2}{g}_{(a_1 a_2} \nabla^{b_1}  W^{(s)}_{a_3 \dots a_{s+1} )b_1}.
\ee
For $s=0$, we define $P=\D$.
The formal adjoint of $P$ is the same as $\D^\dagger$ restricted to $S_0^{s+1}(\mathcal{M})$, so the kernel of $P^{\dagger}$ consists of divergence-free traceless symmetric tensor fields. 

We can decompose a symmetric tensor field into a divergence-free traceless tensor, the symmetrized traceless derivative of a traceless tensor, and a trace \cite{Bonifacio:2021msa}
\be
\varphi^{(s)}= \varphi_0^{(s)} + P \psi^{(s-1)}+ L \chi^{(s-2)},
\ee
where $\varphi_0^{(s)} \in S^s_0(\mathcal{M})$ is divergence-free, $\psi^{(s-1)} \in S^{s-1}_0(\mathcal{M})$, and $ \chi^{(s-2)} \in S^{s-2}(\mathcal{M})$. 
This is an orthogonal decomposition with respect to the inner product \eqref{eq:inner_product}. For $s=1$, the decomposition  expresses a 1-form as the sum of a co-closed form and an exact form, 
\be
\varphi^{(1)}= \varphi_0^{(1)} + \D \psi^{(0)}.
\ee
By iteratively applying this decomposition to the lower-rank terms, we obtain a decomposition of a symmetric tensor into divergence-free traceless tensors, 
\be \label{eq:TTdecomposition}
\varphi^{(s)}= \sum_{p=0}^{\lfloor\frac{s}{2}\rfloor} \sum_{q=0}^{s-2p} L^p P^q \varphi_p^{(s-q-2p)},
\ee
where $\varphi_p^{(s)} \in  S^s_0(\mathcal{M})$ is divergence-free. The decomposition \eqref{eq:TTdecomposition} is orthogonal with respect to the inner product \eqref{eq:inner_product} on manifolds of constant curvature.

\subsection{Laplace spectrum}
For the rest of this section, we further assume that $(\mathcal{M},g)$ is flat.
Define the (non-negative) rough Laplacian on $S^s(\mathcal{M})$, 
\be
\Delta: S^s(\mathcal{M}) \rightarrow S^s(\mathcal{M}), \qquad \varphi^{(s)}_{ a_1 \dots a_s} \mapsto - \nabla_b \nabla^b  \varphi^{(s) }_{a_1 \dots a_s}.
\ee 
We consider the Laplace eigenvalue equation for symmetric tensor fields,
\be
\Delta \phi_i^{(s)} = \lambda_i^{(s)} \phi_i^{(s)}.
\ee
By standard results of spectral theory, the eigenvalues are discrete and non-negative, and the eigentensors form a basis for square-integrable symmetric tensor fields. The Laplacian preserves the decomposition \eqref{eq:TTdecomposition} (by the constant curvature assumption), so we can restrict to the eigenvalue equation for traceless divergence-free symmetric tensor fields.

We denote the Laplace eigentensors for traceless divergence-free symmetric rank-$s$ tensor fields by $\phi^{(s)}_i$ and their eigenvalues by $\lambda^{(s)}_i$, where the eigenvalues are in non-decreasing order and are repeated according to their multiplicity,
\be
 \lambda^{(s)}_0 \leq \lambda^{(s)}_1 \leq \lambda^{(s)}_2 \leq \lambda^{(s)}_3  \leq \dots \rightarrow \infty.
\ee
The distinct positive eigenvalues are denoted by $\sigma^{(s)}_i$ with $i \in \mathbb{Z}_{> 0}$,
\be \label{eq:distinct_eigenvalues}
0< \sigma_1^{(s)} < \sigma^{(s)}_2 < \sigma^{(s)}_3  < \dots \rightarrow \infty.
\ee
We also take the eigensections to be real and  orthonormal,
\be \label{eq:orthonormal}
\bar{\phi}^{(s)}_i= \phi^{(s)}_i, \qquad \langle \phi^{(s)}_i, \phi^{(s)}_j \rangle = \delta_{ij}.
\ee
For notational convenience, we define $\lambda_i \coloneqq \lambda_i^{(0)}$. We also set $\sigma_i \coloneqq \sigma_i^{(0)}$ for $i>0$ and $\sigma_0=0$.

\subsubsection{Flat tori}

The Laplace--Beltrami operator on a flat torus $\mathbb{R}^d/\Lambda$, with the metric induced from the Euclidean metric on $\mathbb{R}^d$, is given in Cartesian coordinates by
\be
	\Delta =-\sum_{a=1}^d \frac{\partial^2}{\partial x_a^2}.
\ee
There is a complex eigenfunction $\chi_{\vec{k}}$ for each vector $\vec{k}$ in the dual lattice,
\be
	\chi_{\vec{k}} (\pvec{x}) = \frac{1}{\sqrt{V}} e^{2\pi i \vec{k} \cdot \vec{x}}, \quad \vec{k}\in\Lambda^*,
\ee
where $V$ is the volume of $\mathbb{R}^d/\Lambda$. The corresponding eigenvalue is given by 
\be
\lambda_{\vec{k}} = 4\pi^2 |\vec{k}|^2.
\ee
The $k$\textsuperscript{th} distinct nonzero Laplace eigenvalue of $\mathbb{R}^d/\Lambda$, $\sigma_k$, is thus related to the $k$\textsuperscript{th} distinct nonzero norm of $\Lambda^*$,
\be \label{eq:norm_to_eigenvalue}
\sigma_k= 4 \pi^2 n_k(\Lambda^*).
\ee

Let $\Lambda^*_+ \subset \Lambda^*$ contain one representative from each pair $\{ \vec{k}, - \vec{k} \}$ with $\vec{k} \in \Lambda^*\backslash \{\vec{0}\}$, e.g., by choosing the vector whose first nonvanishing component is positive.
We can write an orthonormal basis of real eigenmodes as
\be\label{eq:sincosEigenmode}
\phi_0 = \frac{1}{\sqrt{V}}, \quad \phi^+_{\vec{k}} (\pvec{x}) = \sqrt{\frac{2}{V }} \cos (2\pi \vec{k} \cdot \vec{x}), \qquad \phi^-_{\vec{k}} (\pvec{x}) =  \sqrt{\frac{2}{V}} \sin (2\pi \vec{k}\cdot \vec{x}),
\ee
for $\vec{k} \in \Lambda^*_+$.
The Fourier expansion of a smooth real function $f$ is
\be \label{eq:Fourier_expand}
f(\vec{x})= \alpha_0 \phi_0 + \sum_{\vec{k} \in \Lambda^*_{+}  } \left( \alpha^+_{\vec{k} } \, \phi^+_{\vec{k}}(\vec{x})+ \alpha^-_{\vec{k} } \, \phi^-_{\vec{k}}(\vec{x}) \right),
\ee
where
\be
\alpha_0= \langle f, \phi_0\rangle, \quad \alpha_{\vec{k}}^{\pm} = \langle f, \phi_{\vec{k}}^{\pm} \rangle.
\ee

The eigentensors can be written as eigenfunctions multiplied by constant polarization tensors. A basis of traceless divergence-free rank-$s$ symmetric tensor eigenmodes is given by 
\be
\phi^{(s)}_{\rho, {0} }= e^{(s)}_{\rho}({0}), \quad \phi^{(s) \pm}_{\rho, \vec{k} } = e^{(s)}_{\rho}(\vec{k})  \phi_{\vec{k}}^{\pm}
\ee
where $e^{(s)}_{\rho}(\vec{k})$ is a normalized constant symmetric rank-$s$ polarization tensor satisfying
\be
0 = g^{a_1 a_2}\left(e^{(s)}_{\rho}(\vec{k}) \right)_{a_1 \dots a_s}=k^{a_1} \left(e^{(s)}_{\rho}(\vec{k})\right)_{a_1 \dots a_s },
\ee
and the index $\rho$ runs over a basis of such tensors.
The corresponding eigenvalue is
\be
\lambda^{(s)}_{\rho, \vec{k}}=4\pi^2 |\vec{k}|^2.
\ee
For flat tori with $d>2$, the tensor spectrum is thus the same as the scalar spectrum up to multiplicity.
The zero-modes $e^{(s)}_{\rho}(0)$ are constant symmetric traceless rank-$s$  tensors.
For later use, we define the projector for such tensors,
\be \label{eq:projector}
\pi^{(s)}_{a_1 \dots a_s,b_1 \dots b_s} \coloneqq V\sum_{\rho} e^{(s)}_{\rho}(0)_{ a_1 \dots a_s} e^{(s)}_{\rho}(0)_{ b_1 \dots b_s}.
\ee 

\subsubsection{Flat orbifolds}
For a flat orbifold $\mathbb{R}^d / \Gamma$, the eigenmodes correspond to the $\Gamma/\Lambda$-invariant eigenmodes on the torus cover $\mathbb{R}^d/\Lambda$, which can be obtained by summing over the orbits of the finite group $\Gamma / \Lambda$. For example, for the eigenmodes of the scalar Laplacian we consider combinations of plane waves of the form
\begin{equation}
	\phi(\vec{x})=\sum_{\gamma\in\Gamma/\Lambda} e^{2\pi i \vec{k} \cdot (\gamma \vec{x} )}.
\end{equation}

\subsubsection{Transverse-traceless spectral decomposition}
Combining the decomposition \eqref{eq:TTdecomposition} with the decomposition into eigensections gives
\be \label{eq:TTeigendecomposition}
\varphi^{(s)}= \sum_{p=0}^{\lfloor\frac{s}{2}\rfloor} \sum_{q=0}^{s-2p}{\vphantom{\sum}}' \sum_{i=0}^\infty C_{pqi} L^p P^q\phi_i^{(s-q-2p)},
\ee
where $C_{pqi}$ are constants and the prime denotes that the middle sum is restricted to $q=0$ for eigenmodes with vanishing eigenvalues. The different terms in this decomposition are orthogonal,
\be \label{eq:orthogonality}
	\langle L^p P^q \phi_i^{(s-q-2p)}, L^{p'} P^{q'} \phi_j^{(s-q'-2p')} \rangle = \delta_{i j} \delta_{p p'} \delta_{q q'} \left| L^p P^q \phi_i^{(s-q-2p)} \right|^2,
\ee
and so the coefficients in \eqref{eq:TTeigendecomposition} are given by
\be
C_{pqi} = \frac{\langle L^p P^q \phi_i^{(s-q-2p)}, \varphi^{(s)} \rangle}{\left| L^p P^q \phi_i^{(s-q-2p)} \right|^2}.
\ee

\subsection{Norms}

We now derive an explicit expression for the norms appearing on the RHS of \eqref{eq:orthogonality}, which we use below. Our approach is to first find the norm of $P^{q}  \phi_i^{(s)}$ and then to use a recursion relation to get the norm of $L^{p} P^{q}  \phi_i^{(s)} $. The final result is given in \eqref{eq:LpPq_norm} below.

First, note that we can write
\be \label{eq:Pq_expansion}
P^{q}  \phi_i^{(s)} = \sum_{p=0}^{\lfloor q/2 \rfloor} \alpha_{i,q,p}^{(s)} L^p \D^{q-2p} \phi_i^{(s)},
\ee
where $\alpha^{(s)}_{i,q,0}=1$ and tracelessness fixes
\be \label{eq:alpha}
\alpha^{(s)}_{i,q,p} = 
\frac{(-1)^{p+1}(q+1-2p) (q+2-2p)_{2p-1}}{ 2^{2p-1}(2s+d+2q-4)p! (3-q-d/2-s)_{p-1}}  \left(\lambda_i^{(s)}\right)^p, \quad 0<p\leq  q/2.
\ee
The Pochhammer symbol is defined by $(x)_n = x(x+1)\dots(x+n-1)$. 

Next, note that
\be \label{eq:norm_precursor_1}
	\langle \D^q \phi_i^{(s)}, L^p \D^{q-2p} \phi_i^{(s)} \rangle = \frac{(-1)^p s!q!}{(q+s)!}  \left(\lambda_i^{(s)}\right)^{q-p},
\ee
which follows from integration by parts. 
Now consider
\begin{align}
	\left| P^q\phi_i^{(s)} \right|^2 & =  \langle \D^q\phi_i^{(s)} , P^q\phi_i^{(s)} \rangle  =    \sum_{p=0}^{\lfloor q/2 \rfloor} \alpha_{i,q,p}^{(s)} \langle \D^q\phi_i^{(s)} , L^p \D^{q-2p} \phi_i^{(s)}\rangle.
\end{align}
Using \eqref{eq:norm_precursor_1} and \eqref{eq:alpha}, we obtain
\be \label{eq:norm_precursor_2}
	\left| P^q\phi_i^{(s)} \right|^2  =\frac{\sqrt{\pi } s! q! (d+2 s+q-3)! \left(\lambda_i^{(s)}\right)^q}{(s+q)!  2^{d+2 s+q-3}\Gamma
   \left(\frac{d-1}{2}+s\right) \Gamma \left(\frac{d}{2}+s+q-1\right)},
\ee
except when $d=2$, $s=q=0$, in which case we have \eqref{eq:orthonormal}. 

There is a recursion relation for the norms in \eqref{eq:orthogonality}. For $p\geq 1$, we have
\be
	\left| L^{p} P^{q}  \phi_i^{(s)} \right|^2 = \frac{2 p (d+2q+2s+2p-2)}{(q+s+2p)(q+s+2p-1) } \left| L^{p-1} P^{q}  \phi_i^{(s)} \right|^2,
\ee
which can be derived by peeling off one of the metric factors in the integral. Applying this $p$ times gives
\be \label{eq:recursion_solve}
	\left| L^p P^{q}  \phi_i^{(s)} \right|^2 = \frac{2^p p! (q+s)! (d+2q+2s+2p- 2)!!}{(q+s+2p)!(d+2q+2s- 2)!!} \left| P^{q} \phi_i^{(s)} \right|^2 .
\ee
Using \eqref{eq:norm_precursor_2}, we then have
\be \label{eq:LpPq_norm}
	\left| L^p P^{q}  \phi_i^{(s)} \right|^2 = \frac{ \sqrt{\pi } 2^{p+3-d-2s-q}s! p!  q! (d+2s+q-3)! (d+2q+2s+2p- 2)!! \left(\lambda_i^{(s)}\right)^{q}}{(q+s+2p)!(d+2q+2s- 2)!! \, \Gamma
   \left(\frac{d-1}{2}+s\right) \Gamma \left(\frac{d}{2}+q+s-1\right)},
\ee
except when $d=2$, $q=0$, in which case we can use \eqref{eq:recursion_solve} and \eqref{eq:orthonormal}.
 
\subsection{Triple overlap integrals}
We define the triple product integrals between two scalar eigenfunctions and a traceless divergence-free symmetric rank-$s$ eigentensor,
\be\label{eq:cijktensor}
	c_{ijk}^{(s)} \coloneqq  \langle  \phi_j \D^s \phi_i,  \phi^{(s)}_{k}\rangle =\int_{\mathcal{M}} \nabla^{(a_1} \cdots \nabla^{a_s)} \phi_i \, \phi_j \, \phi^{(s)}_{k, a_1 \dots a_s} dV.
\ee
We can write a general cubic integral involving derivatives of the eigenmodes $ \phi_i$, $\phi_j$, and $\phi_k^{(s)}$ in terms of $c_{ijk}^{(s)}$ by repeatedly integrating by parts. Explicitly, using the condensed notation $a(t)$ for a group of indices $a_1, \dots, a_t$ that are symmetrized with unit weight, we have
\begin{align}
	\int_{\mathcal{M}} & \nabla^{a(t_1)} \nabla^{b(t_2)} \nabla^{c(t_3)} \phi_i \nabla^{e(t_4)} \nabla^{f(t_5)}\nabla_{b(t_2)} \phi_j \nabla_{c(t_3)} \nabla_{e(t_4)} \phi^{(s)}_{k, \,a(t_1) f(t_5)} dV \nonumber \\
	& = \frac{(-1)^{t_5}}{2^{t_2+t_3+t_4}} \left( \lambda_i+\lambda_j-\lambda_k^{(s)}\right)^{t_2} \left( \lambda_i+\lambda_k^{(s)}-\lambda_j\right)^{t_3} \left( \lambda_j+\lambda_k^{(s)}-\lambda_i\right)^{t_4} c_{ijk}^{(s)}. \label{eq:cijk_reduction}
\end{align} 
Note that by the reality of the eigenmodes we have
\be \label{eq:c_reality}
\left( c_{ijk}^{(s)}\right)^2 \geq 0.
\ee
We also note that 
\be
c_{ijk}^{(s)} = (-1)^s c_{jik}^{(s)},
\ee
which implies that $c_{iik}^{(s)}$ vanishes if $s$ is odd.

Generalizing \eqref{eq:c_ij_def} from the introduction, for $i,j>0$ we define 
\be \label{eq:ctilde}
	{ c}_{ij}^{(s)} (\sigma ) \coloneqq   \sqrt{\frac{V}{ (\lambda_i \lambda_j)^{\frac{s}{2}} }\sum_{l}  \left( c_{ijl}^{(s)}\right)^2},
\ee
where the sum is over the eigenspace corresponding to the eigenvalue $\sigma$ and ${ c}_{ij}^{(s)} (\sigma )=0$ if the eigenspace is empty. Later we show how to constrain this quantity using bootstrap bounds. It is independent of the choice of basis for the eigenspace that is summed over and it is invariant under constant rescalings of the metric.
We also write $c_{ijk} \coloneqq c_{ijk}^{(0)}$, as in the introduction, and ${c}_{ij}(\sigma) \coloneqq {c}_{ij}^{(0)}(\sigma)$.

\subsubsection{${c}_{ij}(\sigma_k)$ for flat tori}

Consider the torus $\mathbb{R}^d/\Lambda$ defined by the lattice $\Lambda$.
Take a real Fourier eigenbasis as in \eqref{eq:sincosEigenmode}. Then the nonzero scalar triple overlaps have three cosines or one cosine and two sines, e.g.,
\be \label{eq:cijk_+++}
\langle \phi^+_{\vec{k}_i}  \phi^+_{\vec{k}_j}, \phi^+_{\vec{k}_k} \rangle= \frac{1}{\sqrt{2 V}} \left(\delta_{\vec{k}_i+\vec{k}_j-\vec{k}_k}+\delta_{\vec{k}_i-\vec{k}_j+\vec{k}_k}+\delta_{-\vec{k}_i+\vec{k}_j+\vec{k}_k} +\delta_{\vec{k}_i+\vec{k}_j+\vec{k}_k} \right)
\ee
and 
\be \label{eq:cijk_++-}
\langle \phi^+_{\vec{k}_i}  \phi^-_{\vec{k}_j}, \phi^-_{\vec{k}_k} \rangle = \frac{1}{\sqrt{2 V}} \left( \delta_{\vec{k}_i+\vec{k}_j-\vec{k}_k} +\delta_{\vec{k}_i-\vec{k}_j+\vec{k}_k}-\delta_{-\vec{k}_i+\vec{k}_j+\vec{k}_k} -\delta_{\vec{k}_i+\vec{k}_j+\vec{k}_k} \right),
\ee
where $\delta_{\vec{k}} = 0$ if $\vec{k} \neq \vec{0}$ and  $\delta_{\vec{0}} = 1$.

The triple products are basis-dependent due to the multiplicity of the eigenvalues. Suppose we rotate each eigenspace to obtain a new basis $ \phi'_{i} = \sum_{j} R_{ij} \phi_{j}$. Then the triple products in the new basis are
\be
c'_{ijk} = \sum_{l, m, n} R_{i l}R_{j m} R_{k n} c_{lmn}.
\ee

Using the expansion into Fourier modes \eqref{eq:Fourier_expand}, we can write $c_{ij}(\sigma_k)$ for a general basis in terms of the overlaps \eqref{eq:cijk_+++} and \eqref{eq:cijk_++-} in the Fourier basis. For a fixed non-constant eigenfunction $\phi_i$, we write
\be  \label{eq:Fourier_phi_i}
\phi_i(\pvec{x})= \sum_{\vec{k}_i \in \Lambda^*_{+} } \left( \alpha^+_{\vec{k}_i } \, \phi^+_{\vec{k}_i}(\vec{x})+ \alpha^-_{\vec{k}_i } \, \phi^-_{ \vec{k}_i}(\vec{x}) \right).
\ee
Since $\phi_i$ is an eigenfunction with eigenvalue $\lambda_i$, the coefficients $ \alpha^{\pm}_{\vec{k}_i }$ in \eqref{eq:Fourier_phi_i} can be non-vanishing only if $ |\vec{k}_i|^2=\lambda_i/4 \pi^2  $.
Similarly, for $\phi_j$ non-constant we write
\be \label{eq:Fourier_phi_j}
\phi_j(\pvec{x})= \sum_{\vec{k}_j \in \Lambda^*_{+} } \left( \beta^+_{\vec{k}_j } \, \phi^+_{\vec{k}_j}(\vec{x})+ \beta^-_{\vec{k}_j } \, \phi^-_{ \vec{k}_j}(\vec{x}) \right).
\ee
Orthonormality gives
\be \label{orthonormal}
\langle \phi_i, \phi_j \rangle= \sum_{\vec{k}_i \in \Lambda_+^*} \left(\alpha^+_{\vec{k}_i} \beta^+_{\vec{k}_i}  +\alpha^-_{\vec{k}_i} \beta^-_{\vec{k}_i} \right)  = \delta_{ij}.
\ee
For $\sigma=0$, we have
\be \label{eq:cij0}
c^2_{ij}(0) = V c^2_{ij0}= \delta_{ij}.
\ee
For $\sigma_k>0$, we get
\begin{align}
{c}^2_{ij}(\sigma_k) =V\sum_{\vec{k}_k\in \Lambda^*_{+,k}} \Bigg[ & \bigg[\sum_{\vec{k}_i, \vec{k}_j \in \Lambda^*_+ }  \left( \alpha_{\vec{k}_i}^+ \beta_{\vec{k}_j}^+ \langle \phi_{\vec{k}_i}^+\phi_{\vec{k}_j}^+,\phi_{\vec{k}_k}^+ \rangle +\alpha_{\vec{k}_i}^- \beta_{\vec{k}_j}^- \langle \phi_{\vec{k}_i}^-\phi_{\vec{k}_j}^-,\phi_{\vec{k}_k}^+ \rangle \right) \bigg]^2 \nonumber \\
& +\bigg[\sum_{\vec{k}_i, \vec{k}_j \in \Lambda^*_+ }  \left( \alpha_{\vec{k}_i}^+ \beta_{ \vec{k}_j}^- \langle \phi_{\vec{k}_i}^+\phi_{\vec{k}_j}^-,\phi_{\vec{k}_k}^- \rangle +\alpha_{\vec{k}_i}^- \beta_{\vec{k}_j}^+ \langle \phi_{\vec{k}_i}^-\phi_{\vec{k}_j}^+,\phi_{\vec{k}_k}^- \rangle \right) \bigg]^2  \Bigg], \label{eq:cij}
\end{align}
where we have defined 
\be
 \Lambda^*_{+,k}  \coloneqq \Lambda^*_{+} \cap S_k(\Lambda^*),
 \ee
 with $S_k(\Lambda^*)$ the $k$\textsuperscript{th} shell of $\Lambda^*$.

\subsubsection{${c}^{(s)}_{ij}(\sigma_k)$ for flat tori}
Now consider ${c}^{(s)}_{ij}(\sigma_k)$ for a flat torus $\mathbb{R}^d/\Lambda$. This quantity involves the overlaps of scalar eigenfunctions with tensor eigenmodes. We assume that $s$ is even. Writing $\phi_i$ and $\phi_j$ as in \eqref{eq:Fourier_phi_i} and \eqref{eq:Fourier_phi_j}, for $k>0$ we can write 
\begin{align}
\left( { c}_{ij}^{(s)}(\sigma_k) \right)^2 
& = \frac{(2 \pi)^{2s}  V}{(\lambda_i \lambda_j)^{\frac{s}{2}}} \sum_{\vec{k}_k\in \Lambda^*_{+,k }} \sum_{\rho} \sum_{v\in \{ +, - \} }  \Bigg[ \sum_{u  \in \{ +, - \}}\sum_{\vec{k}_i, \vec{k}_j \in \Lambda^*_+ }  k_i^{a_1}\dots k_i^{a_s} e_{\rho}^{(s)}(\vec{k}_k)_{a_1 \dots a_s} \alpha_{\vec{k}_i}^u \beta_{\vec{k}_j}^{uv} \langle \phi_{\vec{k}_i}^u \phi_{\vec{k}_j}^{uv},\phi_{\vec{k}_k}^v \rangle  \Bigg]^2 ,
\end{align}
where $uv$ stands for the multiplication of the signs $\pm$.
For the case $\sigma=0$, we have
\begin{align} \label{eq:zeromode_overlap}
\left({ c}_{ij}^{(s)}(0)\right)^2  & = \frac{(2 \pi)^{2s} }{ (\lambda_i \lambda_j)^{\frac{s}{2}}  } \sum_{\vec{k}_i,\vec{k}_j \in \Lambda_+^*} \left(\alpha^+_{\vec{k}_i} \beta^+_{\vec{k}_i}  +\alpha^-_{\vec{k}_i} \beta^-_{\vec{k}_i} \right)  \left(\alpha^+_{\vec{k}_j} \beta^+_{\vec{k}_j}  +\alpha^-_{\vec{k}_j} \beta^-_{\vec{k}_j} \right)\pi^{(s)} (\vec{k}_{i}, \vec{k}_{j}),
\end{align}
where we have defined the following contraction of the projector \eqref{eq:projector}:
\be
	\pi^{(s)} (\vec{k}_{i}, \vec{k}_{j}) \coloneqq \vec{k}_{i}^{a_1}\cdots \vec{k}_{i}^{a_s} \pi^{(s)}_{a_1 \dots a_s,b_1 \dots b_s} \vec{k}_{j}^{b_1}\cdots \vec{k}_{j}^{b_s}.
\ee
Note that this is a homogeneous harmonic polynomial of degree $s$ in each of $\vec{k}_i$ and $\vec{k}_j$.
For $d>2$, we have
\be
	\pi^{(s)} (\vec{k}_{i},\vec{k}_{j}) =\frac{s!}{2^s \left( \frac  d 2 - 1 \right)_s} |\vec{k}_{i}|^s |\vec{k}_{j}|^s  \, C_s^{\left( \frac d2 - 1 \right)}\left( \frac{ \vec{k}_{i} \cdot \vec{k}_{j} }{  |\vec{k}_{i}| |\vec{k}_{j}|} \right),
\ee
where $C^{(\alpha)}_s(x)$ is a Gegenbauer polynomial, and for $d=2$ and  $s\in \mathbb{Z}_{>0}$, we have
\be \label{eq:2D_projector_contraction}
\pi^{(s)} (\vec{k}_{i},\vec{k}_{j}) =\frac{|\vec{k}_{i}|^s |\vec{k}_{j}|^s  }{2^{s-1}} \, T_s\left( \frac{ \vec{k}_{i} \cdot \vec{k}_{j} }{ |\vec{k}_{i}| |\vec{k}_{j}|  } \right),
\ee
where $T_s$ is a Chebyshev polynomial of the first kind \cite{Taronna:2010qq,Dolan:2011dv,Costa_2016}.
For the special case where we take $\phi_i$ and $\phi_j$ to be  real Fourier modes \eqref{eq:sincosEigenmode}, we have
\be \label{eq:tensor_zero_mode_overlap}
	{ c}_{ij}^{(s)}(0) =  \sqrt{\frac{(d-2)_s}{2^s \left( \frac d2 - 1 \right)_s}}\delta_{ij},
\ee
where we used $C_s^{(\alpha)}(1)=(2 \alpha)_s/s!$ and $|\vec{k}_{i}|^2 =\lambda_i/(4\pi^2)$.

\subsubsection{Orbifolds}

We can compute the eigenmode overlaps $c_{ij}(\sigma)$ for flat orbifolds by restricting to the $\Gamma/\Lambda$-invariant eigenmodes on the torus cover. Let $\phi_i$ be a normalized orbifold eigenmode and $\tilde{\phi}_i$ its lift to the torus cover, rescaled by a normalization constant,
\be
 \phi_i = \sqrt{ |\Gamma/\Lambda|} \, \tilde{\phi}_i \Big|_{\mathbb{R}^d/\Gamma},
\ee
where we used $V(\mathbb{R}^d/\Lambda)/V(\mathbb{R}^d/\Gamma) = |\Gamma/\Lambda|$.
Then we have
\be
\sqrt{V(\mathbb{R}^d/\Gamma) } \int_{\mathbb{R}^d/\Gamma} {\phi}_i {\phi}_j {\phi}_k \, dV = \sqrt{V(\mathbb{R}^d/\Lambda) } \int_{\mathbb{R}^d/\Lambda} \tilde{\phi}_i \tilde{\phi}_j \tilde{\phi}_k \, dV.
\ee
As a result, we do not need an explicit orbifold fundamental domain to calculate the scale-invariant eigenmode overlaps, we can just use the invariant eigenmodes on the covering torus.

\subsection{Quartic overlap integrals}
\label{sec:quartic_overlaps}

Fix $s_1, s_2, s_3 \in \mathbb{Z}_{\geq 0}$. We consider the following quartic overlap integral of the scalar eigenfunction $\phi_i$:
\be \label{eq:quartic_overlap}
I_{s_1, s_2, s_3} \coloneqq \langle \D^{s_1+s_2+s_3} \phi_i, \D^{s_1} \phi_i \D^{s_2} \phi_i  \D^{s_3} \phi_i  \rangle,
\ee 
where $\D$ is the symmetrized covariant derivative \eqref{eq:d_def}. 
Then we trivially have 
\be \label{eq:s=t=u}
I_{s_1, s_2, s_3} =I_{s_2, s_3,s_1}=I_{s_3, s_1, s_2}.
\ee
Our goal is to write this quartic overlap integral as a sum of products of triple overlap integrals in multiple ways.

Define the symmetric tensor field $T^{(s_2+s_3)}$ obtained by contracting all indices of $\D^{s_1} \phi_i$ with $\D^{s_1+s_2+s_3} \phi_i$,
\be
T^{(s_2+s_3)}_{a_1 \dots a_{s_2+s_3}} \coloneqq \nabla_{(a_1} \cdots \nabla_{a_{s_1+s_2+s_3})} \phi_i \nabla^{(a_{s_2+s_3+1}} \cdots \nabla^{a_{s_1+s_2+s_3}) }\phi_i.
\ee
We can then write 
\be \label{eq:s-channel_precursor}
I_{s_1, s_2, s_3}  =\langle T^{(s_2+s_3)}, \D^{s_2} \phi_i  \D^{s_3} \phi_i  \rangle .
\ee
Decomposing each tensor in this inner product using \eqref{eq:TTeigendecomposition} gives the s-channel decomposition of the quartic overlap integral:
\be \label{eq:s-channel}
I_{s_1, s_2, s_3}  = \sum_{p=0}^{\lfloor\frac{s_2+s_3}{2}\rfloor} \sum_{q=0}^{s_2+s_3-2p}{\vphantom{\sum}}' \sum_{k=0}^\infty \frac{ \langle T^{(s_2+s_3)}, L^p P^q\phi_k^{(s_2+s_3-q-2p)} \rangle \langle L^p P^q\phi_k^{(s_2+s_3-q-2p)}, \D^{s_2} \phi_i  \D^{s_3} \phi_i  \rangle }{\left| L^p P^{q}  \phi_k^{(s_2+s_3-q-2p)} \right|^2}.
\ee

The inner products appearing in the numerator of the RHS of \eqref{eq:s-channel} can be expressed in terms of the triple overlap integrals $c_{ijk}^{(s)}$. Let us start with $\langle T^{(s_2+s_3)}, L^p P^q\phi_k^{(s_2+s_3-q-2p)} \rangle$. Using \eqref{eq:Pq_expansion} and then \eqref{eq:cijk_reduction} gives
\begin{align}
&\langle T^{(s_2+s_3)}, L^p P^q\phi_k^{(s_2+s_3-q-2p)} \rangle  = \sum_{r=0}^{\lfloor q/2 \rfloor} \alpha_{k,q,r}^{(s_2+s_3-q-2p)}  \langle T^{(s_2+s_3)}, L^{p+r} \D^{q-2r} \phi_k^{(s_2+s_3-q-2p)} \rangle \\
& = \sum_{r=0}^{\lfloor q/2 \rfloor} \frac{\alpha_{k,q,r}^{(s_2+s_3-q-2p)}}{2^{s_1+q-2r}}  (- \lambda_i)^{p+r}\left(2 \lambda_i-\lambda_k^{(s_2+s_3-q-2p)}\right)^{s_1}\left[ \lambda_k^{(s_2+s_3-q-2p)} \right]^{q-2r} c_{iik}^{(s_2+s_3-q-2p)}. \label{eq:IP1}
\end{align}
The other inner product in \eqref{eq:s-channel} is more difficult to evaluate. Using \eqref{eq:Pq_expansion}, we get
\begin{align}
 \langle \D^{s_2} \phi_i  \D^{s_3} \phi_i  , L^p P^q\phi_k^{(s_2+s_3-q-2p)}\rangle = \sum_{r=0}^{\lfloor q/2 \rfloor} \alpha_{k,q,r}^{(s_2+s_3-q-2p)} \langle \D^{s_2} \phi_i  \D^{s_3} \phi_i  ,  L^{p+r} \D^{q-2r} \phi_k^{(s_2+s_3-q-2p)}\rangle. \label{eq:IP2_a}
\end{align}
This can be written in terms of the triple overlap integrals \eqref{eq:cijktensor} using the following result, which we derive below:
\begin{align}
& \big\langle \D^{s_2} \phi_i  \D^{s_3} \phi_i , L^{p} \D^{q} \phi^{(l)}_{k}   \big\rangle \nonumber \\
& = \sum^{\min(s_2, l)}_{\gamma=\max(0, l-s_3)} \sum_{\beta=0}^{\min\left(\lfloor\frac{s_2-\gamma}{2} \rfloor, p \right)} \sum^{\min( \alpha_1, p-\beta,s_3+\gamma-l)}_{\alpha=\max(0,\alpha_{0})}   \frac{(-1)^{\gamma} s_2!s_3! p!l!q!(-\lambda_i)^{p-\alpha} \left(\lambda_k^{(l)}\right)^{q}\left(2 \lambda_i-\lambda_k^{(l)}\right)^{\alpha} c_{iik}^{(l)}}{2^{q}\alpha! \beta! \gamma! (s_2+s_3)! (l-\gamma)!(\alpha_1-\alpha)!(\alpha-\alpha_{0})!(p-\alpha-\beta)!}, \label{eq:IP2_b}
\end{align}
where $\alpha_{0} \coloneqq  2p-2\beta+l-\gamma-s_3$, $\alpha_1\coloneqq s_2-\gamma-2\beta$, and $2p+q+l=s_2+s_3$. 

By combining \eqref{eq:s-channel} with \eqref{eq:IP1}, \eqref{eq:IP2_a}, and \eqref{eq:IP2_b}, we obtain an explicit expression for $I_{s_1, s_2, s_3}$ in terms of eigenvalues and triple products. 

\subsubsection{Deriving \eqref{eq:IP2_b}}
Let us now show how to derive \eqref{eq:IP2_b}. We have
\begin{align}
 \big\langle \D^{s_2} \phi_i  \D^{s_3} \phi_i , L^{p} \D^{q} \phi^{(l)}_{k}   \big\rangle
& = \frac{1}{(s_2+s_3)!} \sum_{\sigma \in S_{s_2 +s_3}} \int I(\sigma)\, dV,
\end{align}
where 
\begin{align}
I(\sigma) = & \nabla^{a_1}\dots \nabla^{a_{s_2}} \phi_i \nabla^{a_{s_2+1}} \dots \nabla^{a_{s_2+s_3}} \phi_i \, \times \nonumber \\
&  g_{a_{\sigma(1)} a_{\sigma(2)}} \dots g_{a_{\sigma(2p-1)} a_{\sigma(2p)}} \nabla_{a_{\sigma(2p+1)}} \dots \nabla_{a_{\sigma(2p+q)}}\phi^{(l)}_{k, a_{\sigma(2p+q+1)}, \dots, a_{\sigma(2p+q+l)}},
\end{align}
and $S_{s_2+s_3}$ is the symmetric group consisting of all  permutations $\sigma$ of the set $\{1,2, \dots, s_2+s_3\}$. 
There is a right action of the subgroup $H=(\mathbb{Z}_2)^p\times S_p\times S_q \times S_{l}$ on $S_{s_2 +s_3}$ that leaves $I(\sigma)$ invariant, $I(\sigma h) = I(\sigma)$ for all $h\in H$, due to the slot symmetries of $ L^{p} \D^{q} \phi^{(l)}_{k}$.
We can thus write the inner product as a sum over cosets,
\begin{align} \label{eq:coset_sum}
 \big\langle \D^{s_2} \phi_i  \D^{s_3} \phi_i , L^{p} \D^{q} \phi^{(l)}_{k}   \big\rangle & = \frac{|H|}{(s_2+s_3)!} \sum_{ \sigma H \in  S_{s_2 +s_3}/H} \int I(\sigma) \, dV,
\end{align}
where $|H| = 2^pp!l!q!$.

We now classify the cosets according to how the indices of $\nabla^{a_1}\dots \nabla^{a_{s_2}} \phi_i$ are contracted. Let $n_{\alpha, \beta, \gamma}$ be the number of right cosets  $\sigma H \in S_{s_2 +s_3}/H$ satisfying the following conditions:
\begin{enumerate}
\item $\gamma$ of the indices $a_1, \dots, a_{s_2}$ contract with  $\phi_k^{(l)}$, i.e.,  $$ \lvert \left\{ i \in \{ 2p+q+1, \dots, 2p+q+l \} : \sigma(i) \in \{1, \dots, s_2\}\right\} \rvert=\gamma.$$
\item $2 \beta$ of the indices  $a_1, \dots, a_{s_2}$ contract in pairs with $\beta$ metric factors, i.e., $$|\left\{i \in \{ 1, \dots, p \} : \sigma(2i-1), \, \sigma(2i)\in \{1, \dots, s_2 \} \right\}| = \beta.$$
\item $\alpha$ of the indices  $a_1, \dots, a_{s_2}$ contract with one index from $\alpha$ distinct metric factors, i.e., $$|\left\{i \in \{ 1, \dots, 2p \} : \sigma(i)\in \{1, \dots,  s_2 \} \right\}| = \alpha+2\beta.$$
\end{enumerate}
This combinatorial problem has the following solution: 
\be \label{eq:n_abc}
n_{\alpha, \beta, \gamma}=\frac{s_2!s_3!}{2^{p-\alpha} \alpha! \beta!\gamma! (l-\gamma)! (p-\alpha-\beta)!(s_2-\gamma-2\beta-\alpha)!(s_3-l+\gamma-2(p-\beta)+\alpha)! }.
\ee
For  $\sigma H$ satisfying the above conditions, \eqref{eq:cijk_reduction} gives
\begin{align} \label{eq:I_abc}
\int I(\sigma) \, dV =\frac{(-1)^{l+p-\alpha-\gamma}\lambda_i^{p-\alpha}\left(2 \lambda_i -\lambda_k^{(l)}\right)^{\alpha} \left(\lambda_k^{(l)}\right)^q c_{iik}^{(l)}}{2^{q+\alpha}}.
\end{align}

We can sum over all cosets by summing over appropriate ranges of $\alpha$, $\beta$, and $\gamma$. The integer $\gamma$ satisfies the inequalities $\gamma \leq s_2$, $\gamma \leq l$, and $s_3 \geq l- \gamma$, to ensure that there are enough indices available for the contractions. Similarly, we have $2 \beta \leq s_2 - \gamma$ to ensure there are enough remaining indices to contract with the $\beta$ metric factors and $\beta \leq p$ to ensure there are enough metric factors. We also have $\alpha \leq s_2-\gamma-2 \beta$, $\alpha \leq p- \beta$, and $\alpha \leq s_3-(l-\gamma)$, to ensure there are enough remaining indices and metric factors for the contractions with the $\alpha$ metric factors. Lastly, we have $2(p-\alpha -\beta) \leq s_3-(l-\gamma) -\alpha$ to ensure that there are enough remaining indices on $\D^{s_3} \phi_i$ to contract with any leftover metric factors. We thus get
\begin{align}
& \big\langle \D^{s_2} \phi_i  \D^{s_3} \phi_i , L^{p} \D^{q} \phi^{(l)}_{k}   \big\rangle \nonumber \\
&= \frac{ p!l!q!2^p}{(s_2+s_3)!} \sum^{\min(s_2, l)}_{\gamma=\max(0, l-s_3)} \sum_{\beta=0}^{\min\left(\lfloor\frac{s_2-\gamma}{2} \rfloor, p \right)} \sum^{\min( \alpha_1, p-\beta,s_3+\gamma-l)}_{\alpha=\max(0,\alpha_{0})}  \frac{n_{\alpha, \beta, \gamma}(-1)^{l+p-\alpha-\gamma}\lambda_i^{p-\alpha}\left(2 \lambda_i -\lambda_k^{(l)}\right)^{\alpha} \left(\lambda_k^{(l)}\right)^q c_{iik}^{(l)}}{2^{q+\alpha}}, \label{eq:IP2_c}
\end{align}
where $\alpha_{0}\coloneqq 2p-2\beta+l-\gamma-s_3$ and $\alpha_1 \coloneqq s_2-\gamma-2\beta$. We then get \eqref{eq:IP2_b} by substituting \eqref{eq:n_abc} into \eqref{eq:IP2_c} and using $(-1)^l=1$, since $c_{iik}^{(l)}$ vanishes if $l$ is odd.

\section{Bootstrap bounds} \label{sec:bounds}

We now discuss spectral identities and use these to derive bootstrap bounds.

\subsection{Spectral identities}

We saw in Section \ref{sec:quartic_overlaps} how to write $I_{s_1, s_2, s_3}$ in terms of eigenvalues and triple products. We can similarly decompose $I_{s_2, s_3, s_1}$ and $I_{s_3, s_1, s_2}$, which gives the $t$-channel and $u$-channel decompositions of the quartic overlap integral, respectively.
By equating these different decompositions using \eqref{eq:s=t=u}, we obtain two spectral identities for the eigenvalues $\lambda_i^{(s)}$ and overlaps $c_{ijk}^{(s)}$. 

As an example, the first equality in \eqref{eq:s=t=u} with $s_1=s_3=0$ and $s_2=1$ gives
\be
\sum_{k=0}^{\infty} c_{iik}^2(4 \lambda_i-3 \lambda_k ) =0.
\ee
This is the simplest consistency condition for identical external states and it appears in \cite{Csaki:2003dt}. 
As a more complicated example, taking $s_1=0$ and $s_2=s_3=1$ gives
\be \label{eq:cc_011}
\sum_{k=0}^{\infty}\left(c_{iik}^{(2)}\right)^2 +\sum_{k=1}^{\infty}\frac{(16  \lambda_i^2 + 
  4 (d-3 ) \lambda_i \lambda_k + (4 - 3 d) \lambda_k^2)}{16 (d-1)} c_{iik}^2 + \frac{\lambda_i^2}{d}\frac{1}{V}=0,
\ee
where we used \eqref{eq:cij0}. Writing this in terms of the quantity $c_{ij}^{(s)}(\sigma)$ defined in \eqref{eq:ctilde} and multiplying by $V/\lambda_i^2$, we get
\be \label{eq:cc_011_v2}
\sum_{k=1}^{\infty}\left[{ c}_{ii}^{(2)}\left(\sigma^{(2)}_k\right) \right]^2 +\left[c_{ii}^{(2)}(0)\right]^2+ \sum_{k=1}^{\infty}\frac{(16 + 
  4 (d-3 )  \sigma_k\lambda^{-1}_i + (4 - 3 d) \sigma_k^2 \lambda_i^{-2})}{16 (d-1)} c^2_{ii}(\sigma_k) + \frac{c^2_{ii}(0)}{d} =0.
\ee

Different choices of $s_i$ can lead to different spectral identities, but these may not all be independent. Let $\cutoff$ be a positive even integer. If we consider quartic overlap integrals of the form \eqref{eq:quartic_overlap} with at most $\cutoff$ derivatives, so that $2(s_1+s_2+s_3) \leq \cutoff$, then we find that the number of independent spectral identities  is
\begin{align}
n_{\cutoff} \coloneqq \begin{cases} \frac{1}{72} \left(16 + 3  \cutoff (6 + \cutoff) - 16  \cos( \pi \cutoff/ 3) \right), \quad d>2, \\
  \cutoff/2, \quad d=2.
  \end{cases}
\end{align}
We can generate the $n_{\cutoff} $ spectral identities for a given $\cutoff$ when $d>2$ by considering the identities $I_{s_1, s_2, s_3} =I_{s_2, s_3,s_1}$ with the following triples $(s_1, s_2, s_3)$:
\be
	\left\{ (0,s_2,s_3)\in\mathbb{Z}^3 : s_2 \text{ odd, } 1\leqslant s_2 \leqslant \left\lfloor\frac{\cutoff+1}3\right\rfloor, \frac{s_2-1}2 \leqslant s_3 \leqslant \frac{\cutoff}{2}-s_2\right\}.
\ee
For $d=2$, we instead use the triples $\left\{ (0,s_2, 0 ) : s_2 \in \mathbb{Z}, \, 1 \leq s_2 \leq \cutoff/2 \right\}$.
For a given $\cutoff$ and $d$, we can construct the $n_{\cutoff}$ independent spectral identities and write them as a vector equation of the form
\be\label{eq:sumrulesFin}
  \sum_{\substack{s=0 \\ s \, \, \rm{even} } }^{\cutoff/2} \left[ \vec{A}_s \left[ { c}_{ii}^{(s)}(0)\right]^2 + \sum_{k=1}^{\infty} \vec{F}_s \! \left( \sigma^{(s)}_k \lambda^{-1}_i \right)  \left[ {  c}_{ii}^{(s)}\left(\sigma^{(s)}_k\right) \right]^2\right]=0,
\ee
where $\vec{F}_s \! \left( \sigma^{(s)}_k \lambda^{-1}_i \right)$ is a vector of length $n_{\cutoff}$ whose components are polynomials in $ \sigma^{(s)}_k \lambda^{-1}_i$ and $\vec{A}_s$ is a constant vector of length $n_{\cutoff}$. We include in ancillary files the vectors $\vec{A}_s$ and $\vec{F}_s(x)$ for the first 100 spectral identities for $d>2$, corresponding to $\cutoff=46$.

Let us make a few remarks about the spectral identities:
\begin{itemize}
\item We can have $\vec{A}_s \neq \vec{F}_s(0)$, as illustrated by the example in \eqref{eq:cc_011_v2}. This is unlike the spectral identities for hyperbolic manifolds, where the zero mode contributions of a given spin are equal to the nonzero mode contributions in the limit of vanishing eigenvalue. This difference arises because the zero modes for flat manifolds are constant and so their derivatives do not appear in the spectral decomposition.
\item The sums over $k$ have only finitely many nonzero terms, unlike for hyperbolic manifolds.
\item For $d=2$, we set $\vec{F}_s=\vec{0}$ when $s>0$, since the only tensor eigenmodes are zero modes. 
\item Our derivation assumed that $\mathcal{M}$ is a closed flat manifold. However, the  spectral identities also apply to quotients $ \Gamma \backslash \mathcal{M}$ by restricting to $\Gamma$-invariant eigenmodes on $\mathcal{M}$, where $\Gamma$ is any subgroup of the isometry group of $\mathcal{M}$. The spectral identities thus apply to all closed flat orbifolds, since any such orbifold can be written as a finite quotient of a flat torus.
\end{itemize}

\subsection{Semidefinite programming}

We can obtain bootstrap bounds for flat orbifolds by showing that certain choices of spectra and triple products are inconsistent with the spectral identities. To bound triple products, we follow the conformal bootstrap approach for bounding OPE coefficients \cite{Caracciolo:2009bx, Poland:2011ey}. To bound $c_{11}(\sigma_1)$, we use
\begin{proposition} \label{prop:c11_bound} 
Fix a constant $x_0 \geq 0$, an integer dimension $d\geq 2$, and a positive even integer $\cutoff$. Given the spectral identities \eqref{eq:sumrulesFin}, suppose that we have a vector $\vec\alpha\in\mathbb{R}^{n_\cutoff}$ such that the following conditions hold:
\vspace{-.25cm}
\begin{subequations}
\begin{align}
		\vec{\alpha}\cdot\vec{F}_{0}(1) &=1,\label{eq:alpha_norm_c111} \\
		\vec{\alpha}\cdot\vec{F}_0(x) &\geqslant0, \quad  \forall x \geqslant 1, \\
		\vec{\alpha}\cdot \vec A_s &\geqslant 0, \quad  s=2,4, \dots, 2 \lfloor \cutoff/4 \rfloor, \\
		\vec{\alpha}\cdot\vec{F}_s(x) &\geqslant0, \quad  \forall x\geqslant x_0, \quad  s =2, 4, \dots, 2 \lfloor \cutoff/4 \rfloor. \label{eq:tensor_gap}
\end{align}
\label{eq:numconstralpha} 
\end{subequations}
\noindent Then for all $d$-dimensional flat orbifolds with either i) $d=2$, or with ii) $d>2$ and $\sigma_1^{(s)}/\sigma_1 \geq x_0$ for $s =2, 4, \dots, 2 \lfloor \cutoff/4 \rfloor$, we have
\be\label{eq:objineq}
	c_{11}(\sigma_1) \leqslant \sqrt{-\vec\alpha\cdot\vec{A}_0}.
\ee
\end{proposition}
\begin{proof}
Set $i=1$ in the spectral identities \eqref{eq:sumrulesFin}, so that $\lambda_i=\sigma_1$ is the smallest nonzero Laplace eigenvalue.
Taking the dot product with $\vec\alpha$ then gives
\begin{align}
 \vec{\alpha} \cdot \vec{A}_0 + \vec{\alpha} \cdot  \vec{F}_0 \! \left( 1\right)   {c}^2_{11}(\sigma_1) =&-  \sum_{\substack{s=2 \\ s \, \, \rm{even} } }^{\cutoff/2}  \vec{\alpha} \cdot \vec{A}_s \left[ { c}_{11}^{(s)}(0)\right]^2 -  \sum_{k=2}^{\infty}\vec{\alpha} \cdot  \vec{F}_0 \! \left( \sigma_k \sigma^{-1}_1 \right)   {c}^2_{11}\left(\sigma_k\right)  \nonumber \\
& - \sum_{\substack{s=2 \\ s \, \, \rm{even} } }^{\cutoff/2}  \sum_{k=1}^{\infty}\vec{\alpha} \cdot  \vec{F}_s \! \left( \sigma^{(s)}_k \sigma^{-1}_1 \right)  \left[ { c}_{11}^{(s)}\left(\sigma^{(s)}_k\right) \right]^2. \label{eq:expanded_sum_rule}
\end{align}
By \eqref{eq:numconstralpha} and \eqref{eq:c_reality}, all terms on the RHS of \eqref{eq:expanded_sum_rule} are negative or zero, so \eqref{eq:alpha_norm_c111} gives \eqref{eq:objineq}.
\end{proof}
The constant $x_0$ is the assumed lower bound for the ratio of the smallest positive eigenvalue for tensors to the smallest positive eigenvalue for scalars.
This ratio is always equal to one for tori when $d>2$, but it can be less than one for general orbifolds (since the eigenspaces on the torus cover may contain no trivial representations of the quotient group). 
For $x_0=0$, we thus obtain an upper bound on $c_{11}(\sigma_1)$ that is valid for all $d$-dimensional flat orbifolds, which we call an orbifold bound. For $x_0=1$, we obtain a (possibly stronger) upper bound on $c_{11}(\sigma_1)$ for flat $d$-dimensional tori, which we call a torus bound. For $d=2$, there is no distinction between these bounds.

We obtain the strongest bounds by maximizing the objective function $\vec\alpha\cdot\vec{A}_0$ subject to the conditions \eqref{eq:numconstralpha}. 
The search for such an $\vec\alpha$
can be formulated in the language of semidefinite optimization. The resulting semidefinite program can be solved numerically on a computer, as in the numerical conformal bootstrap \cite{Poland:2018epd}. We use the arbitrary-precision semidefinite program solver \texttt{SDPB} \cite{Simmons-Duffin:2015qma, Landry:2019qug}, which is designed to solve problems of this type.

\subsection{Flat tori examples}

Before presenting our bootstrap bounds, let us discuss how to compare these bounds to explicit examples of flat tori. Suppose we have a torus $\mathbb{R}^d/\Lambda$ with a fixed eigenfunction basis $\{ \phi_i \}_{i=0}^{\infty}$. The quantity ${c}_{ij}(\sigma_k)$ depends on the choice of basis. The bootstrap bounds that we derive from the spectral identities apply to any choice of basis. Thus, to compare our bounds to examples, we want to choose a basis that maximizes ${c}_{ij}(\sigma_k)$. We can do this by expanding a general basis in Fourier modes and then optimizing over the choice of Fourier coefficients.

Consider \eqref{eq:cij} with $i=j=1$, which we can write as
\begin{align}
{c}^2_{11}(\sigma_k) =V \sum_{\vec{k}_k\in \Lambda^*_{+,k }} \Bigg[ & \bigg[\sum_{\vec{k}_i, \vec{k}_j \in \Lambda^*_{+,1} }  \left( \alpha_{\vec{k}_i}^+ \alpha_{\vec{k}_j}^+ \langle \phi_{\vec{k}_i}^+\phi_{\vec{k}_j}^+,\phi_{\vec{k}_k}^+ \rangle +\alpha_{\vec{k}_i}^- \alpha_{\vec{k}_j}^- \langle \phi_{\vec{k}_i}^-\phi_{\vec{k}_j}^-,\phi_{\vec{k}_k}^+ \rangle \right) \bigg]^2 \nonumber \\
& +\bigg[\sum_{\vec{k}_i, \vec{k}_j \in \Lambda^*_{+,1} }  \left( \alpha_{\vec{k}_i}^+ \alpha_{ \vec{k}_j}^- \langle \phi_{\vec{k}_i}^+\phi_{\vec{k}_j}^-,\phi_{\vec{k}_k}^- \rangle +\alpha_{\vec{k}_i}^- \alpha_{\vec{k}_j}^+ \langle \phi_{\vec{k}_i}^-\phi_{\vec{k}_j}^+,\phi_{\vec{k}_k}^- \rangle \right) \bigg]^2  \Bigg]. \label{eq:cii}
\end{align}
For a given torus $\mathbb{R}^d/\Lambda$,  \eqref{eq:cii} defines a non-negative homogeneous polynomial of degree four in the $N_1(\Lambda^*)$ variables $\{ \alpha^{\pm}_{\vec{k}}: \vec{k} \in \Lambda^*_{+,1} \}$. These variables are constrained to lie on a sphere due to the normalization condition \eqref{orthonormal}, 
\be \label{eq:sphere}
\sum_{\vec{k} \in \Lambda_{+,1}^*} \left( \left(\alpha^+_{\vec{k}}\right)^2  + \left( \alpha^-_{\vec{k}} \right)^2 \right)  = 1.
\ee
Since \eqref{eq:cii} is a real continuous function, it is bounded on the compact space defined by \eqref{eq:sphere} and achieves its supremum. We denote this maximum value over all possible choices of the normalized Laplace eigenfunction $\phi_1$ by $\max_{\mathbb{R}^d/\Lambda} {c}^2_{11}(\sigma_k)$. 

The problem of finding the global maximum of a quartic form over a sphere is NP-hard \cite{NESTEROV}. 
We use the function \texttt{NMaximize} in \texttt{Mathematica} to numerically search for the maximum, but this becomes impractical for lattices with a large kissing number, e.g., for the Leech lattice we have a quartic form with 196560 variables. However, we can obtain a useful lower bound to the maximum value of ${c}^2_{11}(\sigma_k)$ on a given torus by making a particular choice for $\phi_1$: namely, we take $\phi_1$ to be a uniform combination of cosines from the first eigenspace.  

 To formulate the lower bound, we first make a definition. Recalling \eqref{eq:kth_shell_at_x}, the set of minimal vectors in $\Lambda$ that are in the first shell centered at $\vec{x}\in \Lambda$ is given by
\be \label{eq:intersection}
S_1(\Lambda) \cap S_1(\vec{x}, \Lambda) = \{ \vec{y} \in S_1(\Lambda) : \vec{x}- \vec{y} \in S_1(\Lambda)  \}.
\ee
We define $\nu_k(\Lambda)$ as the root mean square of the number of such vectors over the $k$\textsuperscript{th} shell $S_k(\Lambda)$, 
\be \label{eq:RMS_def}
\nu_k(\Lambda) \coloneqq \left( \frac{1} {N_{k}(\Lambda)}\sum_{\vec{x}\in S_k(\Lambda)} |S_1( \Lambda) \cap S_1(\vec{x}, \Lambda)|^2 \right)^{1/2}.
\ee
In particular, $\nu_1(\Lambda)$ is the root mean square of the number of minimal vectors that are minimal distance from each minimal vector.
 If the automorphism group of $\Lambda$ acts transitively on the $k$\textsuperscript{th} shell, then $ \nu_k(\Lambda)=|S_1(\Lambda) \cap S_1(\vec{x}, \Lambda)|$ for all $\vec{x}\in S_k(\Lambda)$, since $|S_1(\Lambda) \cap S_1(\vec{x}, \Lambda)|$ is constant on orbits of the automorphism group. 
We can now state the lower bound as follows:
\begin{proposition} \label{prop:lower_bound}
On the torus $\mathbb{R}^d/\Lambda$, we have
\be \label{eq:tori_lower_bound}
\frac{N_{k}(\Lambda^*) \nu^2_{{k}} (\Lambda^*)  }{N^2_{1}(\Lambda^*)}   \leq  \max_{\mathbb{R}^d/\Lambda} {c}^2_{11}(\sigma_k).
\ee
\end{proposition}
\noindent For the proof of this proposition, we use the following lemma:
\begin{lemma} \label{lemma:S_k}
For a lattice $\Lambda$ and $\vec{x} \in \Lambda$, we have
\be \label{eq:Sk_overlaps}
 |S_1(\Lambda) \cap S_1(\vec{x}, \Lambda)|= \sum_{\vec{y}, \vec{z} \in S_1(\Lambda) } \delta_{\vec{x}+\vec{y}+\vec{z}}.
\ee
\end{lemma}
\begin{proof}
From \eqref{eq:intersection}, we can equivalently write $|S_1(\Lambda) \cap S_1(\vec{x}, \Lambda)|$ as the number of ordered pairs of minimal vectors $(\vec{y}, \vec{z})$ that add to $\vec{x}$,
\be
|S_1(\Lambda) \cap S_1(\vec{x}, \Lambda) |=  |\left\{ (\vec{y}, \vec{z}) \in S_1(\Lambda) \times S_1(\Lambda) : \vec{y}+ \vec{z} = \vec{x} \right\}|.
\ee
Since $\Lambda$ is invariant under reflections, we have
\begin{align}
|S_1(\Lambda) \cap S_1(\vec{x}, \Lambda)| & = | S_1(\Lambda) \cap S_1(-\vec{x}, \Lambda) | \label{eq:reflection} \\
& =\left|\left\{ (\vec{y}, \vec{z}) \in S_1(\Lambda) \times S_1(\Lambda) : \vec{x}+ \vec{y}+ \vec{z} = \vec{0} \right\} \right|  \\
& =\sum_{\vec{y}, \vec{z} \in S_1(\Lambda) } \delta_{\vec{x}+\vec{y}+\vec{z}}.
\end{align}
\end{proof}

\begin{proof}[Proof of Proposition~\ref{prop:lower_bound}]
To obtain the lower bound, we evaluate $ {c}^2_{11}(\sigma_k)$ with a particular choice of $\phi_1$. In particular, take the non-vanishing Fourier coefficients in \eqref{eq:Fourier_phi_i} to be
\be \label{eq:uniform_cosines}
\alpha^+_{\vec{k}} = \sqrt{\frac{2}{N_1(\Lambda^*)}}, \quad \vec{k} \in \Lambda^*_{+,1},
\ee 
so that we have
\be \label{eq:phi_1}
\phi_1(\vec{x}) = \sqrt{\frac{4}{N_1(\Lambda^*) V}}    \sum_{\vec{k} \in \Lambda^*_{+,1} } \cos (2\pi \vec{k} \cdot \vec{x}).
\ee
Using this in \eqref{eq:cii} gives
\be
\frac{1}{N_1^2(\Lambda^*)}  \sum_{\vec{k}\in S_k(\Lambda^*)} \left[  \sqrt{2 V} \sum_{\vec{k}_i, \vec{k}_j \in \Lambda^*_{+,1} } \langle \phi_{\vec{k}_i}^+\phi_{\vec{k}_j}^+,\phi_{\vec{k}}^+ \rangle\right]^2   \leq  \max_{\mathbb{R}^d/\Lambda} {c}^2_{11}(\sigma_k) ,
\ee
where we replaced the sum over $ \Lambda^*_{+, k }$ by half the sum over $S_k(\Lambda^*)$.
By \eqref{eq:cijk_+++}, we have
\begin{align}
\sqrt{2 V} \sum_{\vec{k}_i, \vec{k}_j \in \Lambda^*_{+,1} } \langle \phi^+_{\vec{k}_i}  \phi^+_{\vec{k}_j}, \phi^+_{\vec{k}} \rangle & =  \sum_{\vec{k}_i, \vec{k}_j \in \Lambda^*_{+,1} }  \left(\delta_{\vec{k}_i+\vec{k}_j-\vec{k}}+\delta_{\vec{k}_i-\vec{k}_j+\vec{k}}+\delta_{-\vec{k}_i+\vec{k}_j+\vec{k}} +\delta_{\vec{k}_i+\vec{k}_j+\vec{k}} \right) \nonumber \\
& =  \sum_{\vec{k}_i, \vec{k}_j \in S_1(\Lambda^*) }  \delta_{\vec{k}_i+\vec{k}_j+\vec{k}}.
\end{align}
Using Lemma~\ref{lemma:S_k} and the definition \eqref{eq:RMS_def}, we get
\be
\frac{1}{N_1^2(\Lambda^*)}  \sum_{\vec{k}\in S_k(\Lambda^*)} |S_1(\Lambda^*) \cap S_1(\vec{k}, \Lambda^*)|^2 = \frac{N_{k}(\Lambda^*) \nu^2_k(\Lambda^*)}{N_1^2(\Lambda^*)}     \leq  \max_{\mathbb{R}^d/\Lambda} {c}^2_{11}(\sigma_k),
\ee
which completes the proof.
\end{proof}

For $k=1$, Proposition~\ref{prop:lower_bound} gives
\be \label{eq:tori_lower_bound_k=1}
\frac{ \nu^2_{{1}} (\Lambda^*)  }{N_1(\Lambda^*)}   \leq  \max_{\mathbb{R}^d/\Lambda} {c}^2_{11}(\sigma_1).
\ee
For lattices whose automorphism group acts transitively on the $1$\textsuperscript{st} shell, the quantity $\nu_1$ is sometimes called the ``necking number" \cite{ElkiesLectureNotes}; in terms of the lattice sphere packing, this is the number of spheres that simultaneously touch two spheres that are touching, i.e., the number of common neighbors of two touching spheres. As a simple example, in the hexagonal lattice the first shell contains six minimal vectors and the automorphism group, which is the dihedral group of order 12, acts transitively on these vectors. As illustrated in the figure below, there are two circles simultaneously touching any two touching circles in the hexagonal packing and so $\nu_1(A_2)=2$.
\begin{figure}[h]
\centering
\begin{tikzpicture}[scale=2.0]

\coordinate (O) at (0,0);
\coordinate (P1) at ( 1, 0);
\coordinate (P2) at ( 0.5, 0.866025);
\coordinate (P3) at (-0.5, 0.866025);
\coordinate (P4) at (-1, 0);
\coordinate (P5) at (-0.5,-0.866025);
\coordinate (P6) at ( 0.5,-0.866025);

\draw[thick] (P1)--(P2)--(P3)--(P4)--(P5)--(P6)--cycle;
\draw[thick] (O)--(P3);
\draw[thick] (O)--(P4);
\draw[thick] (O)--(P5);
\draw[thick] (O)--(P1);
\draw[thick] (O)--(P2);
\draw[thick] (O)--(P6);

\node[circle, fill=black, inner sep=1.2pt] at (O) {};
\node[circle, fill=black, inner sep=1.2pt] at (P1) {};
\node[circle, fill=black, inner sep=1.2pt] at (P2) {};
\node[circle, fill=black, inner sep=1.2pt] at (P3) {};
\node[circle, fill=black, inner sep=1.2pt] at (P4) {};
\node[circle, fill=black, inner sep=1.2pt] at (P5) {};
\node[circle, fill=black, inner sep=1.2pt] at (P6) {};

\draw[thick] (O)  circle[radius=0.5];
\draw[thick, red] (P1) circle[radius=0.5];
\draw[thick] (P2) circle[radius=0.5];
\draw[thick, red] (P3) circle[radius=0.5];

\end{tikzpicture}

\end{figure}
For the torus $\mathbb{R}^2/A_2^*$, we thus have
\be \label{A2_example}
\frac{2^2}{6} =\frac{2}{3} \leq \max_{\mathbb{R}^2/A_2^*} {c}^2_{11}(\sigma_1).
\ee
We will see below that $2/3$ is also an upper bound on $ {c}^2_{11}(\sigma_1)$ for any 2-dimensional torus.

As other examples, we have $\nu_1=56$ for the $E_8$ lattice and $\nu_1=4600$ for the Leech lattice. For a full-rank lattice in $\mathbb{R}^d$, the set of minimal vectors that are minimal distance from a given minimal vector produces, after rescaling, a spherical code on the unit sphere in $\mathbb{R}^{d-1}$ with minimal angle $\cos^{-1}(1/3)\approx 70.53^{\circ} $ \cite[Theorem~1, Ch.~14]{Conway:1988oqe}. For a $d$-dimensional lattice $\Lambda$, we thus have
\be
\nu_1(\Lambda) \leq A(d-1, \cos^{-1}(1/3)),
\ee 
where $A(d-1, \phi)$ is the maximum number of points in a spherical code on the unit sphere in $\mathbb{R}^{d-1}$ with minimal angle $\phi$. As a corollary, $\nu_1(\Lambda)$ is bounded above by $A(d-1, \pi/3)$, the kissing number in one lower dimension. Linear programming upper bounds on $A(d, \cos^{-1}(1/3))$ are given in Table~9.2 of \cite{Conway:1988oqe}, and these are saturated in some dimensions by spherical codes coming from lattices; for example, $A(7, \cos^{-1}(1/3)) = 56$ and $A(23, \cos^{-1}(1/3)) = 4600$, with the spherical codes coming from the set of spheres touching two touching spheres in the $E_8$ and Leech lattice sphere packings. 

We emphasize that \eqref{eq:tori_lower_bound} is a lower bound on the maximum value of $ {c}^2_{11}(\sigma_k)$ for a given torus $\mathbb{R}^d/\Lambda$. Below we derive bootstrap upper bounds on ${c}^2_{11}(\sigma_k)$ for all flat tori of a given dimension. We will see that some of these upper bounds are saturated by the lower bound \eqref{eq:tori_lower_bound} for particular tori, which implies that \eqref{eq:uniform_cosines} realizes the global maximum of the quartic form for these tori.

\subsubsection{${c}_{11}(\sigma_2)$}
\label{subsubsec:c112}
Let us comment on when ${c}_{11}(\sigma_2)$ can be non-vanishing for flat tori. From \eqref{eq:cii}, \eqref{eq:cijk_+++}, and \eqref{eq:cijk_++-}, ${c}_{11}(\sigma_2)$ can only be nonzero if we have three lattice vectors satisfying $\vec{k}_i+\vec{k}_j+\vec{k}_k=0$ with $k_i^2 = k_j^2=\sigma_1/4 \pi^2 $ and $k_k^2=\sigma_2/4 \pi^2 $. This gives
\be \label{eq:c112_relation}
	|\vec{k}_i+\vec{k}_j|^2 =\vec{k}_k^2 \quad \implies \quad 2 \sigma_1+8 \pi^2  \vec{k}_i \cdot \vec{k}_j=  \sigma_2 .
\ee
By \eqref{eq:distinct_eigenvalues} and \eqref{eq:norm_to_eigenvalue}, there can be no lattice vectors with length squared in the range $(0, \sigma_1/4 \pi^2 )$ or $( \sigma_1/4 \pi^2 , \sigma_2/4 \pi^2 )$, and thus $|\vec{k}_i - \vec{k}_j|^2$ cannot fall in these intervals. Using \eqref{eq:c112_relation}, we have
\be
	4 \pi^2 |\vec{k}_i - \vec{k}_j|^2 =2 \sigma_1 - 8 \pi^2  \vec{k}_i \cdot \vec{k}_j = 4\sigma_1 -\sigma_2.
\ee
So for a flat torus with ${c}_{11}(\sigma_2) > 0$, we must have
\be
	4\sigma_1 -\sigma_2 \notin(0, \sigma_1) \cup (\sigma_1, \sigma_2) \quad \implies \quad \sigma_2/\sigma_1 \notin (2, 3) \cup (3, 4).
\ee

\subsection{Triple product bounds}
We now present our bootstrap bounds and compare them to examples.

\subsubsection{$ c_{11}(\sigma_1)$}

For a fixed dimension $d$, we find an upper bound on ${c}_{11}(\sigma_1)$ for flat tori by considering the conditions \eqref{eq:numconstralpha} with $x_0=1$, which imposes that the smallest nonzero tensor eigenvalues are at least as large as the scalar spectral gap. This bound is shown in the introduction in Fig.~\ref{fig:c11_plot}. We obtain a bound that applies more generally to flat orbifolds by instead taking $x_0=0$. We plot both of these numerical bounds for $d=2, \dots,  31$ in Fig.~\ref{fig:c111}. 
We also show in this plot the quantity $\nu_{1} (\Lambda)/ \sqrt{N_1(\Lambda)}$ for the laminated lattices $\Lambda_d$,\footnote{The laminated lattices are not unique in $d=11, 12, 13,$ and for at least some $d\geq 25$. Our choice of $\Lambda_d$ corresponds to \texttt{Lattice("Lambda",d)} in \texttt{Magma} \cite{MR1484478}. These are the laminated lattices with the largest kissing number given in the catalog \cite{LatticeCatalog}.} which gives a lower bound on $ c_{11}(\sigma_1)$. 
We can see from this plot that the orbifold bound appears to be saturated by $\mathbb{R}^d/\Lambda_d^*$ for $d=2$ and $d=8$, where $\Lambda_2$ is the hexagonal lattice and $\Lambda_8$ is the $E_8$ lattice. Additionally, the bound for flat tori appears to be saturated by $\mathbb{R}^d/\Lambda_d^*$ for $d=4$ and $d=24$, where $\Lambda_4$ is the $D_4$ lattice and $\Lambda_{24}$ is the Leech lattice. We give a rigorous proof in Section~\ref{sec:functionals} that these bounds are indeed saturated in these dimensions. 
\begin{figure}[ht!]
\begin{center}
	\epsfig{file=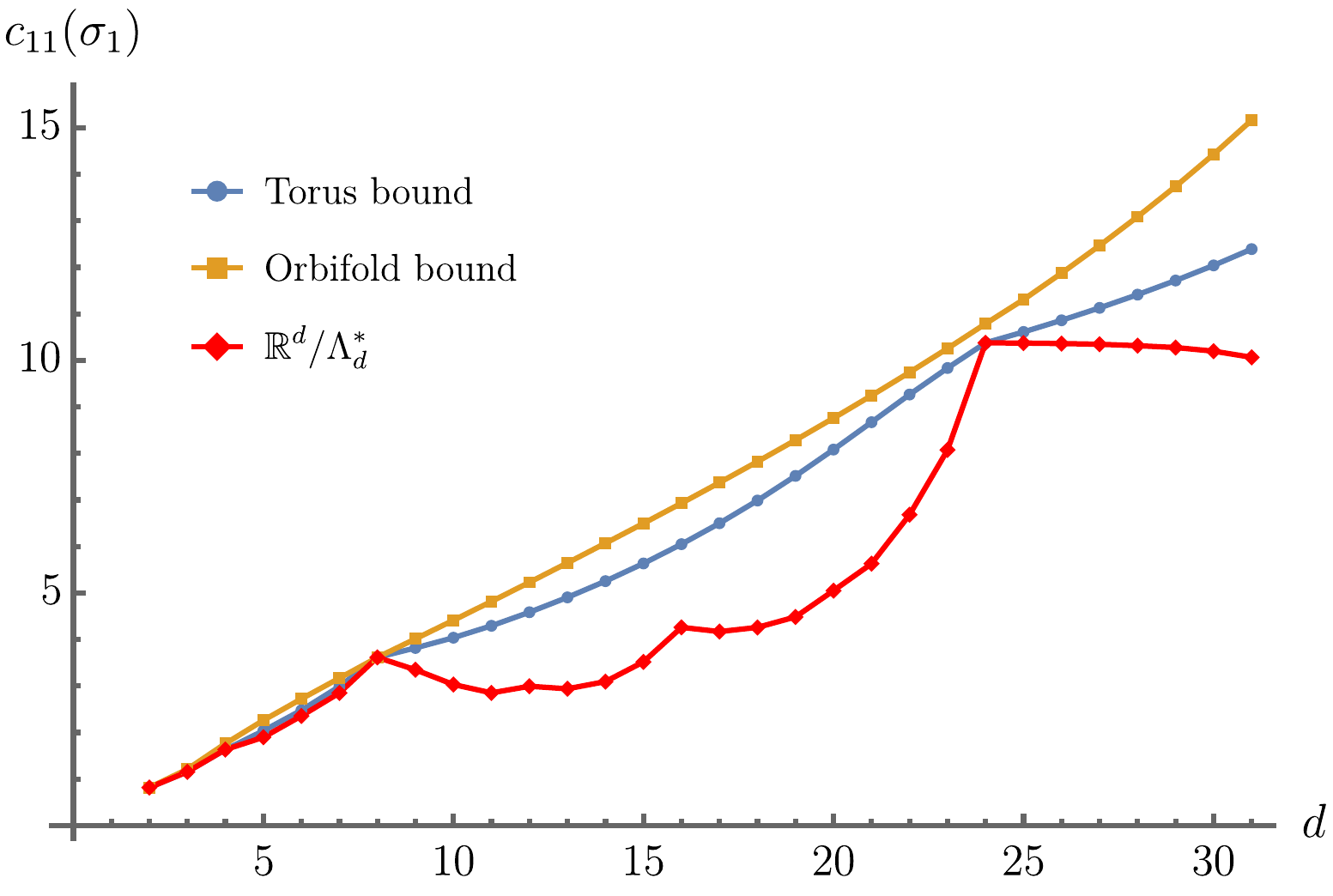,width=12cm}
\end{center}
\caption{This plot shows the bootstrap upper bounds on $c_{11}(\sigma_1)$ for flat tori and flat orbifolds in $d$ dimensions, where we take $\cutoff=42$ for $2 \leq d\leq 24$ and $\cutoff=122$ for $25 \leq d \leq 31$. The red diamonds show $ c_{11}(\sigma_1)$ for the tori $\mathbb{R}^d/\Lambda^*_d$ with the choice of eigenfunction \eqref{eq:phi_1}, so that $c_{11}(\sigma_1)= \nu_{1} (\Lambda_d)/ {\sqrt{N_1(\Lambda_d)}}$, where $\Lambda_d$ is the $d$-dimensional laminated lattice with the largest kissing number given in the catalog \cite{LatticeCatalog}.}
\label{fig:c111}
\end{figure}

We summarize the lower and upper bounds on $\nu^2_{1} (\Lambda)/ {N_1(\Lambda)}$ for lattices in various dimensions in Table~\ref{table:bounds}. The bound in $d=3$ is nearly saturated by $\mathbb{R}^d/\Lambda_3^*$, where $\Lambda_3$ is the $A_3$ (or fcc) lattice, but the numerical upper bound appears to converge slightly above $4/3$. This is reminiscent of the near saturation of hyperbolic orbifold bounds \cite{Bonifacio:2021aqf,Kravchuk:2021akc,Radcliffe:2024jcg}. There are additional spectral identities in 3 dimensions, as considered for the hyperbolic case in \cite{Bonifacio:2023ban}, which may help to bring the bound down to $4/3$. For sufficiently large $\cutoff$, the numerical bounds in $d=5$, $d=6$, and $d=9$ appear to converge to the rational numbers $4096/987$, $350/57$, and $204800/14061$, respectively. 
\begin{table}
\begin{center}
\def\arraystretch{1.125}%
\begin{tabular}{ c | c |c }
 $d$ & $ \nu^2_{1} (\Lambda_d)/ {N_1(\Lambda_d)}$ & Upper bound  \\ \hline
 1 & 0 & 0 \\
 2 & $\frac{2^2}{6}=\frac{2}{3}$ & $\frac{2}{3}$ \\  
 3 & $\frac{4^2}{12} \approx 1.3333$ & 1.3381  \\
 4 &  $\frac{8^2}{24}=\frac{8}{3}$  & $\frac{8}{3}$  \\
 5 & $\frac{12^2}{40} = 3.6 $ & 4.1500 \\
 6 & $\frac{20^2}{72}\approx 5.5556 $ & 6.1418 \\
 7& $\frac{32^2}{126} \approx 8.1270$ & 8.9302 \\
 8 & $\frac{56^2 }{240} =\frac{196}{15}$ & $\frac{196}{15}$ \\
 9 & $\frac{32 \times 28^2+128 \times 56^2+112 \times 60^2}{272^2}\approx 11.2145$ & $14.5751$ \\
16 & $\frac{280^2 }{4320} \approx 18.1481 $ & $36.5804$ \\
24 & $\frac{4600^2}{196560}=\frac{264500}{2457}$ & $\frac{264500}{2457}$\\
\end{tabular}
 \caption{Lower and upper bounds on $\nu^2_{1} (\Lambda)/ {N_1(\Lambda)}$ in different dimensions.}
 \label{table:bounds}
 \end{center}
 \end{table}

Note that the $d=12$ point in the sphere-packing plot in Fig.~\ref{fig:bestpackings} is the Coxeter--Todd lattice $K_{12}$, whereas the $d=12$ point in the bootstrap plot in Fig.~\ref{fig:c11_plot} or Fig.~\ref{fig:c111}  is the laminated lattice $\Lambda_{12}$. The Coxeter--Todd lattice has
\be
\frac{\nu^2_{1} (K_{12})}{ N_1( K_{12} ) }=\frac{82^2}{756} \approx 8.8942,
\ee
whereas the laminated lattice has the slightly larger value
\be
\frac{\nu^2_{1} ( \Lambda_{12} ) }{ N_1( \Lambda_{12} ) }=\frac{\frac{1}{648}(72^2\times 576 +104^2 \times 72)}{648} \approx 8.9657,
\ee
where the two terms in the numerator correspond to the two orbits of the automorphism group acting on the first shell.

\subsubsection{$c_{11}(\sigma_2)$}
Now we consider bounds on the quantity $ c_{11}(\sigma_2)$, which involves the overlaps of eigenfunctions from the first and second eigenspaces, using the spectral identities \eqref{eq:sumrulesFin}. We compute an upper bound on $ c_{11}(\sigma_2)$ for different values of the ratio of the smallest distinct nonzero scalar eigenvalues, $\sigma_2/\sigma_1$. This ratio is restricted to lie in the interval $(1,4]$ for Ricci-flat manifolds \cite{Bonifacio:2019ioc}.  
Let us fix $x^* \in [1, 4]$, an integer dimension $d\geq 2$, and a positive even integer $\cutoff$. Suppose that we have a vector $\vec\alpha\in\mathbb{R}^{n_\cutoff}$ such that the following conditions hold:
\begin{subequations} \label{eq:c112_conditions}
\begin{align}
		\vec{\alpha}\cdot\vec{F}_{0}(x^*) &=1,  \\
		\vec{\alpha}\cdot\vec{F}_0(x) &\geqslant0, \quad  \forall x \in\{1\}\cup[x^*,\infty), \\
		\vec{\alpha}\cdot \vec A_s &\geqslant0, \quad  s=2,4, \dots, 2 \lfloor \cutoff/4 \rfloor, \\
		\vec{\alpha}\cdot\vec{F}_s(x) &\geqslant0, \quad  \forall x\geq0, \quad  s =2, 4, \dots, 2 \lfloor \cutoff/4 \rfloor. \label{eq:tensor_gap_c112}
\end{align}
\end{subequations}
Taking the dot product of \eqref{eq:sumrulesFin} with $\vec\alpha$ and using \eqref{eq:c112_conditions} together with \eqref{eq:c_reality} then yields the inequality 
\be
	 c_{11}(x^* \sigma_1) \leqslant \sqrt{-\vec\alpha\cdot\vec{A}_0}.
\ee
We thus obtain an upper bound on $c_{11}(\sigma_2)$ for flat orbifolds with $\sigma_2/\sigma_1 =x^*$. If we replace \eqref{eq:tensor_gap_c112} by the more restrictive conditions
\be\label{eq:tensor_gap_c112_tori}
\vec{\alpha}\cdot\vec{F}_s(x) \geqslant0, \quad  \forall x \in\{1\}\cup[x^*,\infty), \quad  s =2, 4, \dots, 2 \lfloor \cutoff/4 \rfloor,
\ee
then we obtain potentially stronger bounds that apply to all flat tori. We show the resulting numerical bounds for $d=3$, $4$, and $5$ in Figs.~\ref{fig:c1123d}, \ref{fig:c1124d}, and \ref{fig:c1125d}. In these plots, we also include examples of flat tori $\mathbb{R}^d/\Lambda^*$ where $\Lambda$ ranges over: 1) all named lattices for $d=3,4, 5$ in the catalog \cite{LatticeCatalog}, and 2) all lattices defined by primitive positive-definite quadratic forms with absolute values of their discriminants bounded by 1000 for $d=3$, by 992 for $d=4$, and by 300 for $d=5$, which are enumerated in \cite{LatticeCatalog}.
For each of these examples, we numerically searched for the maximum of the quartic form \eqref{eq:cii}, rather than using the lower bound \eqref{eq:tori_lower_bound}. 
\begin{figure}[ht!]
	\begin{center}
		\epsfig{file=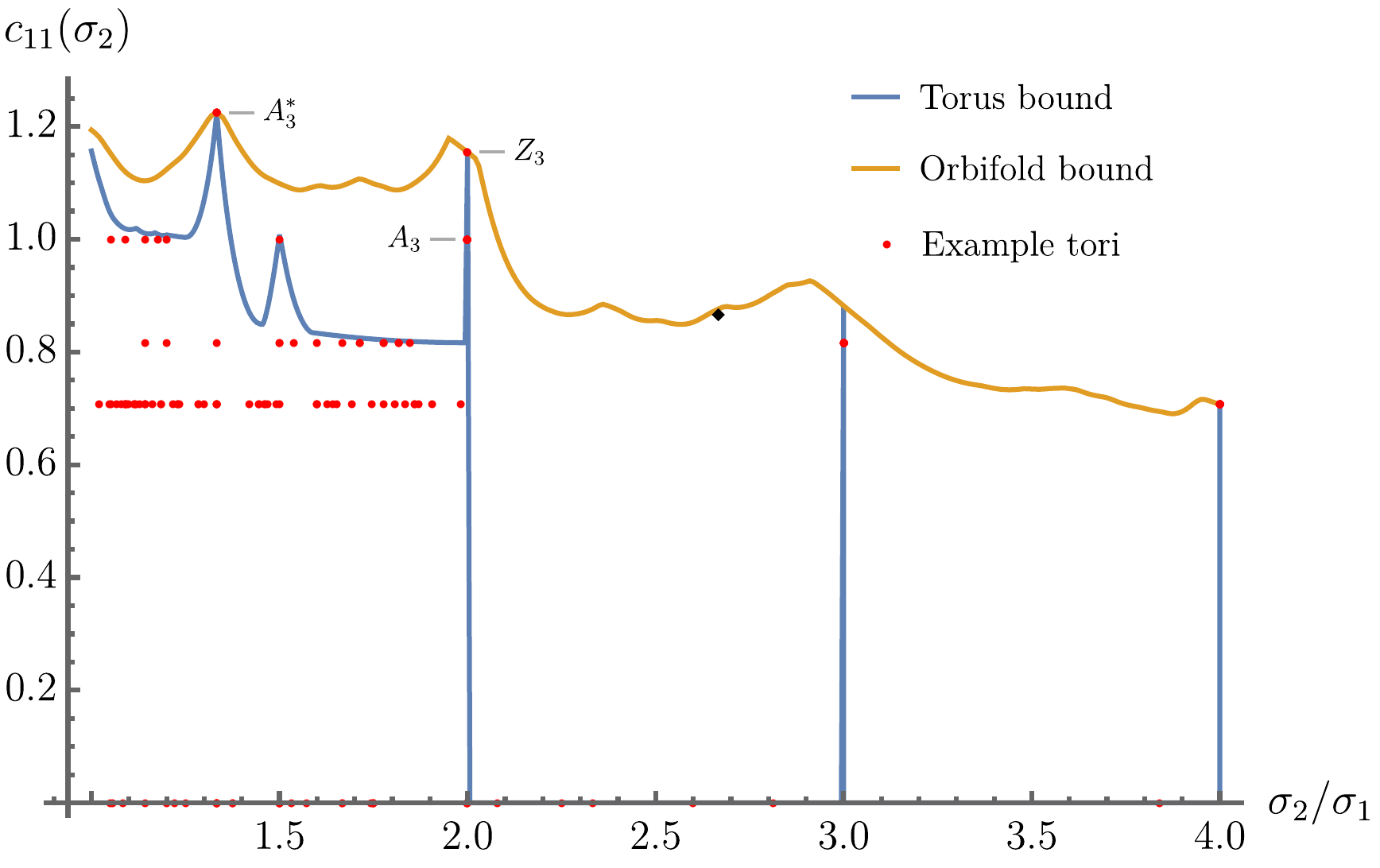,width=12cm}
	\end{center}
	\caption{This plot shows the bootstrap bounds on $ c_{11} (\sigma_2)$ for different values of $\sigma_2/\sigma_1$ for flat tori and orbifolds with $d=3$, where we take $\cutoff=50$. The red dots are examples of tori $\mathbb{R}^3/\Lambda^*$, labeled by $\Lambda$. The black diamond marker is the flat 3-orbifold defined by the space group Fd$\bar{3}$m, which is the symmetry group of the diamond cubic crystal structure.}
	\label{fig:c1123d}
\end{figure}

 We can see from these plots that the tori bounds vanish for $\sigma_2/\sigma_1\in (2,3)\cup (3,4)$, which is consistent with the fact that $c_{11}(\sigma_2)$ must vanish for tori with eigenvalue ratios in these ranges, as explained in Section~\ref{subsubsec:c112}. There are also several examples that are close to the bounds. Let us mention further details for some of these examples in $d=3$, taking $\cutoff=50$ unless stated otherwise:
 \begin{itemize}
 \item The torus $\mathbb{R}^3/\Lambda_3$, where $\Lambda^*_3$ is the $A^*_3$ (or bcc) lattice, has $\sigma_2/\sigma_1=4/3$ and $c^2_{11}(\sigma_2) =3/2$, which is very close to the numerical orbifold bound.  
 \item The lattice $\Lambda$ with Gram matrix 
 \be
	G =\begin{pmatrix}
		4 && 0 && 2 \\ 0 && 4 && 2 \\ 2 && 2 && 5
	\end{pmatrix}
\ee
defines a torus $\mathbb{R}^3/\Lambda$ with $\sigma_2/\sigma_1 =3/2$ and $ {c}_{11}(\sigma_2) =1$. It defines an odd ternary quadratic form $q_G$ with discriminant $-192$ \cite{LatticeCatalog}. The numerical upper bound on $ {c}_{11}(\sigma_2)$ for tori with $\sigma_2/\sigma_1 =3/2$ is $1.0083$, which forms a sharp kink that is nearly saturated by this example.
\item The orbifold $\mathbb{R}^3/\Gamma$ with $\Gamma$ the space group Fd$\bar{3}$m, which is the symmetry group of the diamond cubic crystal structure, has $\sigma_2/\sigma_1 =8/3$ and $ {c}^2_{11}(\sigma_2) =3/4$. The numerical orbifold bound for $\sigma_2/\sigma_1=8/3$ with $\cutoff =100$ is $c^2_{11}(\sigma_2)  \leq 0.7608$. If we impose that the tensors have a spectral gap at least as large as the scalar gap, by replacing \eqref{eq:tensor_gap_c112} with $\vec{\alpha}\cdot\vec{F}_s(x) \geqslant0$,  $\forall x\geq1$, then the numerical bound is very close to $3/4$. 
\item The value $ {c}^2_{11}(\sigma_2) =1/2$ with $\sigma_2/\sigma_1 =4$ is realized by any rectangular torus with one cycle at least twice as long as the others, e.g., the torus $\mathbb{R}^3/\Lambda$ with $\Lambda = \{(2n_1,n_2, n_3): \, n_i \in \mathbb{Z} \}$.
\end{itemize}
\begin{figure}[ht!]
\begin{center}
	\epsfig{file=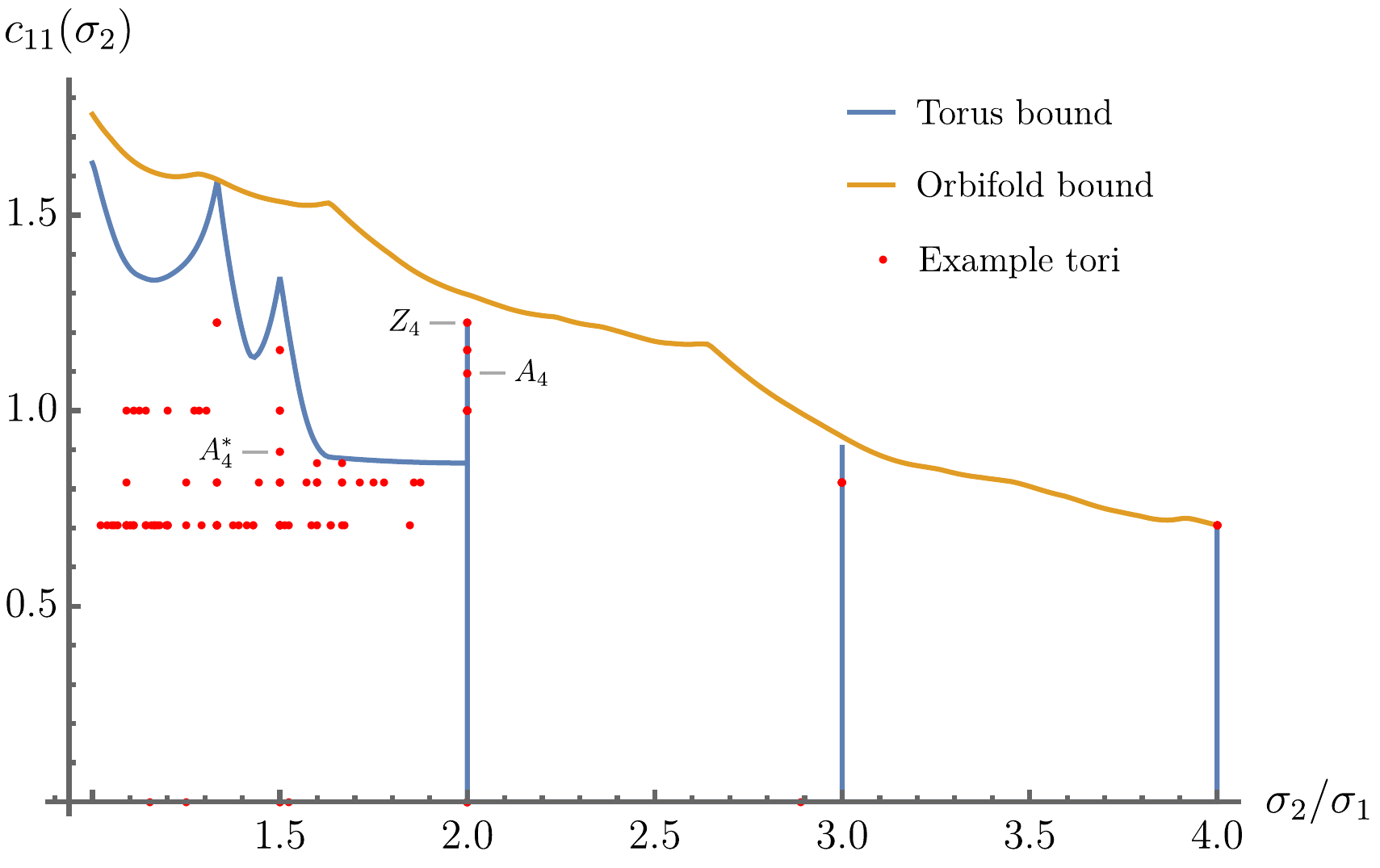,width=12cm}
\end{center}
	\caption{This plot shows the bootstrap bounds on $ c_{11} (\sigma_2)$ for different values of $\sigma_2/\sigma_1$ for flat tori and orbifolds with $d=4$, where we take $\cutoff=50$. The red dots are examples of tori $\mathbb{R}^4/\Lambda^*$, labeled by $\Lambda$.}
\label{fig:c1124d}
\end{figure}

\begin{figure}[ht!]
\begin{center}
	\epsfig{file=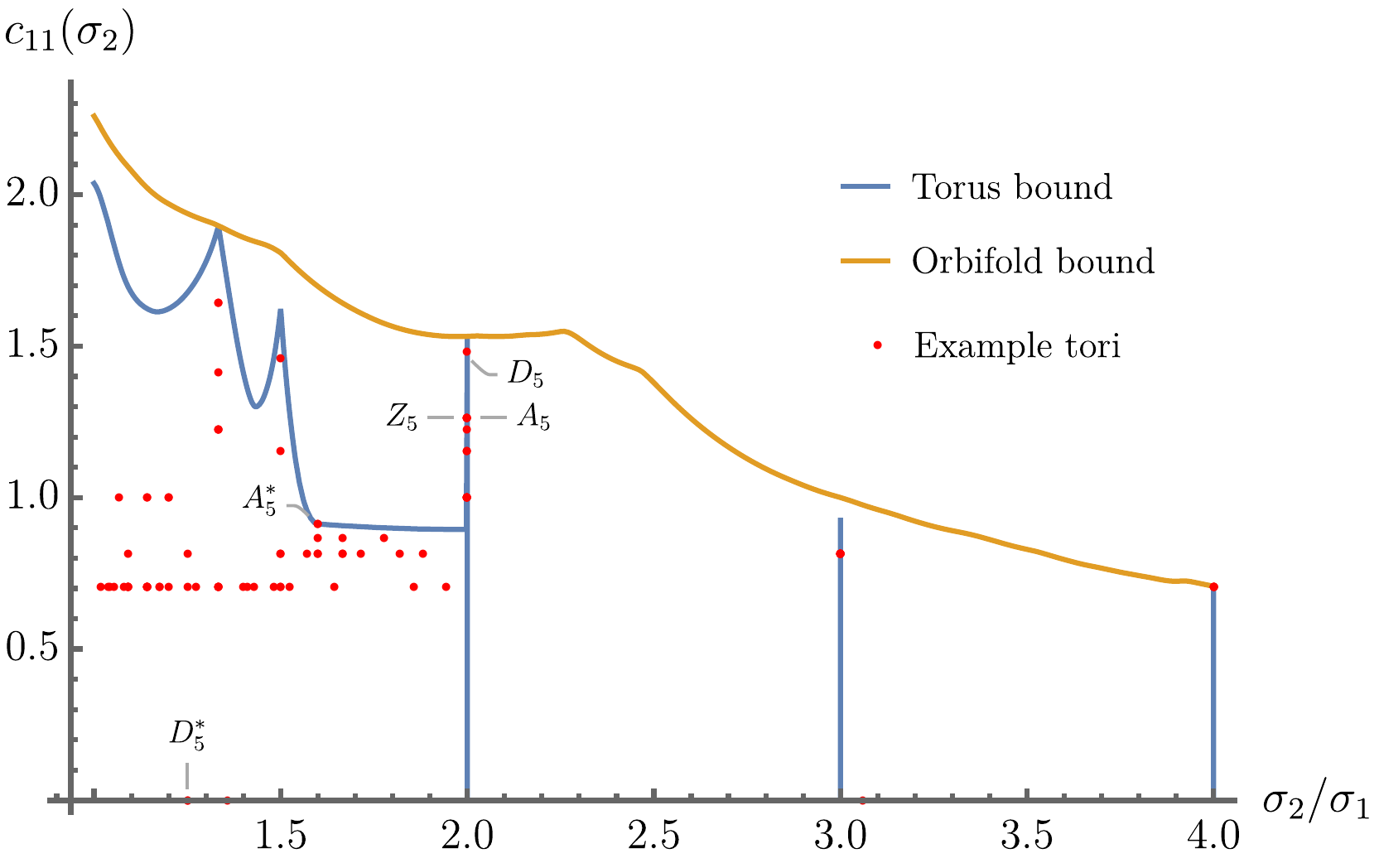,width=12cm}
\end{center}
	\caption{This plot shows the bootstrap bounds on $c_{11} (\sigma_2)$ for different values of $\sigma_2/\sigma_1$ for flat tori and orbifolds with $d=5$, where we take $\cutoff=62$. The red dots are examples of tori $\mathbb{R}^5/\Lambda^*$, labeled by $\Lambda$.}
\label{fig:c1125d}
\end{figure}

We do not include plots of the bounds on ${c}_{11}(\sigma_2)$ for $d>5$, but let us mention some results for $d=8$ and $d=24$. In 8 dimensions, the orbifold upper bound for ${c}_{11}(\sigma_2)$ with $\sigma_2/\sigma_1=2$ is given by 
\be \label{eq:E8_bound_2}
{c}^2_{11}(\sigma_2) \leq \frac{147}{20},  \quad \frac{\sigma_2}{\sigma_1}= 2.
\ee
This is saturated by the $E_8$ lattice\footnote{Here is some numerology: the number of distinct nonzero monomials in the fully expanded quartic form \eqref{eq:cii} for the $E_8$ lattice with $k=2$ is $\frac{240\times126}{2} (1 + (14-2))=196560$, i.e., the kissing number of the Leech lattice.} with $\phi_1$ given by \eqref{eq:phi_1}, since
\be \label{eq:E8_c112}
 \frac{N_2(E_8) \nu_2^2(E_8)}{N^2_1(E_8)}=  \frac{2160 \times 14^2}{240^2}=\frac{147}{20}   .
\ee
In 24 dimensions, the torus upper bound for ${c}_{11}(\sigma_2)$ with $\sigma_2/\sigma_1=3/2$ is given by 
\be \label{eq:Leech_bound_2}
{c}^2_{11}(\sigma_2) \leq \frac{541696}{4095}, \quad \frac{\sigma_2}{\sigma_1}=\frac{3}{2}.
\ee
This is saturated by the Leech lattice with $\phi_1$ given by \eqref{eq:phi_1}, since
\be \label{eq:Leech_c112}
 \frac{N_2(\Lambda_{24}) \nu_2^2(\Lambda_{24})}{N^2_1(\Lambda_{24})}=   \frac{16773120 \times 552^2}{196560^2} =\frac{541696}{4095} .
\ee

\subsubsection{$c_{11}^{(2)}(\sigma_1)$}
The final bootstrap bound we consider, shown in Fig.~\ref{fig:c111rank2}, is a numerical upper bound on  ${ c}^{(2)}_{11} (\sigma_1)$ for tori with $\cutoff = 50$ and $3 \leq d \leq 10$. This quantity comes from the spinning overlaps of the first non-constant scalar eigenfunction with the first non-constant divergence-free traceless rank-$2$ symmetric eigentensors. For these bounds, we maximize the objective function $\vec\alpha\cdot\vec{A}_0$ subject to the following conditions:
\begin{subequations}
\begin{align}
		\vec{\alpha}\cdot\vec{F}_{2}(1) &=1, \\
		\vec{\alpha}\cdot \vec A_s &\geqslant 0, \quad  s=2,4, \dots, 2 \lfloor \cutoff/4 \rfloor, \\
		\vec{\alpha}\cdot\vec{F}_s(x) &\geqslant0, \quad  \forall x\geqslant 1, \quad  s =0, 2, \dots, 2 \lfloor \cutoff/4 \rfloor,
\end{align}
\label{eq:spin2_conditions}
\end{subequations}
which gives the torus bound 
\be
{ c}_{11}^{(2)}\left(\sigma_1\right) \leq \sqrt{-\vec\alpha\cdot\vec{A}_0}.
\ee
We also include in Fig.~\ref{fig:c111rank2} the torus in each dimension with the largest numerically maximized value of ${c}^{(2)}_{11} (\sigma_1)$ that we found. The bound in $d=3$ is ${ c}^{(2)}_{11} (\sigma_1)\leq 1/\sqrt{3}$ and this is saturated by $\mathbb{R}^3/A_3^*$ with, e.g., the following choice for $\phi_1$:
\begin{align}
\phi_1 =&  \frac{2}{\sqrt{6}}\bigg[ \sin \left(\pi  \left(4 x_1-2 \sqrt{2} x_2\right)/\sqrt{3}\right)+\sin
   \left(2 \sqrt{2} \pi  x_3\right)+\cos \left(\sqrt{2} \pi  \left(x_3-\sqrt{3}
   x_2\right)\right)+\cos \left(\sqrt{2} \pi  \left(\sqrt{3} x_2+x_3\right)\right) \nonumber \\
   &-\cos
   \left(\pi/3  \left(4 \sqrt{3} x_1+\sqrt{6} x_2-3 \sqrt{2}
   x_3\right)\right)+\cos \left[ \pi/3  \left(4 \sqrt{3} x_1+\sqrt{6} x_2+3
   \sqrt{2} x_3\right) \right] \bigg].
\end{align}
\begin{figure}[ht!]
	\begin{center}
		\epsfig{file=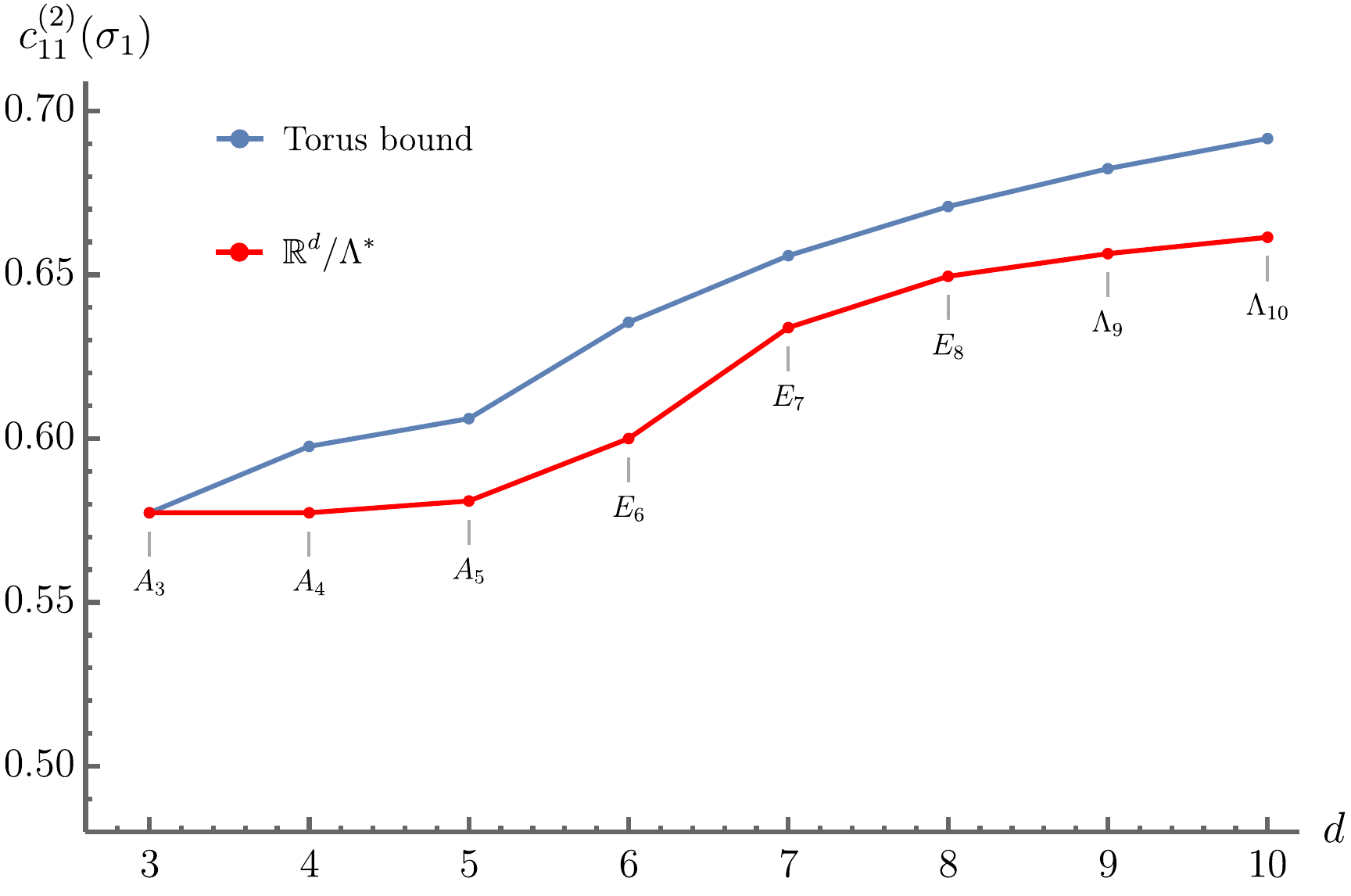, width=11.5cm}
	\end{center}
	\caption{This plot shows the bootstrap upper bounds on ${ c}^{(2)}_{11} (\sigma_1)$ for flat tori with $3 \leq d\leq 10$, taking  $\cutoff=50$. The red dots show $ {c}^{(2)}_{11}(\sigma_1)$ for tori $\mathbb{R}^d/\Lambda^*$, labeled by $\Lambda$.}
	\label{fig:c111rank2}
\end{figure}

\subsection{Exact functionals}

\label{sec:functionals}

Some of the numerical bounds discussed above are nearly saturated by particular tori defined by special Euclidean lattices. In certain dimensions, we can prove that the bounds are saturated at finite $\cutoff$ by finding an explicit functional, which is just a finite-dimensional rational vector. These functionals are thus finite-dimensional analogs
of the sphere packing magic functions. We show that the orbifold bound on $c_{11}(\sigma_1)$ is saturated by $\mathbb{R}^d/\Lambda_d^*$ for  $d=2$ and $d=8$, where $\Lambda_2$ is the $A_2$ lattice and $\Lambda_8$ is the $E_8$ lattice, while the torus bound on $c_{11}(\sigma_1)$ is saturated by $\mathbb{R}^d/\Lambda_d^*$ in $d=4$ and $d=24$, where $\Lambda_4$ is the $D_4$ lattice and $\Lambda_{24}$ is the Leech lattice. In each of these cases, the bound is saturated by taking $\phi_1$ as in \eqref{eq:phi_1}, and therefore Theorem~\ref{thm:main_theorem} follows as a consequence.
We also give functionals for: 1) the $c_{11}(\sigma_2)$ bounds that are saturated by $\mathbb{R}^d/\Lambda_d^*$ for $d=8$ and $d=24$, and 2) the $c^{(2)}_{11}(\sigma_1)$ bound for $d=3$, which is saturated by $\mathbb{R}^d/A_3^*$.

Let us focus on the $c_{11}(\sigma_1)$ bounds and discuss how we can find the exact functionals, which are needed to prove Theorem~\ref{thm:main_theorem}. In order for the inequality \eqref{eq:objineq} to be saturated, all terms on the RHS of \eqref{eq:expanded_sum_rule} must vanish. 
The polynomials $\vec{\alpha} \cdot  \vec{F}_0(x) $ must therefore have double zeros at $x=\sigma_k / \sigma_1$ for $k \geq 2$ whenever $c_{11}(\sigma_k) \neq 0$. 
The polynomials $\vec{\alpha} \cdot  \vec{F}_s(x) $ with even $s>0$ must have zeros at $x=\sigma^{(s)}_k / \sigma_1$ for $k \geq 1$ whenever $c_{11}^{(s)}(\sigma_k^{(s)}) \neq 0$, and these must be double zeros except possibly for $k=1$ when $x_0=1$ in \eqref{eq:tensor_gap}. Similarly, if  ${c}_{11}^{(s)}(0) \neq 0$ for $s>0$, then for saturation we must have $\vec{\alpha}\cdot \vec{A}_s = 0$.

To find an exact functional $\vec{\alpha}$, we find a system of constraints with a unique solution corresponding to a valid functional. 
For the bounds on $c_{11}(\sigma_1)$ for a given $d$ and with $n_\cutoff$ spectral identities, we fix an optimal functional $\vec{\alpha} \in \mathbb{Q}^{n_\cutoff}$ by imposing the following conditions: 
\begin{enumerate}
\item There are double zeros at certain points $x^{(s)}_i$, 
\begin{equation} \label{eq:double_zeros}
		\vec\alpha \cdot \vec F_s \left(x^{(s)}_i \right) =0, \quad 	\frac{d}{d x} \left( \vec\alpha \cdot \vec F_s(x) \right)\bigg|_{x=x^{(s)}_i} = 0.
\end{equation}
\item There are simple zeros $\vec\alpha \cdot \vec F_s \left(1\right) =0$ for certain values of $s$.
\item The polynomials $\vec\alpha\cdot \vec F_s (x)$ vanish for certain values of $s$. 
\item The constants $\vec\alpha\cdot \vec A_s$ vanish for certain values of $s$. 
\item The normalization condition  $\vec{\alpha}\cdot\vec{F}_{0}(1) =1$.
\item The bound \eqref{eq:objineq} is saturated by a particular lattice $\Lambda$, using the lower bound from \eqref{eq:tori_lower_bound_k=1},
\be
\vec\alpha \cdot \vec A_0 =- \frac{ \nu^2_{{1}} (\Lambda)  }{N_1(\Lambda)} .
\ee
\item If the above conditions do not fully fix a functional, then we also impose positivity of $\vec\alpha \cdot \vec F_s(x)$ between the double zeros.
\end{enumerate}
Once we have a candidate functional $\vec{\alpha} \in \mathbb{Q}^{n_\cutoff}$, we can rigorously check the non-negativity conditions in \eqref{eq:numconstralpha} with rational arithmetic using the algorithm given in Appendix B of \cite{Kravchuk:2021akc}. We include these $\vec{\alpha}$ as ancillary files, together with an implementation of the non-negativity checks in the \texttt{Mathematica} notebook \texttt{check\_polynomials.nb}. These verified functionals are sufficient to complete the proof of Theorem \ref{thm:main_theorem}:
\begin{proof}[Proof of Theorem \ref{thm:main_theorem} given exact functionals]
For $d=2,4,8,24$, we have an $\vec{\alpha} \in \mathbb{Q}^{n_\cutoff}$ satisfying conditions \eqref{eq:numconstralpha} with $x_0=1$ (or, more strongly, $x_0=0$) for some $\cutoff$ and with 
\be
\vec{\alpha} \cdot \vec{A}_0 = - \frac{ \nu^2_{{1}} (\Lambda_d)  }{N_1(\Lambda_d)} .
\ee
By Proposition~\ref{prop:c11_bound}, we thus have
\be
\sup_{\Lambda\subset \mathbb{R}^d} \max_{\mathbb{R}^d/\Lambda} c^2_{11}(\sigma_1) \leq \frac{ \nu^2_{{1}} (\Lambda_d)  }{N_1(\Lambda_d)},
\ee
where the supremum is over all full rank $d$-dimensional lattices $\Lambda$. Combining this with Proposition~\ref{prop:lower_bound} with $k=1$ applied to $\mathbb{R}^d/\Lambda^*_d$, we have
\be
\frac{ \nu^2_{{1}} (\Lambda_d)  }{N_1(\Lambda_d)}   \leq  \max_{\mathbb{R}^d/\Lambda^*_d} {c}^2_{11}(\sigma_1) \leq \sup_{\Lambda\subset \mathbb{R}^d} \max_{\mathbb{R}^d/\Lambda} c^2_{11}(\sigma_1) \leq \frac{ \nu^2_{{1}} (\Lambda_d)  }{N_1(\Lambda_d)},
\ee
which implies that the bootstrap bound is saturated for $d=2,4,8,24$, 
\be
\sup_{\Lambda\subset \mathbb{R}^d} \max_{\mathbb{R}^d/\Lambda} c^2_{11}(\sigma_1) = \frac{ \nu^2_{{1}} (\Lambda_d)  }{N_1(\Lambda_d)}.
\ee
Applying Proposition~\ref{prop:lower_bound} again, for any full rank lattice $\Lambda \subset \mathbb{R}^d$ with $d=2,4,8,24$, we have 
\be 
\frac{ \nu^2_{{1}} (\Lambda)  }{N_1(\Lambda)}   \leq  \max_{\mathbb{R}^d/\Lambda^*} {c}^2_{11}(\sigma_1) \leq \frac{ \nu^2_{{1}} (\Lambda_d)  }{N_1(\Lambda_d)},
\ee
which gives Theorem~\ref{thm:main_theorem}.
\end{proof}

We now discuss each of the functionals and their associated bounds.

\subsubsection{Hexagonal lattice}
The bootstrap upper bound on $ c_{11}(\sigma_1)$ for flat orbifolds in $d=2$ is
\be \label{eq:A2_bound}
c_{11}^2(\sigma_1) \leq \frac{2}{3} .
\ee
As shown in  \eqref{A2_example}, this is saturated by $\mathbb{R}^2/A^*_2$ with $\phi_1$ given by \eqref{eq:phi_1}.

To prove \eqref{eq:A2_bound}, we take $\cutoff =10$ and consider $n_{\cutoff} = 5$ spectral identities in 2D. These can be written in the form \eqref{eq:sumrulesFin} with
\begin{align}
\vec{A}_0  =& \left\{-1,-\frac{1}{2},-1,-\frac{5}{8},-1\right\}, \\
\vec{A}_2 =&\{0,1,0,1,0\}, \\
\vec{A}_4  =& \{0,0,0,1,0\}, \\
\vec{F}_0(x)  =&\Bigg\{\frac{(3 x-4)}{4} ,-\frac{x(x-4)}{8} ,\frac{3 x^3-18 x^2+36 x-16}{16},-\frac{x (x-4) \left(x^2-4
   x+8\right)}{32} , \nonumber \\
   & \, \, \frac{3 x^5-30 x^4+120 x^3-240 x^2+240 x-64}{64} \Bigg\},
\end{align}
and $\vec{F}_s=0$ for $s>0$.
This example is simple enough that we can obtain a rational functional by inspecting the \texttt{SDPB} output,
\be
\vec{\alpha} = \left\{-\frac{20}{27},-\frac{28}{9},\frac{4}{27},\frac{32}{9},\frac{16}{27}\right\}.
\ee
This satisfies
\be
\vec{\alpha} \cdot \vec{A}_0  = - \frac{2}{3}, \quad \vec{\alpha} \cdot \vec{A}_2 =  \frac{4}{9},  \quad \vec{\alpha} \cdot \vec{A}_4 =  \frac{32}{9}, \quad \vec{\alpha} \cdot \vec{F}_0(x)=  \frac{ x(x-4)^2 (x-3)^2}{36}, \label{eq:hexagonal_poly}
\ee
so the conditions \eqref{eq:numconstralpha} are satisfied and thus by Proposition~\ref{prop:c11_bound} we have the bound \eqref{eq:A2_bound}. 

We can understand the location of the double zeros of the polynomial $\vec{\alpha} \cdot \vec{F}_0(x)$ from the spectrum and overlaps of the hexagonal torus. If we choose the scaling of the lattice so that the minimal vectors have length squared equal to 2, then $\sigma_1=8 \pi^2$, $\sigma_2 =24 \pi^2$, and $\sigma_3 =32 \pi^2$.
For $\phi_1$ given by \eqref{eq:phi_1},
we have the non-vanishing overlaps
\begin{align}
c_{11}(\sigma_1) & = \sqrt{\frac{2}{3}}, \quad  c_{11}(\sigma_2) = \sqrt{\frac{2}{3}}, \quad  c_{11}(\sigma_3) = \sqrt{\frac{1}{6}}.
\end{align}
The double zeros of $ \vec{\alpha} \cdot \vec{F}_0(x)$ occur at $x=\sigma_k/\sigma_1$ when $x>1$ and the overlap $c_{11}(\sigma_k) $ is non-vanishing, as required for saturation of the bound. Similarly, since $\vec{\alpha} \cdot \vec{A}_2 \neq 0$ and $\vec{\alpha} \cdot \vec{A}_4 \neq 0$ in \eqref{eq:hexagonal_poly}, the zero mode overlaps must vanish for $s=2$ and $s=4$ for the hexagonal lattice to saturate the bound. Using \eqref{eq:2D_projector_contraction}, we have for $s>0$ 
\begin{align}
 \left({ c}_{11}^{(s)}(0)\right)^2  & = \frac{1+2\cos \left( \pi s/3 \right) }{ 2^{s-1} 3}   ,
\end{align}
which indeed vanishes when $s=2$ and $s=4$. This is a consequence of the fact that the minimal vectors of the $A_2$ lattice form a spherical $5$-design.

\subsubsection{$E_8$ lattice}
The orbifold bound on $ c_{11}(\sigma_1)$ for $d=8$ is 
\be \label{eq:E8_bound}
c_{11}^2(\sigma_1) \leq \frac{196}{15} .
\ee
This bound is saturated by the $E_8$ lattice with $\phi_1$ given by \eqref{eq:phi_1}, since we have
\be
 \frac{\nu_1^2(E_8)}{N_1(E_8)}= \frac{56^2 }{240} = \frac{196}{15} .
\ee

To prove \eqref{eq:E8_bound}, we take $\cutoff=22$ and consider $n_{\cutoff} = 26$ spectral identities. We can fix the functional $\vec{\alpha} \in \mathbb{Q}^{26}$ by imposing
\begin{align}
\vec{\alpha} \cdot \vec{A}_2 & = \vec{\alpha} \cdot \vec{A}_6= \vec{\alpha} \cdot \vec{A}_8= \vec{\alpha} \cdot \vec{A}_{10}=0, \label{eq:zero_mode_functional} \\
\vec\alpha\cdot \vec F_6 (x) & =\vec\alpha\cdot \vec F_8 (x)=\vec\alpha\cdot \vec F_{10} (x)=0,
\end{align}
together with the double zeros \eqref{eq:double_zeros} at  $x_1^{(0)}  =2$, $x_2^{(0)}  =3$, $x_3^{(0)} =4$, $x_1^{(2)}  =3$, $x_1^{(4)}  =2$, $x_2^{(4)}  =3$, the normalization condition $\vec{\alpha}\cdot\vec{F}_{0}(1) =1$, and the condition $\vec\alpha \cdot \vec A_0 =-196/15$. 
The resulting vector $\vec \alpha$ is included in the ancillary file \texttt{alpha\_E8.txt}. The nonzero polynomials $\vec{\alpha}\cdot \vec{F}_s (x)$ are
{\small
\begin{align}
\vec\alpha\cdot \vec{F}_0 (x) & = \frac{(x-4)^2 (x-3)^2 (x-2)^2 \left(48621 x^5-31938 x^4+82565 x^3+9696 x^2+816 x-49280\right)}{2177280}, \\
\begin{split}
\vec\alpha\cdot \vec{F}_2 (x) & =\frac{(x-3)^2}{142560} (290873 x^7-2148846 x^6+5538593 x^5-6165780 x^4 \\
& \qquad \qquad \qquad \qquad \qquad \qquad +4066020 x^3-2154016 x^2+960704 x+17920), 
\end{split} \\
\vec\alpha\cdot \vec{F}_4 (x) & =\frac{(x-3)^2 (x-2)^2 \left(28149 x^3-8520 x^2-910 x+1120\right)}{1620}.
\end{align}
}

\noindent We have $\vec\alpha\cdot \vec{F}_0 (x) \geq 0$  for $x \geq 1$ and $\vec\alpha\cdot \vec{F}_s (x) \geq 0$ for $s=2,4$ and $x \geq 0$, which can be checked rigorously using the algorithm from Appendix B of \cite{Kravchuk:2021akc} after factoring out the double roots. We also have $\vec{\alpha} \cdot \vec{A}_4=224/9 >0$. Thus all of the conditions in \eqref{eq:numconstralpha} with $x_0=0$ are satisfied and the orbifold bound \eqref{eq:E8_bound} follows by Proposition~\ref{prop:c11_bound}.

The positions of the double zeros can be understood from the nonzero values of $c^{(s)}_{11}(\sigma_k)$ for the $E_8$ lattice with $\phi_1$ given by \eqref{eq:phi_1}. We have
\begin{align}
 c_{11}(\sigma_1) & = \frac{14}{\sqrt{15}}, \quad c_{11}(\sigma_2) = \frac{7}{2} \sqrt{\frac{3}{5}}, \quad c_{11}(\sigma_3) = \sqrt{\frac{7}{15}}, \quad c_{11}(\sigma_4) = \frac{1}{4 \sqrt{15}}, \\
c_{11}^{(2)}(\sigma_1 ) & =c_{11}^{(2)}(\sigma_2) =c_{11}^{(2)}(\sigma_4) =0, \quad c_{11}^{(2)}(\sigma_3) = \frac{1}{2\sqrt{10}}, \\
c_{11}^{(4)}(\sigma_1 ) & =c_{11}^{(4)}(\sigma_4) =0, \quad c_{11}^{(4)}(\sigma_2) = \frac14 \sqrt{\frac7{10}}, \quad c_{11}^{(4)}(\sigma_3) = \frac1{12} \sqrt{\frac7{55}}, 
\end{align}
where $\sigma_1=8 \pi^2$, $\sigma_2 =16 \pi^2$,  $\sigma_3 =24 \pi^2$,  and $\sigma_4 =32 \pi^2$. We see that $\vec{\alpha}\cdot \vec{F}_s (\sigma_k/\sigma_1)$ vanishes when the overlap $c^{(s)}_{11}(\sigma_k)$ is nonzero (excluding the normalization condition), as is required for saturation.

Using \eqref{eq:zeromode_overlap}, the zero-mode overlaps for the $E_8$ lattice with $\phi_1$ given by \eqref{eq:phi_1} are given by
\begin{align}
\left( { c}_{11}^{(s)}(0)\right)^2 & =  \frac{ s!}{2^s   120  \left(   3 \right)_s }    \left[ C_s^{\left(  3 \right)} \! \left(1 \right)+56 C_s^{\left(  3 \right) \! }\left(1/2 \right)+63 C_s^{\left(  3 \right)}\!\left(0\right) \right],
\end{align}
where we recall that $C^{(3)}_s(x)$ is a Gegenbauer polynomial. These zero-mode overlaps vanish for $s\in \{ 2,4,6,10\}$, which is a consequence of the fact that the minimal vectors of the $E_8$ lattice form a spherical $7 \frac{1}{2}$-design. The spinning zero-mode overlaps thus vanish if there are no nonzero cusp forms for $SL(2, \mathbb{Z})$ with weight $4+s$. 

To prove the bound on $c_{11}(\sigma_2)$ in \eqref{eq:E8_bound_2}, which is also saturated by the $E_8$ lattice \eqref{eq:E8_c112}, we take $\cutoff=18$ and consider $n_{\cutoff} = 18$ spectral identities. We can fix a functional using the same approach as for $c_{11}(\sigma_1)$. The resulting $\vec{\alpha}\in \mathbb{Q}^{18}$ is given in the ancillary file \texttt{alpha\_E8\_2.txt}. With this $\vec{\alpha}$, the conditions \eqref{eq:c112_conditions} are satisfied with $x^*=2$ and $\vec{\alpha}\cdot \vec{A}_0=-147/20$, which proves \eqref{eq:E8_bound_2}.

\subsubsection{Leech lattice}
The torus bound on $ c_{11}(\sigma_1)$ in $d=24$ is
\be \label{eq:Leech_bound}
c_{11}^2(\sigma_1) \leq \frac{264500}{2457} .
\ee
This bound is saturated by the Leech lattice $\Lambda_{24}$ with $\phi_1$ given by \eqref{eq:phi_1}, since we have
\be
 \frac{\nu_1^2(\Lambda_{24})}{N_1(\Lambda_{24})}=\frac{4600^2}{196560}= \frac{264500}{2457}  .
\ee
For the Leech lattice, the ratios $\sigma_k/\sigma_1$ take the values $1, 3/2, 2, 5/2, \dots$. To prove \eqref{eq:Leech_bound}, we consider $82$ spectral identities with $\cutoff=42$.\footnote{Note that $n_{\cutoff} = 84$ for $\cutoff=42$. We use the first 82 spectral identities given in the ancillary files. We can pad $\vec{\alpha}$ with zeros to get an $\vec{\alpha} \in \mathbb{Q}^{84}$ complying with the conditions of Proposition~\ref{prop:c11_bound}.} 
We can fix a functional $\vec{\alpha} \in \mathbb{Q}^{82}$ by imposing
\begin{align}
\vec{\alpha} \cdot \vec{A}_s & =0, \quad s \in \{2, 6, 8, 10, 12, 14, 16, 18, 20\} ,  \label{eq:zero_mode_functional_Leech} \\
\vec\alpha\cdot \vec F_{s} (x) & =0, \quad s \in \{12, 14, 16, 18, 20\},
\end{align}
together with double zeros \eqref{eq:double_zeros} at 
$ x_i^{(0)}  \in \{ 3/2, 2, 5/2, 3, 4 \}$,  $x_i^{(2)}  \in \{ 5/2, 3 \}$,  $x_i^{(4)}  \in \{ 2, 5/2, 3 \}$,  $x_i^{(6)}  \in \{ 3/2, 2, 5/2, 3\}$,  $x_i^{(8)}  \in \{ 3/2, 2, 5/2, 3 \}$, and  $x_i^{(10)}  \in \{ 3/2, 2, 5/2, 3 \}$, 
the simple zeros $\vec\alpha \cdot \vec F_s \left(1\right) =0$ for $s \in \{4,8,10 \}$, the normalization condition $\vec{\alpha}\cdot\vec{F}_{0}(1) =1$, and the condition $\vec\alpha \cdot \vec A_0 =-264500/2457$.  
The resulting $\vec \alpha$ is included in the ancillary file \texttt{alpha\_Leech.txt}. 
The nonzero polynomials $\vec{\alpha}\cdot \vec{F}_s (x)$ for $s \in \{ 0, 2, 4, 6, 8, 10\}$ can be rigorously shown to be non-negative for $x \geq 1$ using rational arithmetic with the algorithm from Appendix B of \cite{Kravchuk:2021akc} after factoring out the double roots, and we also have $\vec{\alpha} \cdot \vec{A}_4 =3680/27>0$. Thus all of the conditions in \eqref{eq:numconstralpha} are satisfied with $x_0=1$ and the torus bound follows by Proposition~\ref{prop:c11_bound}.

Using \eqref{eq:zeromode_overlap}, the zero-mode overlaps for the Leech lattice with external state \eqref{eq:phi_1} are given by
\begin{align} 
\left({c}_{11}^{(s)}(0)\right)^2  & =  \frac{ 2^{-s} s!}{ 98280 \left(   11 \right)_s}  \left[ C_s^{\left( 11 \right)} \!\left(  1\right)+4600\, C_s^{\left( 11 \right)} \! \left( 1/2 \right)+47104 \, C_s^{\left(11 \right)} \! \left( 1/4 \right) +46575 C_s^{\left( 11 \right)} \! \left( 0 \right) \right].
\end{align}
These overlaps vanish for $s\in \{2, 4, 6, 8, 10, 14\}$, which is a consequence of the fact that the 196560 minimal vectors form a spherical $11 \frac{1}{2}$-design. The spinning zero-mode overlaps thus vanish if there are no nonzero cusp forms for $SL(2, \mathbb{Z})$ with weight $12+s$ that vanish with order at least $2$ at the cusp.

To prove the bound on $c_{11}(\sigma_2)$ in \eqref{eq:Leech_bound_2}, which is saturated by the Leech lattice \eqref{eq:Leech_c112}, we again take $\cutoff=42$ and consider $82$ spectral identities. To fix the functional, we needed to impose positivity at intermediate points between the double zeros. The resulting $\vec{\alpha}\in \mathbb{Q}^{82}$ is given in the ancillary file \texttt{alpha\_Leech\_2.txt}. With this $\vec{\alpha}$, we have
\be
\begin{split}
		\vec{\alpha}\cdot\vec{F}_{0}(3/2) &=1, \quad  \vec{\alpha}\cdot \vec A_0  =-\frac{541696}{4095}, \quad \vec{\alpha}\cdot \vec A_s \geqslant0, \quad  s=2,4, \dots, 20, \\
		\vec{\alpha}\cdot\vec{F}_s(x) &\geqslant 0, \quad  \forall x\geq 1, \quad  s =0, 2, 4, \dots, 20,
\end{split}
\ee
which proves \eqref{eq:Leech_bound_2} for tori.

\subsubsection{$D_4$ lattice}

The torus bound on $ c_{11}(\sigma_1)$ for tori in $d=4$ is
\be \label{eq:D4_bound}
c_{11}^2(\sigma_1) \leq \frac{8}{3} .
\ee
This bound is saturated by $\mathbb{R}^4/{\Lambda_4^*}$, where the laminated lattice $\Lambda_4$ is the $D_4 \cong D_4^*$ lattice, with $\phi_1$ given by \eqref{eq:phi_1}, since we have
\be \label{eq:D4_value}
\frac{\nu_1^2(D_{4})}{N_1(D_{4})}= \frac{8^2}{24} = \frac{8}{3}  .
\ee 

To prove \eqref{eq:D4_bound}, we take $\cutoff=30$ and consider $n_{\cutoff}=45$ spectral identities. To completely fix a functional, we needed to impose positivity at intermediate points between the double zeros. The resulting $\vec{\alpha} \in \mathbb{Q}^{45}$ is given in the ancillary file \texttt{alpha\_D4.txt}. Using this $\vec{\alpha}$, we can rigorously verify that the conditions in \eqref{eq:numconstralpha} are satisfied with $x_0=1$ and $\vec\alpha\cdot\vec{A}_0=-8/3$, which implies \eqref{eq:D4_bound}.

\subsubsection{$A_3$ lattice}

To prove the torus bound ${ c}^{(2)}_{11} (\sigma_1)\leq 1/\sqrt{3}$ in $d=3$, which is saturated by $\mathbb{R}^3/A_3^*$, we take $\cutoff= 26$ and consider $n_{\cutoff}=35$ spectral identities. In the ancillary file \texttt{alpha\_A3.txt}, we give a vector $\vec{\alpha} \in \mathbb{Q}^{35}$ that satisfies the conditions in \eqref{eq:spin2_conditions} with $\vec\alpha\cdot\vec{A}_0= -1/3$, which proves the bound.

\section{Conclusions}\label{sec:conclusion}

In this paper, we have used bootstrap methods to place upper bounds on the integrals of products of eigenfunctions of the Laplace--Beltrami operator on flat manifolds and orbifolds. We compared these numerical bounds to examples and constructed exact functionals to prove that some bounds are saturated by tori defined by special Euclidean lattices. 
As a consequence, we showed that the $A_2$, $D_4$, $E_8$, and Leech lattices have the largest values of $\nu_1^2(\Lambda)/N_1(\Lambda)$ amongst lattices $\Lambda$ in their respective dimensions, where $\nu_1(\Lambda)$ is the root mean square of the number of minimal vectors that are minimal distance from each minimal vector and $N_1(\Lambda)$ is the kissing number of $\Lambda$. 

We could find exact functionals with finitely many spectral identities because there are only finitely many non-vanishing overlaps for a given external state.  This is similar to other integrable examples in the bootstrap literature, such as conformal minimal models in two dimensions \cite{Belavin:1984vu} and integrable $S$-matrix theories in $(1+1)$ dimensions \cite{Zamolodchikov:1977nu,Paulos:2016but}. Minimal models are a subclass of rational CFTs, and thus their fusion rules have finitely many elements. Similarly, the $n$-particle amplitude for integrable $S$-matrix theories factorizes into a finite number of $2\to2$ scattering processes.

The lattices that saturate the bootstrap bounds also occur as optimal examples in other lattice optimization problems. What happens in higher dimensions? Natural candidates for saturating the bootstrap bounds would be even self-dual lattices with the largest minimal norm, which are called extremal. Following \cite{ElkiesLectureNotes}, extremal  even self-dual  lattices in $d=48$ have minimal norm 6 and
\be
 \frac{\nu_1^2(\Lambda)}{N_1(\Lambda)}=\frac{36848^2}{52416000} \approx 25.9,
\ee
whereas the bootstrap bound for $d=48$ at $\cutoff = 100$ only gives $c^2_{11}(\sigma_1) < 676$. 

There are several natural future directions that could be explored, such as generalizing the spectral identities to mixed correlators and considering a wider class of orbifold examples to find more saturating crystallographic groups. Other possible directions are adapting to flat orbifolds the representation theoretic derivation of the spectral identities \cite{Kravchuk:2021akc} and Adve's converse theorem \cite{Adve:2025sld}. 
It would also be interesting to see whether other problems can be usefully reformulated as bootstrap problems.  For example, the shell of the $E_8$ lattice of norm $2m$ is a spherical $8$-design if and only if $\tau(m)=0$, where $\tau(m)$ is Ramanujan's $\tau$-function \cite{Pache_designs}. There is a famous conjecture by Lehmer that $\tau(m) \neq 0$ for all $m \geq 1$. This conjecture is equivalent to the statement that, for every $m \geq 1$, we have $c_{ii}^{(8)}(0) \neq 0$ when $\phi_i$ is a uniform sum of cosines from the $m$\textsuperscript{th} nontrivial scalar eigenspace of $\mathbb{R}^8/E_8$. Another optimistic target is trying to combine the spectral identities with other constraints to find improved upper bounds on the Hermite constant in dimensions $10\leq d \leq 23$ and $d\geq 25$ (see \cite{9D_Hermite} for recent progress in $d=9$).

\paragraph{Acknowledgements:}

We would like to thank Anshul Adve, Noah Knutson, Jeremy Mann, Dalimil Maz\'{a}\v{c}, and Sridip Pal for helpful discussions. 
FB thanks the Yukawa Institute for Theoretical Physics and the organizers of the workshop ``Progress on Theoretical Bootstrap'' for hospitality during the completion of this work.

\renewcommand{\em}{}
\bibliographystyle{utphys}
\addcontentsline{toc}{section}{References}
\bibliography{hyperbolic-refs}

\providecommand{\href}[2]{#2}\begingroup\raggedright\begin{thebibliography}{10}

\bibitem{Belavin:1984vu}
A.~A. Belavin, A.~M. Polyakov, and A.~B. Zamolodchikov, ``{Infinite Conformal
  Symmetry in Two-Dimensional Quantum Field Theory},''
  \href{http://dx.doi.org/10.1016/0550-3213(84)90052-X}{{\em Nucl. Phys. B}
  {\bf 241} (1984)  333--380}.

\bibitem{Rattazzi:2008pe}
R.~Rattazzi, V.~S. Rychkov, E.~Tonni, and A.~Vichi, ``{Bounding scalar operator
  dimensions in 4D CFT},''
  \href{http://dx.doi.org/10.1088/1126-6708/2008/12/031}{{\em JHEP} {\bf 12}
  (2008)  031},
\href{http://arxiv.org/abs/0807.0004}{{\tt arXiv:0807.0004 [hep-th]}}.

\bibitem{Rychkov:2009ij}
V.~S. Rychkov and A.~Vichi, ``{Universal Constraints on Conformal Operator
  Dimensions},'' \href{http://dx.doi.org/10.1103/PhysRevD.80.045006}{{\em Phys.
  Rev. D} {\bf 80} (2009)  045006}, \href{http://arxiv.org/abs/0905.2211}{{\tt
  arXiv:0905.2211 [hep-th]}}.

\bibitem{Caracciolo:2009bx}
F.~Caracciolo and V.~S. Rychkov, ``{Rigorous Limits on the Interaction Strength
  in Quantum Field Theory},''
  \href{http://dx.doi.org/10.1103/PhysRevD.81.085037}{{\em Phys. Rev.} {\bf
  D81} (2010)  085037},
\href{http://arxiv.org/abs/0912.2726}{{\tt arXiv:0912.2726 [hep-th]}}.

\bibitem{Poland:2011ey}
D.~Poland, D.~Simmons-Duffin, and A.~Vichi, ``{Carving Out the Space of 4D
  CFTs},'' \href{http://dx.doi.org/10.1007/JHEP05(2012)110}{{\em JHEP} {\bf 05}
  (2012)  110},
\href{http://arxiv.org/abs/1109.5176}{{\tt arXiv:1109.5176 [hep-th]}}.

\bibitem{Poland:2018epd}
D.~Poland, S.~Rychkov, and A.~Vichi, ``{The Conformal Bootstrap: Theory,
  Numerical Techniques, and Applications},''
  \href{http://dx.doi.org/10.1103/RevModPhys.91.015002}{{\em Rev. Mod. Phys.}
  {\bf 91} (2019)  015002}, \href{http://arxiv.org/abs/1805.04405}{{\tt
  arXiv:1805.04405 [hep-th]}}.

\bibitem{ElShowk:2012ht}
S.~El-Showk, M.~F. Paulos, D.~Poland, S.~Rychkov, D.~Simmons-Duffin, and
  A.~Vichi, ``{Solving the 3D Ising Model with the Conformal Bootstrap},''
  \href{http://dx.doi.org/10.1103/PhysRevD.86.025022}{{\em Phys. Rev. D} {\bf
  86} (2012)  025022}, \href{http://arxiv.org/abs/1203.6064}{{\tt
  arXiv:1203.6064 [hep-th]}}.

\bibitem{El-Showk:2014dwa}
S.~El-Showk, M.~F. Paulos, D.~Poland, S.~Rychkov, D.~Simmons-Duffin, and
  A.~Vichi, ``{Solving the 3d Ising Model with the Conformal Bootstrap II.
  c-Minimization and Precise Critical Exponents},''
  \href{http://dx.doi.org/10.1007/s10955-014-1042-7}{{\em J. Stat. Phys.} {\bf
  157} (2014)  869}, \href{http://arxiv.org/abs/1403.4545}{{\tt arXiv:1403.4545
  [hep-th]}}.

\bibitem{Kos:2014bka}
F.~Kos, D.~Poland, and D.~Simmons-Duffin, ``{Bootstrapping Mixed Correlators in
  the 3D Ising Model},'' \href{http://dx.doi.org/10.1007/JHEP11(2014)109}{{\em
  JHEP} {\bf 11} (2014)  109},
\href{http://arxiv.org/abs/1406.4858}{{\tt arXiv:1406.4858 [hep-th]}}.

\bibitem{Kos:2016ysd}
F.~Kos, D.~Poland, D.~Simmons-Duffin, and A.~Vichi, ``{Precision Islands in the
  Ising and $O(N)$ Models},''
  \href{http://dx.doi.org/10.1007/JHEP08(2016)036}{{\em JHEP} {\bf 08} (2016)
  036}, \href{http://arxiv.org/abs/1603.04436}{{\tt arXiv:1603.04436
  [hep-th]}}.

\bibitem{Chang:2024whx}
C.-H. Chang, V.~Dommes, R.~S. Erramilli, A.~Homrich, P.~Kravchuk, A.~Liu, M.~S.
  Mitchell, D.~Poland, and D.~Simmons-Duffin, ``{Bootstrapping the 3d Ising
  stress tensor},'' \href{http://dx.doi.org/10.1007/JHEP03(2025)136}{{\em JHEP}
  {\bf 03} (2025)  136}, \href{http://arxiv.org/abs/2411.15300}{{\tt
  arXiv:2411.15300 [hep-th]}}.

\bibitem{Bonifacio:2020xoc}
J.~Bonifacio and K.~Hinterbichler, ``{Bootstrap Bounds on Closed Einstein
  Manifolds},'' \href{http://dx.doi.org/10.1007/JHEP10(2020)069}{{\em JHEP}
  {\bf 10} (2020)  069}, \href{http://arxiv.org/abs/2007.10337}{{\tt
  arXiv:2007.10337 [hep-th]}}.

\bibitem{Bonifacio:2021msa}
J.~Bonifacio, ``{Bootstrap Bounds on Closed Hyperbolic Manifolds},''
  \href{http://dx.doi.org/10.1007/JHEP02(2022)025}{{\em JHEP} {\bf 02} (2022)
  025}, \href{http://arxiv.org/abs/2107.09674}{{\tt arXiv:2107.09674
  [hep-th]}}.

\bibitem{Bonifacio:2021aqf}
J.~Bonifacio, ``{Bootstrapping closed hyperbolic surfaces},''
  \href{http://dx.doi.org/10.1007/JHEP03(2022)093}{{\em JHEP} {\bf 03} (2022)
  093}, \href{http://arxiv.org/abs/2111.13215}{{\tt arXiv:2111.13215
  [hep-th]}}.

\bibitem{Kravchuk:2021akc}
P.~Kravchuk, D.~Maz\'a\v{c}, and S.~Pal, ``{Automorphic spectra and the
  conformal bootstrap},'' \href{http://dx.doi.org/10.1090/cams/26}{{\em Commun.
  Am. Math. Soc.} {\bf 4} (2024) no.~1, 1--63},
  \href{http://arxiv.org/abs/2111.12716}{{\tt arXiv:2111.12716 [hep-th]}}.

\bibitem{Radcliffe:2024jcg}
A.~Radcliffe, ``{Non-saturation of bootstrap bounds by hyperbolic orbifolds},''
  \href{http://dx.doi.org/10.1007/JHEP12(2025)115}{{\em JHEP} {\bf 12} (2025)
  115}, \href{http://arxiv.org/abs/2404.14479}{{\tt arXiv:2404.14479
  [hep-th]}}.

\bibitem{Bonifacio:2023ban}
J.~Bonifacio, D.~Maz\'a\v{c}, and S.~Pal, ``{Spectral Bounds on Hyperbolic
  3-Manifolds: Associativity and the Trace Formula},''
  \href{http://dx.doi.org/10.1007/s00220-024-05222-0}{{\em Commun. Math. Phys.}
  {\bf 406} (2025)  51}, \href{http://arxiv.org/abs/2308.11174}{{\tt
  arXiv:2308.11174 [math.SP]}}.

\bibitem{Gesteau:2023brw}
E.~Gesteau, S.~Pal, D.~Simmons-Duffin, and Y.~Xu, ``{Bounds on spectral gaps of
  Hyperbolic spin surfaces},''
  \href{http://dx.doi.org/10.56994/JAMR.003.001.003}{{\em J. Assoc. Math. Res.}
  {\bf 3} (2025) no.~1, 72--139}, \href{http://arxiv.org/abs/2311.13330}{{\tt
  arXiv:2311.13330 [math.SP]}}.

\bibitem{Adve:2025rvf}
A.~Adve, J.~Bonifacio, P.~Kravchuk, D.~Mazac, S.~Pal, A.~Radcliffe, and
  G.~Rogelberg, ``{Weyl bound for trilinear periods via conformal bootstrap},''
  \href{http://arxiv.org/abs/2508.20576}{{\tt arXiv:2508.20576 [math.NT]}}.

\bibitem{Adve:2025sld}
A.~Adve, ``{A converse theorem for hyperbolic surface spectra and the conformal
  bootstrap},'' \href{http://arxiv.org/abs/2509.17935}{{\tt arXiv:2509.17935
  [math.SP]}}.

\bibitem{monk2026spectral}
L.~Monk and F.~Naud, ``Spectral gaps on large hyperbolic surfaces,''
  \href{http://arxiv.org/abs/2601.13988}{{\tt arXiv:2601.13988 [math.SP]}}.

\bibitem{bourque2019kissing}
M.~F. Bourque and B.~Petri, ``Kissing numbers of closed hyperbolic manifolds,''
  {\em American Journal of Mathematics} {\bf 144} (2022) no.~4, 1067--1085,
  \href{http://arxiv.org/abs/1905.11083}{{\tt arXiv:1905.11083 [math.GT]}}.

\bibitem{Bourque:2023woe}
M.~F. Bourque and B.~Petri, ``{Linear programming bounds for hyperbolic
  surfaces},'' \href{http://arxiv.org/abs/2302.02540}{{\tt arXiv:2302.02540
  [math.GT]}}.

\bibitem{cohn2003new}
H.~Cohn and N.~Elkies, ``{New upper bounds on sphere packings I},'' {\em Annals
  of mathematics} (2003)  689--714,
  \href{http://arxiv.org/abs/math/0110009}{{\tt arXiv:math/0110009 [math.MG]}}.

\bibitem{MR314545}
P.~Delsarte, ``Bounds for unrestricted codes, by linear programming,'' {\em
  Philips Res. Rep.} {\bf 27} (1972)  272--289.

\bibitem{Blichfeldt1935}
H.~Blichfeldt, ``The minimum values of positive quadratic forms in six, seven
  and eight variables,'' {\em Mathematische Zeitschrift} {\bf 39} (1935)
  1--15.

\bibitem{MR2600869}
H.~Cohn and A.~Kumar, ``Optimality and uniqueness of the {L}eech lattice among
  lattices,'' \href{http://dx.doi.org/10.4007/annals.2009.170.1003}{{\em Ann.
  of Math. (2)} {\bf 170} (2009) no.~3, 1003--1050}.

\bibitem{viazovska2017sphere}
M.~S. Viazovska, ``{The sphere packing problem in dimension 8},''
  \href{http://dx.doi.org/10.4007/annals.2017.185.3.7}{{\em Annals of
  mathematics} (2017)  991--1015}, \href{http://arxiv.org/abs/1603.04246}{{\tt
  arXiv:1603.04246 [math.NT]}}.

\bibitem{cohn2017sphere}
H.~Cohn, A.~Kumar, S.~Miller, D.~Radchenko, and M.~Viazovska, ``{The sphere
  packing problem in dimension 24},''
  \href{http://dx.doi.org/10.4007/annals.2017.185.3.8}{{\em Annals of
  mathematics} {\bf 185} (2017) no.~3, 1017--1033},
  \href{http://arxiv.org/abs/1603.06518}{{\tt arXiv:1603.06518 [math.NT]}}.

\bibitem{Conway:1988oqe}
J.~H. Conway and N.~J.~A. Sloane,
  \href{http://dx.doi.org/10.1007/978-1-4757-6568-7}{{\em {Sphere Packings,
  Lattices and Groups}}}.
\newblock Springer, 2013.

\bibitem{Hartman:2019pcd}
T.~Hartman, D.~Maz\'a\v{c}, and L.~Rastelli, ``{Sphere Packing and Quantum
  Gravity},'' \href{http://dx.doi.org/10.1007/JHEP12(2019)048}{{\em JHEP} {\bf
  12} (2019)  048}, \href{http://arxiv.org/abs/1905.01319}{{\tt
  arXiv:1905.01319 [hep-th]}}.

\bibitem{Afkhami-Jeddi:2020hde}
N.~Afkhami-Jeddi, H.~Cohn, T.~Hartman, D.~de~Laat, and A.~Tajdini,
  ``{High-dimensional sphere packing and the modular bootstrap},''
  \href{http://dx.doi.org/10.1007/JHEP12(2020)066}{{\em JHEP} {\bf 12} (2020)
  066}, \href{http://arxiv.org/abs/2006.02560}{{\tt arXiv:2006.02560
  [hep-th]}}.

\bibitem{Mazac:2016qev}
D.~Maz\'a\v{c}, ``{Analytic bounds and emergence of AdS$_{2}$ physics from the
  conformal bootstrap},'' \href{http://dx.doi.org/10.1007/JHEP04(2017)146}{{\em
  JHEP} {\bf 04} (2017)  146}, \href{http://arxiv.org/abs/1611.10060}{{\tt
  arXiv:1611.10060 [hep-th]}}.

\bibitem{Dolan:1989vr}
L.~Dolan, P.~Goddard, and P.~Montague, ``{Conformal Field Theory of Twisted
  Vertex Operators},''
  \href{http://dx.doi.org/10.1016/0550-3213(90)90644-S}{{\em Nucl. Phys. B}
  {\bf 338} (1990)  529--601}.

\bibitem{Dolan:1989kf}
L.~Dolan, P.~Goddard, and P.~Montague, ``{Conformal Field Theory, Triality and
  the Monster Group},''
  \href{http://dx.doi.org/10.1016/0370-2693(90)90821-M}{{\em Phys. Lett. B}
  {\bf 236} (1990)  165--172}.

\bibitem{Dolan:1994st}
L.~Dolan, P.~Goddard, and P.~Montague, ``{Conformal field theories,
  representations and lattice constructions},''
  \href{http://dx.doi.org/10.1007/BF02103716}{{\em Commun. Math. Phys.} {\bf
  179} (1996)  61--120}.

\bibitem{Frenkel:1988xz}
I.~Frenkel, J.~Lepowsky, and A.~Meurman, {\em Vertex operator algebras and the
  Monster}, vol.~134 of {\em Pure and Applied Mathematics}.
\newblock Academic Press, United States, 1988.

\bibitem{Dymarsky:2020qom}
A.~Dymarsky and A.~Shapere, ``{Quantum stabilizer codes, lattices, and CFTs},''
  \href{http://dx.doi.org/10.1007/JHEP03(2021)160}{{\em JHEP} {\bf 21} (2020)
  160}, \href{http://arxiv.org/abs/2009.01244}{{\tt arXiv:2009.01244
  [hep-th]}}.

\bibitem{Henriksson:2021qkt}
J.~Henriksson, A.~Kakkar, and B.~McPeak, ``{Classical codes and chiral CFTs at
  higher genus},'' \href{http://dx.doi.org/10.1007/JHEP05(2022)159}{{\em JHEP}
  {\bf 05} (2022)  159}, \href{http://arxiv.org/abs/2112.05168}{{\tt
  arXiv:2112.05168 [hep-th]}}.

\bibitem{Sarnak94}
P.~Sarnak, ``{Integrals of products of eigenfunctions},''
  \href{http://dx.doi.org/10.1155/S1073792894000280}{{\em International
  Mathematics Research Notices} {\bf 1994} (1994) no.~6, 251--260}.

\bibitem{bernstein2010subconvexity}
J.~Bernstein and A.~Reznikov, ``Subconvexity bounds for triple l-functions and
  representation theory,'' {\em Annals of mathematics} (2010)  1679--1718.

\bibitem{LU20193271}
J.~Lu, C.~D. Sogge, and S.~Steinerberger, ``Approximating pointwise products of
  laplacian eigenfunctions,''
  \href{http://dx.doi.org/https://doi.org/10.1016/j.jfa.2019.05.025}{{\em
  Journal of Functional Analysis} {\bf 277} (2019) no.~9, 3271--3282}.

\bibitem{WYMAN2022109404}
E.~L. Wyman, ``Triangles and triple products of laplace eigenfunctions,''
  \href{http://dx.doi.org/https://doi.org/10.1016/j.jfa.2022.109404}{{\em
  Journal of Functional Analysis} {\bf 282} (2022) no.~8, 109404},
  \href{http://arxiv.org/abs/2103.03336}{{\tt arXiv:2103.03336 [math.AP]}}.

\bibitem{Hinterbichler:2013kwa}
K.~Hinterbichler, J.~Levin, and C.~Zukowski, ``{Kaluza-Klein Towers on General
  Manifolds},'' \href{http://dx.doi.org/10.1103/PhysRevD.89.086007}{{\em Phys.
  Rev. D} {\bf 89} (2014) no.~8, 086007},
  \href{http://arxiv.org/abs/1310.6353}{{\tt arXiv:1310.6353 [hep-th]}}.

\bibitem{Bonifacio:2019ioc}
J.~Bonifacio and K.~Hinterbichler, ``{Unitarization from Geometry},''
  \href{http://dx.doi.org/10.1007/JHEP12(2019)165}{{\em JHEP} {\bf 12} (2019)
  165}, \href{http://arxiv.org/abs/1910.04767}{{\tt arXiv:1910.04767
  [hep-th]}}.

\bibitem{SekharChivukula:2019qih}
R.~Sekhar~Chivukula, D.~Foren, K.~A. Mohan, D.~Sengupta, and E.~H. Simmons,
  ``{Sum Rules for Massive Spin-2 Kaluza-Klein Elastic Scattering
  Amplitudes},'' \href{http://dx.doi.org/10.1103/PhysRevD.100.115033}{{\em
  Phys. Rev. D} {\bf 100} (2019) no.~11, 115033},
  \href{http://arxiv.org/abs/1910.06159}{{\tt arXiv:1910.06159 [hep-ph]}}.

\bibitem{LatticeCatalog}
G.~Nebe and N.~Sloane, ``A catalogue of lattices.''
\newblock \url{https://www.math.rwth-aachen.de/~Gabriele.Nebe/LATTICES/}.

\bibitem{Musin_D4_Kissing}
O.~R. {Musin}, ``{The problem of the twenty-five spheres},''
  \href{http://dx.doi.org/10.1070/RM2003v058n04ABEH000651}{{\em Russian
  Mathematical Surveys} {\bf 58} (2003) no.~4, 794--795}.

\bibitem{Levenstein1979}
V.~I. Leven\v{s}te\u{\i}n, ``Boundaries for packings in {$n$}-dimensional
  {E}uclidean space,'' {\em Dokl. Akad. Nauk SSSR} {\bf 245} (1979) no.~6,
  1299--1303.

\bibitem{ODLYZKO1979210}
A.~Odlyzko and N.~Sloane, ``New bounds on the number of unit spheres that can
  touch a unit sphere in n dimensions,''
  \href{http://dx.doi.org/https://doi.org/10.1016/0097-3165(79)90074-8}{{\em
  Journal of Combinatorial Theory, Series A} {\bf 26} (1979) no.~2, 210--214}.

\bibitem{laminatedLattices}
J.~H. Conway and N.~J.~A. Sloane, ``Laminated lattices,'' {\em Annals of
  Mathematics} {\bf 116} (1982) no.~3, 593--620.

\bibitem{Delsarte_Goethals_Seidel}
P.~Delsarte, J.~M. Goethals, and J.~J. Seidel, ``Spherical codes and designs,''
  \href{http://dx.doi.org/10.1007/BF03187604}{{\em Geometriae Dedicata} {\bf 6}
  (1977) no.~3, 363--388}.

\bibitem{MR752931}
B.~B. Venkov, ``Even unimodular extremal lattices,'' vol.~165, pp.~43--48.
\newblock 1984.
\newblock Algebraic geometry and its applications.

\bibitem{Pache_designs}
C.~Pache, ``Shells of selfdual lattices viewed as spherical designs,''
  \href{http://dx.doi.org/10.1142/S0218196705002797}{{\em International Journal
  of Algebra and Computation} {\bf 15} (2005) no.~05n06, 1085--1127}.

\bibitem{thurstonBook}
W.~P. Thurston, {\em Three-Dimensional Geometry and Topology, Volume 1}.
\newblock Princeton University Press, 1997.

\bibitem{Bettiol2018}
R.~G. Bettiol, A.~Derdzinski, and P.~Piccione, ``{Teichm{\"u}ller theory and
  collapse of flat manifolds},''
  \href{http://dx.doi.org/10.1007/s10231-017-0723-7}{{\em Annali di Matematica
  Pura ed Applicata (1923 -)} {\bf 197} (2018) no.~4, 1247--1268},
  \href{http://arxiv.org/abs/1705.08431}{{\tt arXiv:1705.08431 [math.DG]}}.

\bibitem{Charlap:1986bgfm}
L.~S. Charlap, \href{http://dx.doi.org/10.1007/978-1-4613-8687-2}{{\em
  {Bieberbach Groups and Flat Manifolds}}}.
\newblock Springer, 1986.

\bibitem{wolf2011spaces}
J.~A. Wolf, {\em Spaces of Constant Curvature}.
\newblock AMS Chelsea Pub., 2011.

\bibitem{10.1063/1.2735211}
L.~Nilse, ``Classification of 1d and 2d orbifolds,''
  \href{http://dx.doi.org/10.1063/1.2735211}{{\em AIP Conference Proceedings}
  {\bf 903} (2007) no.~1, 411--414},
  \href{http://arxiv.org/abs/hep-ph/0601015}{{\tt arXiv:hep-ph/0601015
  [hep-ph]}}.

\bibitem{kaye_orbifold}
A.~Kaye, ``Two-dimensional orbifolds,'' University of Chicago, VIGRE REU 2007.
\newblock
  \url{https://www.math.uchicago.edu/~may/VIGRE/VIGRE2007/REUPapers/FINALFULL/Kaye.pdf}.

\bibitem{Peng:2019xoj}
Z.-P. Peng, L.~Lindblom, and F.~Zhang, ``{Scalar, Vector and Tensor Harmonics
  on the Flat Compact Orientable Three-Manifolds},''
  \href{http://dx.doi.org/10.1088/1475-7516/2019/12/042}{{\em JCAP} {\bf 12}
  (2019)  042}, \href{http://arxiv.org/abs/1909.06721}{{\tt arXiv:1909.06721
  [gr-qc]}}.

\bibitem{COMPACT:2023rkp}
{\bf COMPACT} Collaboration, J.~R. Eskilt {\em et al.}, ``{Cosmic topology.
  Part~IIa. Eigenmodes, correlation matrices, and detectability of
  orientable~Euclidean manifolds},''
  \href{http://dx.doi.org/10.1088/1475-7516/2024/03/036}{{\em JCAP} {\bf 03}
  (2024)  036}, \href{http://arxiv.org/abs/2306.17112}{{\tt arXiv:2306.17112
  [astro-ph.CO]}}.

\bibitem{Taronna:2010qq}
M.~Taronna, ``{Higher Spins and String Interactions},''
  \href{http://arxiv.org/abs/1005.3061}{{\tt arXiv:1005.3061 [hep-th]}}.

\bibitem{Dolan:2011dv}
F.~A. Dolan and H.~Osborn, ``{Conformal Partial Waves: Further Mathematical
  Results},'' \href{http://arxiv.org/abs/1108.6194}{{\tt arXiv:1108.6194
  [hep-th]}}.

\bibitem{Costa_2016}
M.~S. Costa, T.~Hansen, J.~Penedones, and E.~Trevisani, ``{Projectors and seed
  conformal blocks for traceless mixed-symmetry tensors},''
  \href{http://dx.doi.org/10.1007/jhep07(2016)018}{{\em JHEP} {\bf 2016} (2016)
   018}, \href{http://arxiv.org/abs/1603.05551}{{\tt arXiv:1603.05551
  [hep-th]}}.

\bibitem{Csaki:2003dt}
C.~Csaki, C.~Grojean, H.~Murayama, L.~Pilo, and J.~Terning, ``{Gauge theories
  on an interval: Unitarity without a Higgs},''
  \href{http://dx.doi.org/10.1103/PhysRevD.69.055006}{{\em Phys. Rev. D} {\bf
  69} (2004)  055006}, \href{http://arxiv.org/abs/hep-ph/0305237}{{\tt
  arXiv:hep-ph/0305237}}.

\bibitem{Simmons-Duffin:2015qma}
D.~Simmons-Duffin, ``{A Semidefinite Program Solver for the Conformal
  Bootstrap},'' \href{http://dx.doi.org/10.1007/JHEP06(2015)174}{{\em JHEP}
  {\bf 06} (2015)  174}, \href{http://arxiv.org/abs/1502.02033}{{\tt
  arXiv:1502.02033 [hep-th]}}.

\bibitem{Landry:2019qug}
W.~Landry and D.~Simmons-Duffin, ``{Scaling the semidefinite program solver
  SDPB},'' \href{http://arxiv.org/abs/1909.09745}{{\tt arXiv:1909.09745
  [hep-th]}}.

\bibitem{NESTEROV}
Y.~Nesterov, ``Random walk in a simplex and quadratic optimization over convex
  polytopes,''. \url{https://ideas.repec.org/p/cor/louvco/2003071.html}.

\bibitem{ElkiesLectureNotes}
N.~Elkies. \url{https://people.math.harvard.edu/~elkies/M272.19/nov04.pdf},
  2019.

\bibitem{MR1484478}
W.~Bosma, J.~Cannon, and C.~Playoust, ``The {M}agma algebra system. {I}. {T}he
  user language,'' \href{http://dx.doi.org/10.1006/jsco.1996.0125}{{\em J.
  Symbolic Comput.} {\bf 24} (1997) no.~3-4, 235--265}.

\bibitem{Zamolodchikov:1977nu}
A.~B. Zamolodchikov and A.~B. Zamolodchikov, ``{Relativistic Factorized
  S-Matrix in Two-Dimensions Having O(N) Isotopic Symmetry},''
  \href{http://dx.doi.org/10.1016/0550-3213(78)90239-0}{{\em JETP Lett.} {\bf
  26} (1977)  457}.

\bibitem{Paulos:2016but}
M.~F. Paulos, J.~Penedones, J.~Toledo, B.~C. van Rees, and P.~Vieira, ``{The
  S-matrix bootstrap II: two dimensional amplitudes},''
  \href{http://dx.doi.org/10.1007/JHEP11(2017)143}{{\em JHEP} {\bf 11} (2017)
  143}, \href{http://arxiv.org/abs/1607.06110}{{\tt arXiv:1607.06110
  [hep-th]}}.

\bibitem{9D_Hermite}
M.~{Dutour Sikiri{\'c}} and W.~{van Woerden}, ``{The lattice packing problem in
  dimension 9 by Voronoi's algorithm},''
  \href{http://arxiv.org/abs/2508.20719}{{\tt arXiv:2508.20719 [math.NT]}}.

\end{thebibliography}\endgroup

\end{document}